\documentclass[12pt,letterpaper]{article}
\usepackage[margin=1in]{geometry}
\usepackage[T1]{fontenc}
\usepackage[utf8]{inputenc}
\usepackage{textcomp,lmodern}
\usepackage{microtype}
\UseMicrotypeSet[protrusion]{basicmath}
\usepackage{amsmath,amssymb,amsthm}
\usepackage{graphicx,array,booktabs,multirow,caption}
\usepackage[section]{placeins}
\usepackage{etoolbox}
\usepackage{natbib}
\usepackage[dvipsnames]{xcolor}
\usepackage[unicode,hypertexnames=false,colorlinks=true,linkcolor=blue,citecolor=Blue,urlcolor=Blue]{hyperref}
\usepackage{bookmark}
\hypersetup{
  pdftitle={Identification and Estimation of Semiparametric Multilayered Sample Selection Models},
  pdfauthor={Dongwoo Kim},
  pdfkeywords={Control function, Gender wage gap, Multinomial choice, Occupational sorting, Sieve approximation}}
\bibpunct{(}{)}{;}{a}{}{,}
\providecommand{\harvardand}{and}

\newcommand{\spacingset}[1]{\renewcommand{\baselinestretch}{#1}\small\normalsize}

\AtBeginEnvironment{table}{\spacingset{1}}
\AtBeginEnvironment{figure}{\spacingset{1}}
\let\articleparagraph\paragraph
\renewcommand{\paragraph}[1]{\articleparagraph{#1}\mbox{}}
\DeclareMathOperator*{\argmax}{arg\,max}
\DeclareMathOperator*{\argmin}{arg\,min}
\newcommand{\ci}{\perp\!\!\!\perp}

\numberwithin{equation}{section}
\theoremstyle{plain}
\newtheorem{prop}{Proposition}
\newtheorem{corollary}{Corollary}
\newtheorem{assumption}{Assumption}
\newtheorem{theorem}{Theorem}
\newtheorem{lemma}{Lemma}
\theoremstyle{remark}
\newtheorem{remark}{Remark}

\title{\bf Identification and Estimation of Semiparametric Multilayered Sample Selection Models}
  \author{Dongwoo Kim\thanks{
    Email: dongwook@sfu.ca. The author gratefully acknowledges support from the Social Sciences and Humanities Research Council of Canada under the Insight Grant (435-2024-0322).}\hspace{0.2cm}\\
    Simon Fraser University and Korea University}

\begin{document}
\spacingset{1}
\maketitle
\bigskip
\begin{abstract}
Many selection problems are multilayered, involving both participation and sorting among ordered or unordered categories. Unlike binary selection, in which selection bias can be summarized by a scalar control function, ordered and multinomial sorting generally produce multi-index control functions for which identification depends on both covariate variation and the structure of the selection correction function. I establish identification results for both architectures, showing how nonlinearity in the selection equation can substitute for excluded variables under structural restrictions. I also show how additional structural restrictions reduce the control-function dimension and make estimation practical. I propose $\sqrt{n}$-consistent two-step sieve plug-in estimators and apply the framework to gender wage gaps among Korean college graduates. Under the ordered specification, selection correction substantially narrows the estimated large-firm wage gap.
\end{abstract}

\noindent%
\vfill

\newpage
\spacingset{1.8}

\section{Introduction}\label{sec: introduction}

\citet{heckman1974shadow, heckman1979sample} pioneered the econometric analysis of sample selection by modeling the joint determination of participation and outcomes. Yet Heckman's framework treats selection as a binary event: an individual either participates or not. In many empirical settings, the selection structure is richer. For instance, workers differ both in whether they enter wage employment and in the occupations, firms, or industries in which they work. The outcome of interest is shaped by \emph{two layers} of selection: participation and sorting among alternatives. Ignoring the sorting layer conflates within-occupation wage effects with between-occupation composition effects, potentially distorting policy implications.

This paper develops semiparametric models for multilayered selection and establishes conditions for identification. The selection process can be either parametric, semiparametric, or nonparametric, while the conditional mean of each potential outcome is linear in covariates, with category-specific coefficients. I consider two selection architectures. Vertical sorting involves ordered categories and is modeled through threshold crossings. Horizontal sorting involves unordered categories, such as occupations or sectors, and is modeled through multinomial choice. Under the maintained restrictions on outcome and selection disturbances, the mean outcome conditional on selection combines the linear covariate component with a selection correction function. I characterize this function, establish identification of the outcome coefficients, and propose tractable sieve plug-in estimators.

The central challenge is that selection bias generally depends on multiple indices, which increases requirements for identification and estimation relative to binary selection. Specifically, when selection bias depends on $q$ indices, $q+1$ continuous covariates suffice to identify the outcome coefficients under rank and support conditions. For ordered selection, I show that within each category the same univariate function, averaged over each interval between the selection thresholds, determines the selection bias. This restriction can identify the outcome coefficients with fewer continuous covariates. For horizontal sorting with $K+1$ categories, the selection bias depends on $K$ choice probabilities. Under multinomial logit selection, it reduces to a function of the chosen category's probability when the outcome disturbance depends in mean only on that category's own taste shock. For more general multinomial selection under exchangeability, the selection bias is symmetric in the choice probabilities of the non-chosen categories. I therefore approximate it using a few low-order symmetric polynomials of these probabilities. This reduces the dimension of estimation and can improve accuracy when continuous covariate variation is limited.

I propose two-step sieve estimators that first estimate choice probabilities and then include spline approximations of the estimated probabilities in the outcome regression to control for selection bias. The asymptotic theory shows that the estimator is $\sqrt{n}$-consistent and asymptotically normal, allowing for growing spline bases and data-selected smoothing. Estimating the choice probabilities generally contributes to first-order uncertainty. Analytic standard errors that account for this first-stage influence are provided.  The simulations demonstrate gains in precision from imposing the ordered model's restriction on the selection bias and improved accuracy from the symmetric polynomial approximation.

I apply the proposed framework to gender wage gaps among recent college graduates in South Korea using the Graduates Occupational Mobility Survey (GOMS). I consider an ordered model for firm-size sorting and multinomial models for field (science, technology, engineering, and mathematics, or STEM, versus non-STEM) and sector (public versus private) sorting. Selection correction substantially narrows the large-firm gap, while negative gaps persist in small and medium-sized enterprises (SMEs), non-STEM occupations, and the private sector. For field and sector sorting, the selection-corrected estimates remain close to OLS. In separate decompositions of the raw hourly wage gap, the contribution of participation and sorting is about one-third for firm size, one-tenth for field, and slightly negative for sector. Across graduating cohorts, the gap in STEM jobs narrows markedly, while gaps in SMEs and the private sector persist. \citet{kim2026gender} finds that correcting for employment selection alone has little effect on the entry-level gap in the same survey. Once sorting is modeled together with participation, this does not hold for firm-size sorting. The application contributes to the gender wage gap literature that recognizes employment selection and occupational sorting as pivotal to measuring gender disparities.\footnote{See, for example, \citet{olivetti2008unequal, blau2017gender, blau2024selection, kim2026gender}.}

This paper connects to the semiparametric and nonparametric selection literature \citep{chamberlain1986asymptotic, ahn1993semiparametric, newey1990semiparametric, newey2009two, das2003nonparametric}. Recent work studies binary selection without exclusions \citep{honore2020selection, honore2024heterogeneity}. An alternative approach uses a special regressor \citep{lewbel2007endogenous}. \citet{escanciano2016identification} establish identification results for semiparametric two-index models, \citet{kim2025point} study identification and estimation of semiparametric selection models without exclusion restrictions, and \citet{pan2024locally} develop a related locally robust estimator with a machine-learning first stage.

For multinomial selection, I build on earlier parametric and semiparametric correction methods \citep{lee1983generalized, dubin1984econometric, dahl2002mobility, bourguignon2007selection} and established choice-probability inversion results \citep{hotz1993conditional, berry2013connected, norets2013surjectivity}. The symmetric approximation adapts the exchangeability device of \citet{sheng2025social}. Related multivalued-treatment methods identify heterogeneous responses using continuous excluded instruments \citep{lee2018multivalued}. The Roy tradition provides a foundation for sectoral sorting \citep{roy1951some, heckman1990varieties, french2011identification}. In a binary extended Roy model, \citet{dhaultfoeuille2013inference} identify the nonpecuniary component without exclusions, while identification of sector-specific earnings effects requires exclusions or identification at infinity. Most directly related is \citet{kroft2024lee}, who introduce multilayered selection and extend the nonparametric bounds of \citet{lee2009bounds} under conditional selection monotonicity.

Section~\ref{sec: model} presents the models and identification results. Section~\ref{sec: estimator} develops estimation and inference. Sections~\ref{sec: simulations} and~\ref{sec: application} report the simulations and the application, and Section~\ref{sec: conclusion} concludes. Technical proofs and additional simulation and empirical results are in the Supplementary Appendix.

\section{The Multilayered Selection Models}\label{sec: model}

Consider a population of individuals indexed by $i$, each characterized by observable covariates $X_i \in \mathcal{X} \subseteq \mathbb{R}^{d_X}$, a discrete selection indicator $D_i \in \{0, 1, \ldots, K\}$, and potential outcomes $\{Y_{ik}^*\}_{k=1}^K$. $D_i = 0$ indicates non-participation (e.g., unemployment), while $D_i = k$ for $k \geq 1$ indicates participation in category $k$. The potential outcomes for each $k$ and the observed outcome $Y_i$ are determined by 
\begin{equation}\label{eq: potential outcomes}
Y_{ik}^* = \alpha_k + X_i'\beta_k + V_{ik}, \quad Y_i = \sum_{k=1}^K \mathbf{1}[D_i = k] \cdot Y_{ik}^*,
\end{equation}
where $\beta_k \in \mathbb{R}^{d_X}$ is a category-specific parameter vector, $\alpha_k$ is a category-specific intercept, $V_{ik}$ is unobserved heterogeneity with $E[V_{ik}|X_i] = 0$, and  $\mathbf{1}[\cdot]$ is the indicator function. Conditional on selection into category $k$, the expected observed outcome is
\begin{equation}\label{eq: conditional expectation}
E[Y_i | X_i = x, D_i = k] = \alpha_k + x'\beta_k + \underbrace{E[V_{ik} | X_i = x, D_i = k]}_{\text{selection bias}},
\end{equation}
where the selection bias captures the systematic difference in unobservable characteristics between individuals who select into category $k$ and the population average. 

The key modeling challenge is to specify the selection process generating $D_i$ in a way that (i) allows the selection bias to be characterized by a tractable control function, (ii) permits identification of $\beta_k$ without exclusion restrictions, and (iii) accommodates the economic structure of sorting. I consider two selection models in turn: vertical and horizontal sorting.

\subsection{Vertical sorting: ordered selection}\label{sec: ordered}

Suppose the categories are vertically differentiated. Then selection into categories can be modeled as a threshold crossing structure: 
\begin{equation}\label{eq: ordered probit}
D_i=k\quad\text{if}\quad c_k\leq X_i'\gamma+\varepsilon_i<c_{k+1}, \quad -\infty=c_0<c_1<\cdots<c_K<c_{K+1}=\infty.
\end{equation}
If $\varepsilon_i\sim N(0,1)$, this is an ordered probit model. Define $\operatorname{Corr}(V_{ik},\varepsilon_i)=\rho_k$, and $\operatorname{Var}(V_{ik})=\sigma_k^2.$ Under joint normality of $(V_{ik},\varepsilon_i)\ci X_i$, the selection bias is
\begin{equation}\label{eq: parametric bias}
\begin{aligned}
E[V_{ik}\mid X_i=x,D_i=k]
&=\sigma_k\rho_k\frac{\phi(c_k-x'\gamma)-\phi(c_{k+1}-x'\gamma)}
{\Phi(c_{k+1}-x'\gamma)-\Phi(c_k-x'\gamma)}.
\end{aligned}
\end{equation}
The ordered probit inverse Mills ratio can be nearly collinear with $x'\gamma$ in a wide range of the support (see Figure~\ref{fig: ordered linearity} in Appendix~\ref{sec: appendix cf figures}).

Now consider a semiparametric model by replacing $X_i'\gamma$ with a nonparametric index $h(X_i)$ in \eqref{eq: ordered probit}. Maintain $\varepsilon_i\ci X_i$ and $E[V_{ik}\mid\varepsilon_i=e,X_i]=a_k(e)$. These restrictions permit conditional heteroskedasticity in $V_{ik}$. Suppose $F_\varepsilon$ is known, continuous, and strictly increasing, and normalize the location by $c_1=0$. The selection index is identified from a cumulative probability: $h(x)=-F_\varepsilon^{-1}(P[D_i=0\mid X_i=x])$. The selection bias is
\begin{equation}\label{eq: semiparametric bias ordered}
E[V_{ik}\mid X_i=x,D_i=k]
=\frac{\int_{c_k-h(x)}^{c_{k+1}-h(x)}a_k(e)\,dF_\varepsilon(e)}
{F_\varepsilon(c_{k+1}-h(x))-F_\varepsilon(c_k-h(x))}
=:b_k(h(x)).
\end{equation}
Consequently, $E[Y_i\mid X_i=x,D_i=k]=x'\beta_k+\lambda_k(h(x))$ where $\lambda_k=\alpha_k+b_k$. The identification argument of \citet{kim2025point} directly identifies $\beta_k$ when $h$ satisfies their nonlinearity and support conditions, which can hold with a single continuously distributed covariate.

The common-index model restricts all thresholds to shift through $h(x)$. To allow threshold responses to differ across categories, consider
\begin{equation}\label{eq: nonparametric ordered}
D_i=k\quad\text{if}\quad h_k(X_i)\leq U_i<h_{k+1}(X_i),
\end{equation}
where $U_i\sim\operatorname{Unif}(0,1)$, $U_i\ci X_i$,
$h_0(x)=0$, $h_{K+1}(x)=1$, and
$0<h_1(x)<\cdots<h_K(x)<1$. Maintain the conditional mean restriction
$E[V_{ik}\mid U_i=u,X_i]=a_k(u)$. Writing $p_j(x):=P[D_i=j\mid X_i=x]$, the thresholds are identified by
\begin{equation}\label{eq: threshold identification}
h_k(x)=\sum_{j=0}^{k-1}p_j(x)
=P[D_i\leq k-1\mid X_i=x],\qquad k=1,\ldots,K.
\end{equation}
The selection bias is
\begin{equation}\label{eq: nonparametric bias ordered}
E[V_{ik}\mid X_i=x,D_i=k]
=\frac{\int_{h_k(x)}^{h_{k+1}(x)}a_k(u)\,du}
{h_{k+1}(x)-h_k(x)}
=:b_k(h_k(x),h_{k+1}(x)).
\end{equation}
For an interior selected category $k<K$, define
$H_k(x):=(h_k(x),h_{k+1}(x))$; for the top category use
$H_K(x):=h_K(x)$.\footnote{The top category has a scalar control because its upper threshold is fixed at one.} Absorbing $\alpha_k$ into the selection correction function gives
\begin{equation}\label{eq: conditional mean nonparametric ordered}
m_k(x):=E[Y_i\mid X_i=x,D_i=k]=x'\beta_k+\lambda_k(H_k(x)).
\end{equation}
The partially linear argument of \citet{robinson1988root} identifies $\beta_k$ if
\begin{equation}\label{eq: ordered residual rank}
E[\operatorname{Var}(X_i\mid H_k(X_i),D_i=k)\mid D_i=k]\succ0.
\end{equation}

Three continuously distributed covariates can suffice for identification. More generally, let $T_k$ be a $q$-dimensional control function and write $X_i=(C_i,W_i)$, where $C_i\in\mathbb R^{q+1}$ and $W_i$ collects the remaining continuous or discrete covariates. Then, $q+1$ continuously distributed covariates suffice for identification of $\beta_k$ under the conditions below.

\begin{assumption}\label{ass: derivative rank}
Fix a category $k$ with $P(D_i=k)>0$ and $E[\|X_i\|^2\mid D_i=k]<\infty$.
Given $D_i=k$, for almost every $w$:
(i) conditional on $(W_i=w,D_i=k)$, $C_i$ has a positive density almost everywhere on a nonempty open set $\Omega_w$, and $P(C_i\in\Omega_w\mid W_i=w,D_i=k)=1$; (ii) $T_k(c,w)$ is continuously differentiable in $c$ on $\Omega_w$, with $J_k(c,w):=(\partial T_{k\ell}(c,w)/\partial c_j)_{j,\ell}\in\mathbb R^{(q+1)\times q}$ of rank $q$ throughout $\Omega_w$.
In addition, (iii) the matrix $E[N_k(X_i)\mid D_i=k]$ is positive definite where $N_k(X_i):=I_{q+1}-J_k(X_i)\{J_k(X_i)'J_k(X_i)\}^{-1}J_k(X_i)'$;
(iv) there exists a nonempty open set $\mathcal O_*$ contained in $T_k(\Omega_w,w)$ for almost every $w$;
(v) the selected support of $X_i$ is not contained in any affine hyperplane.
\end{assumption}
\begin{prop}\label{prop: derivative rank}
Under Assumption~\ref{ass: derivative rank}, $E[\operatorname{Var}(X_i\mid T_k(X_i),D_i=k)\mid D_i=k] \succ0$.
\end{prop}

\begin{remark}[Excluded variables]\label{rem: excluded variables}
With valid excluded variables $Z_i$, write the $q$-dimensional control as $T_k(X_i,Z_i)$. A continuous excluded variable can replace a continuous outcome covariate for identification. Moreover, $q$ continuously distributed excluded variables can suffice to identify $\beta_k$ without any continuous outcome covariates as shown in \citet{das2003nonparametric}.\footnote{The excluded variables must provide enough conditional variation to satisfy the relevant rank condition.}
\end{remark}

\begin{remark}[Additional identification power from the ordered structure]\label{rem: ordered structural identification}
For an interior category $k$ and $X_i=x$, the selection correction function can be written as: 
\begin{equation}\label{eq: ordered common mean correction}
\begin{aligned}
\lambda_k(H_k(x))
&=\alpha_k+E[a_k(U_i)\mid X_i=x,D_i=k]=\alpha_k+
\frac{A_k(h_{k+1}(x))-A_k(h_k(x))}
{h_{k+1}(x)-h_k(x)},
\end{aligned}
\end{equation}
where $A_k(u)=\int_0^u a_k(v)\,dv$. As $x$ changes, the interval changes, but the function $a_k$ being averaged remains the same. Thus the same $A_k$ determines the selection bias at every $x$ within category $k$. Restricting the correction function to averages of one common univariate function $a_k$ can permit identification with fewer continuous covariates. Appendix~\ref{sec: structural proofs} shows that the ordered structure can identify $\beta_k$ when $H_k(X_i)$ determines $X_i$.
\end{remark}

\subsection{Horizontal sorting: multinomial selection}\label{sec: multinomial}

When the categories are horizontally differentiated, the ordered threshold-crossing framework is inappropriate. Instead, I model selection as a utility-maximizing multinomial choice:
\begin{equation}\label{eq: multinomial choice}
D_i = k \quad \text{if} \quad u_k(X_i)  + \varepsilon_{ik} \geq u_j(X_i) + \varepsilon_{ij} \quad \text{for all } j \neq k,
\end{equation}
where $u_k(X_i)$ is the deterministic utility component for category $k$, and $\varepsilon_{ik}$ captures preference heterogeneity.\footnote{The classical Roy model \citep{roy1951some} is a special case of \eqref{eq: multinomial choice} in which $u_k(X_i) + \varepsilon_{ik} = Y_{ik}^* = \alpha_k + X_i'\beta_k + V_{ik}$, so that selection is driven by potential wages alone. The generalized Roy model \citep{heckman1990varieties} is the case $u_k(X_i) + \varepsilon_{ik} = Y_{ik}^* - C_{ik}$, allowing a non-pecuniary cost $C_{ik}$. The present framework accommodates both.} Expressing selection bias as a function of choice probabilities requires a stable relationship between covariates and latent heterogeneity. Throughout this subsection, maintain
$\varepsilon_i:=(\varepsilon_{i0},\ldots,\varepsilon_{iK})\ci X_i$ and
$E[V_{ik}\mid\varepsilon_i,X_i]=m_k^\varepsilon(\varepsilon_i)$, with
$E|V_{ik}|<\infty$ and $E[m_k^\varepsilon(\varepsilon_i)]=0$.
Joint independence of $(V_{ik},\varepsilon_i)$ and $X_i$ is a sufficient
primitive condition. These restrictions allow arbitrary dependence among the preference shocks and between preference shocks and the outcome error, but exclude covariate-dependent changes in the latent conditional mean that are not captured by the utility indices. Under these restrictions,
\begin{equation}\label{eq: multinomial bias general}
 E[V_{ik}\mid X_i=x,D_i=k]
 =E\left[m_k^\varepsilon(\varepsilon_i)\mid
       \varepsilon_{ij}-\varepsilon_{ik}\leq\delta_{kj}(x),\ j\ne k\right]
 =:b_k\bigl((\delta_{kj}(x))_{j\ne k}\bigr),
\end{equation}
where $\delta_{kj}(x):=u_k(x)-u_j(x)$. The bias therefore depends on the
$K$ pairwise utility differences. 

Suppose, in addition, that the preference shocks have a continuous,
strictly positive joint density on $\mathbb R^{K+1}$, and normalize $u_0=0$.
The utility-to-probability map
$\Psi:(u_1,\ldots,u_K)\mapsto(p_0(u),\ldots,p_K(u))$ is one-to-one and onto
the set of vectors with strictly positive entries summing to one, with a continuous inverse
\citep[Theorem~1 and Corollary~1]{norets2013surjectivity}.
Substituting this inverse into \eqref{eq: multinomial bias general} yields
\begin{equation}\label{eq: bias as CCP}
 E[V_{ik}\mid X_i=x,D_i=k]=\tilde b_k(p_0(x),\ldots,p_K(x)),
 \qquad \tilde b_k:=b_k\circ L_k\circ\Psi^{-1},
\end{equation}
where $L_ku=(u_k-u_j)_{j\ne k}$ is a linear bijection.
The function $\tilde\lambda_k=\alpha_k+\tilde b_k$ is continuous. Using $p_0(x)=1-\sum_{j=1}^Kp_j(x)$, the conditional mean is
\begin{equation}\label{eq: partial linear multinomial general}
 E[Y_i\mid X_i=x,D_i=k]
 =x'\beta_k+\tilde\lambda_k(P_K(x)),\qquad
 P_K(x):=(p_1(x),\ldots,p_K(x)).
\end{equation}
Hence, $\beta_k$ is identified by the same identification argument of \citet{robinson1988root} as in the ordered selection case. Proposition~\ref{prop: derivative rank} directly applies under
Assumption~\ref{ass: derivative rank} with $T_k=P_K$ and $q=K$, using $K+1$ continuous covariates. Next I propose additional restrictions that permit identification with fewer continuous covariates.

\subsubsection{Multinomial logit (MNL) selection}\label{sec: logit}

Suppose $u_j(x)=a_j+x'\gamma_j$, with $a_0=\gamma_0=0$,
and the preference shocks are i.i.d. standard Gumbel variables
\citep{mcfadden1972conditional}, independent of $X_i$. Impose the following restriction:
\begin{equation}\label{eq: own shock restriction}
 E[V_{ik}\mid\varepsilon_{i0},\ldots,\varepsilon_{iK},X_i]
 =\mu_k(\varepsilon_{ik}),
\end{equation}
meaning that the conditional mean of the potential-outcome
disturbance depends only on preference heterogeneity for that category. For example, a worker's unobserved affinity for STEM occupations may reflect skills that also raise potential earnings in STEM jobs. Conditional on this STEM-specific affinity and $X_i$, tastes for other occupations provide no additional information about mean potential earnings in STEM.

Let $A(x)=\sum_{j=0}^K\exp\{u_j(x)\}$ and
$M=\max_j\{u_j(x)+\varepsilon_{ij}\}$. The joint density of the maximum
and its winning category is $f_{M,D_i\mid X_i}(m,k\mid x)
= \exp\{u_k(x)\}e^{-m}\exp\{-A(x)e^{-m}\}.$ Dividing it by $p_k(x)=\exp\{u_k(x)\}/A(x)$ and using $\varepsilon_{ik}=M-u_k(x)$ when $D_i=k$ gives
\begin{equation}\label{eq: conditional distribution epsilon logit}
 \varepsilon_{ik}\mid D_i=k,X_i=x
 \sim\operatorname{EV}\bigl(\nu_k(x),1\bigr),\qquad
 \nu_k(x):=\log A(x)-u_k(x)=-\log p_k(x).
\end{equation}
Thus $E[Y_i\mid X_i=x,D_i=k]=x'\beta_k+\lambda_k(\nu_k(x))$ where $\lambda_k(v)=\alpha_k+b_k(v)$, $b_k(v)=E[\mu_k(G+v)]$, and $G$ is standard Gumbel.\footnote{This provides a conditional-mean version of the index sufficiency assumption in \citet{dahl2002mobility}.}
Figure~\ref{fig: mnl nonlinearity} in Appendix~\ref{sec: appendix cf figures} illustrates the nonlinearity of this index under linear utilities.
Lemma~\ref{lemma: gumbel class} in Appendix~\ref{sec: gumbel proofs} characterizes the admissible correction functions. Under linear utilities, the Gumbel and own-shock restrictions identify $\beta_k$ even when $X_i$ consists of a single continuous covariate for any $K\geq1$.

\begin{prop}\label{prop: gumbel point identification}
Suppose \eqref{eq: own shock restriction} holds, the preference shocks are i.i.d.\
standard Gumbel and independent of $X_i$, and
$u_j(x)=a_j+x'\gamma_j$ for $j=0,\ldots,K$, with $K\geq1$.
If the support of $X_i\in\mathbb R^{d_X}$ contains a nonempty open set,
then $\beta_k$ is point identified for every observed outcome category.
If $\nu_k$ is nonconstant, $\alpha_k$ is also identified,
and $\mu_k(e)$ is identified for almost every $e$.
\end{prop}

When the conditions of Proposition~\ref{prop: gumbel point identification} hold within one discrete-covariate group and $\nu_k$ varies there, continuous covariate variation identifies $\mu_k$ and the continuous-covariate coefficients. A common $\mu_k$ then determines the selection bias function $b_k$ across groups, so differences in their outcome means after accounting for continuous covariates and selection identify the discrete-covariate coefficients under the usual rank condition (see Corollary~\ref{cor: gumbel mixed identification} in Appendix~\ref{sec: gumbel proofs}).

\subsubsection{Exchangeability-based symmetry approximation}\label{sec: exchangeability}

Now I allow for dependence among preference shocks and between these shocks and outcome errors by introducing an exchangeability condition under the general specification \eqref{eq: multinomial choice}. 

\begin{samepage}
\begin{assumption}[Exchangeability]\label{ass: exchangeability}
For each outcome category $k$, the joint distribution of
$(V_{ik},\varepsilon_{i0},\ldots,\varepsilon_{iK})$ is invariant under
permutations of the preference-shock indices different from $k$. That is,
for any permutation $\pi$ of $\{0,\ldots,K\}$ with $\pi(k)=k$,
\[
 (V_{ik},(\varepsilon_{ij})_{j=0}^K)
 \overset{d}{=}(V_{ik},(\varepsilon_{i,\pi(j)})_{j=0}^K).
\]
\end{assumption}
\end{samepage}

\noindent This assumption allows $V_{ik}$ to depend on all preference shocks, provided their joint distribution is symmetric across alternatives other than $k$. It does not impose independence of irrelevant alternatives (IIA) in the MNL model.

In \eqref{eq: partial linear multinomial general}, permuting the non-chosen alternatives permutes their probabilities and leaves $\tilde\lambda_k$ unchanged by the symmetry argument of \citet[Proposition~3.2(ii)]{sheng2025social}. Define elementary symmetric polynomials:
\begin{equation}\label{eq: elementary symmetric}
 e_1=\sum_{j\ne k}p_j=1-p_k,\qquad
 e_2=\sum_{i<j,\ i,j\ne k}p_ip_j,\qquad\ldots,\qquad
 e_K=\prod_{j\ne k}p_j.
\end{equation}
On compact probability sets bounded away from zero, the continuous function \(\tilde\lambda_k\) can be approximated uniformly by polynomials in \((e_1,\ldots,e_K)\) \citep[Section~2.2 and footnote~10]{altonji2005cross}. For $1\leq L\leq K$, let $E_L(x):=(e_1(p(x)),\ldots,e_L(p(x)))$. Retaining $L<K$ coordinates gives the lower-dimensional approximation $E[Y_i\mid X_i=x,D_i=k] \approx x'\beta_k+\breve\lambda_k^{(L)}(E_L(x))$, which reduces the dimension of the correction function and can improve estimation accuracy.\footnote{Valid inference on $\beta_k$ requires the bias from omitted symmetric coordinates to be negligible relative to the standard error. Appendix~\ref{sec: symmetry approximation details} bounds this bias.}

\section{Estimation and Inference}\label{sec: estimator}

\subsection{Two-step sieve estimation}\label{sec: first stage details}

I estimate category-specific outcome equations in two stages, using spline functions of the estimated selection controls in the outcome regression. The first stage uses the full sample $\{(D_i,X_i)\}_{i=1}^n$. Write $p_j(x)=P(D_i=j\mid X_i=x)$, $p(x)=(p_1(x),\ldots,p_K(x))'$, and $p_0(x)=1-\sum_{j=1}^Kp_j(x)$. In the ordered selection case, for the common-index semiparametric model with known $F_\varepsilon$, I approximate the selection index $h(x)$ by splines and estimate the parameters by sieve maximum likelihood. In the nonparametric threshold model, estimate each $h_j(x)=P(D_i\leq j-1\mid X_i=x)$ by binary sieve probit or logit estimation. Constrained estimation or monotone rearrangement preserves the ordering of the cumulative probabilities \citep{chernozhukov2010quantile}.

For multinomial selection, I estimate the linear-utility MNL model by maximum likelihood. With i.i.d.\ Gumbel preference shocks and unrestricted utilities, I approximate the utilities by splines and maximize the sieve MNL likelihood. More generally, any strictly positive choice-probability vector can be written in logit form:
\[
 p_j(x)=\frac{\exp\{r_j(x)\}}{1+\sum_{\ell=1}^K\exp\{r_\ell(x)\}},
 \qquad r_j(x)=\log\frac{p_j(x)}{p_0(x)},\quad j=1,\ldots,K.
\]
This identity requires no assumption about the preference-shock distribution. In the general nonparametric model, I therefore approximate these log odds by sufficiently rich splines and maximize the same likelihood to estimate the choice probabilities \citep{chen2007large}.

For the flexible spline first stages, I regularize the likelihood by penalizing roughness. Define the selection indices as $\eta(x)\in\mathbb R^{d_\eta}$ and their spline approximation as $\Psi_n(x)'a$, with $a\in\mathbb R^{q_n}$. Let $\mathcal L(D,\eta)$ be the one-observation selection log likelihood. Given smoothing parameters $\lambda=(\lambda_1,\ldots,\lambda_M)'$, estimate
\begin{equation}\label{eq: penalized sieve likelihood}
 \hat a_\lambda=\argmax_a\left\{\sum_{i=1}^n\mathcal L(D_i,\Psi_n(X_i)'a)-\tfrac12a'S(\lambda)a\right\},
 \qquad S(\lambda)=\sum_{j=1}^M\lambda_jS_j.
\end{equation}
Here $S_j$ is a $q_n\times q_n$ symmetric positive semidefinite penalty matrix constructed from the chosen spline basis. The index $j=1,\ldots,M$ distinguishes the penalties applied to different spline components and coordinate directions. For univariate cubic regression splines, $a'S_ja$ measures the integrated squared second derivative of a component. Tensor components receive separate penalties along each coordinate. A larger $\lambda_j$ suppresses curvature in that direction. The penalty applies to spline curvature, leaving the intercept and coefficients on linear and categorical covariates unpenalized.

The basis dimension, interaction terms, and knot rule determine the approximating space within which the smoothing parameters control roughness. Unrestricted multivariate functions require the corresponding higher-order tensor interactions. I select $\hat\lambda$ by restricted maximum likelihood (REML), using either Laplace REML or fast REML (fREML) \citep{wood2016smoothing}. The estimated indices are $\hat\eta(x)=\Psi_n(x)'\hat a$, where $\hat a=\hat a_{\hat\lambda}$. Appendix~\ref{sec: estimation details} gives the model-specific likelihoods and basis construction.

In the second stage, for each category $k$, write the maintained conditional mean as
\begin{equation}\label{eq: general conditional mean}
 E[Y_i\mid X_i=x,D_i=k]=x'\beta_{k0}+\lambda_{k0}(g_{k0}(x)).
\end{equation}
The control function \(g_{k0}(x)\) depends on the selection structure.  It is the common index \(h(x)\) in the semiparametric ordered model, and one or two adjacent cumulative probabilities in the nonparametric ordered model. For general multinomial selection, the control function is the full probability vector. The Gumbel and own-shock restrictions reduce it to \(-\log p_k(x)\). Under exchangeability, I use lower-dimensional approximations based on symmetric polynomials of the choice probabilities.
Write $g_{k0}(x)=\phi_k(p(x))$, where $\phi_k$ is the corresponding transformation of the choice probabilities. Plugging the estimated probabilities into this transformation gives $\hat g_k(x)$. Let $\mathcal I_k=\{i:D_i=k\}$ and $n_k=|\mathcal I_k|$. Construct a B-spline basis $B_n$ containing a constant at these estimated controls, and jointly estimate the outcome and spline coefficients by unpenalized least squares:
\begin{equation}\label{eq: second stage regression}
 (\hat\beta_k,\hat\delta_k)
 =\argmin_{\beta,\delta}\sum_{i\in\mathcal I_k}
 \{Y_i-X_i'\beta-B_n(\hat g_k(X_i))'\delta\}^2.
\end{equation}
The basis is univariate for a scalar control. For multiple controls, a full tensor basis approximates the unrestricted multivariate selection correction function.

\subsection{Inference}\label{sec: asymptotics}

Because the second-stage spline regressors are constructed from estimated selection probabilities, inference must account for first-stage estimation error. I establish a joint asymptotic expansion for these two stages with growing bases and data-selected smoothing. The analysis controls the bias introduced by the first-stage penalty, derives the influence of the estimated multivariate controls, and provides a consistent analytic variance estimator.

For each category $k$, let $T_{ik}=\mathbf1\{D_i=k\}$, $\pi_k=E[T_{ik}]$, $\Pi_k(s)=E[X_i\mid g_{k0}(X_i)=s,D_i=k]$, $U_{ik}=X_i-\Pi_k(g_{k0}(X_i))$, $Q_k=E[T_{ik}U_{ik}U_{ik}']$, and $\varepsilon_{ik}=Y_i-X_i'\beta_{k0}-\lambda_{k0}(g_{k0}(X_i))$.
The outcome error is defined for selected observations, and products multiplied by $T_{ik}$ are zero outside category $k$. All sums and asymptotic normalizations use the full sample of size $n$.

To identify all slopes in category $k$, the regressors must retain variation in every direction after conditioning on the selection controls. The following assumption states this rank condition and the sampling and moment conditions used for inference.

\begin{assumption}[Sampling and regular rank]\label{ass: sampling}\label{ass: rank}\label{ass: errors}
The data are i.i.d., $X_i$ is bounded, $\pi_k>0$, and \eqref{eq: general conditional mean} holds. The conditional fourth moment of $\varepsilon_{ik}$ is uniformly bounded, and its conditional variance is bounded and bounded away from zero. $Q_k$ is positive definite. 
\end{assumption}

For the concave first-stage log likelihood in \eqref{eq: penalized sieve likelihood}, let $s(D,\eta)=\partial_\eta\mathcal L(D,\eta)$ and $I_0(x)=-E[\partial_\eta^2\mathcal L(D_i,\eta_0(x))\mid X_i=x]$. For a joint likelihood, $I_0(x)$ equals the conditional score covariance. When each cumulative choice probability is estimated by a separate probit regression, it is block diagonal with the equation-specific information matrices. Normalize each $S_j$ so that its smallest positive eigenvalue relative to $G_n=E[\Psi_n(X_i)\Psi_n(X_i)']$ is one, and define $J_j(a)^2=a'S_ja$ and $J_j(f)$ for the same roughness of a function. Use common coordinates in which $cI\preceq G_n\preceq CI$ and $\sup_x\|\Psi_n(x)\|\leq C\sqrt{q_n}$. The number $M$ of smoothing parameters is fixed as the sieve dimension $q_n$ grows. Let $\mathcal F$ denote the closed linear function space approximated by the first-stage sieve, with $\eta_0\in\mathcal F$. Put $g_k(x;\eta)=\phi_k(p(x;\eta))$ and define
\begin{equation}\label{eq: drift loading}
 \Gamma_k(x)=p_k(x)U_k(x)\nabla\lambda_{k0}(g_{k0}(x))'
                  \partial_\eta g_k(x;\eta_0(x)),\quad
 v_k^*=\operatorname{Proj}_{I_0}\bigl(\Gamma_k I_0^{-1}\mid\overline{\mathcal F}\bigr).
\end{equation}
Here $U_k(x)=x-\Pi_k(g_{k0}(x))$. The matrix $\Gamma_k$ measures how first-stage errors enter the second-stage estimating equation. For each outcome coefficient, $v_k^*$ captures the part of this effect associated with function changes allowed by $\mathcal F$.

\begin{assumption}\label{ass: penalized sieve}
(i) Within each discrete cell, the continuous covariates have a density bounded above and away from zero on $[0,1]^{d_c}$. All choice probabilities, including conditional ones, are bounded away from zero and one; the control maps are $C^2$ with bounded derivatives on this probability region.
(ii) $\eta_0\in\mathcal F$ has $C^s$ components, $s>2$, each with at most $d_{\mathcal F}$ continuous arguments, where $d_{\mathcal F}$ is fixed and finite.
(iii) The first-stage sieve uses tensor splines of fixed degree $p_\eta\geq3$, with $m_n\to\infty$ predetermined, quasi-uniform knots per axis and $q_n\asymp m_n^{d_{\mathcal F}}$. Its $L^2$ and sup-norm approximation bounds for $\eta_0$, $\omega_{2n}$ and $\omega_{\infty n}$, are both $O(m_n^{-\min(s,p_\eta+1)})$ for unconstrained B-splines. For natural cubic splines, they are respectively $O(m_n^{-\min(s,5/2)})$ and $O(m_n^{-2})$. Set
\[
 \rho_n=\sqrt{q_n/n}+\omega_{2n},\qquad
 e_n=q_n/\sqrt n+\sqrt{q_n}\,\omega_{2n}+\omega_{\infty n}.
\]
(iv) In the common normalized coordinates, penalties on the same component commute, $J_j^2$ is uniformly equivalent to the integrated squared second derivative in its direction, $\operatorname{tr}(G_nS_j^+)=o(n)$, $\|S_j\|=O(m_n^4)$, and the lower smoothing bound $\lambda_{\min}$ is fixed.
(v) $J_j(\eta_0)>0$ in every penalized direction.
(vi) $v_k^*$ is bounded with bounded integrated squared second derivatives. Its $L^2$ spline approximation bound $\omega_{vn}=O(m_n^{-2})$ satisfies $\sqrt{q_n}\,\omega_{vn}=O(1)$.
(vii) The dimension and approximation errors satisfy
\[
 \frac{q_n^2\log q_n}{n}\to0,\qquad
 \sqrt n\,\omega_{2n}^2\to0,\qquad
 \sqrt n(\omega_{vn}+\omega_{2n})\omega_{2n}\to0.
\]
\end{assumption}

Condition (i) ensures that all covariate regions have positive probability. (ii)--(iii) require the growing sieve to approximate the true selection functions. (iv)--(v), together with the growth conditions, bound the selected smoothing parameters without requiring their asymptotic normality. (vi) requires sufficient smoothness of the functions transmitting first-stage error to the outcome coefficients. (vii) balances flexibility against sampling error so that terms beyond the first-stage linear influence vanish at the $n^{-1/2}$ scale. 

The second-stage basis $B_n(s)$ contains $\kappa_n$ spline basis functions, with $J_n\to\infty$ knots for each argument of $\lambda_{k0}(s)$. For a spline approximant $\ell_{kn}=B_n'\delta_{kn}$ to $\lambda_{k0}$, define
\[
 a_n=\|\ell_{kn}-\lambda_{k0}\|_\infty,\quad
 b_n=\|\nabla\ell_{kn}-\nabla\lambda_{k0}\|_\infty,\quad
 H_n=\sup_s\|\nabla^2\ell_{kn}(s)\|_{\rm op}.
\]

\begin{assumption}\label{ass: joint penalized growth}
Conditional on $D_i=k$, $g_{k0}(X_i)$ has a density bounded above and away from zero on a compact rectangle. On a fixed compact neighborhood of this support containing the estimated controls with probability tending to one, $\lambda_{k0}\in C^2$ has bounded derivatives, $\Pi_k$ is Lipschitz with spline approximation error $O(J_n^{-s_X})$, and the bounds $a_n,b_n,H_n$ hold. The tensor basis $B_n$ has fixed degree $p_g\geq3$, predetermined quasi-uniform knots, and normalization $\kappa_n^{1/2}$. Set $\bar c_n=J_n^{-s_X}+\sqrt{\kappa_n/n}+(1+J_n)\rho_n$ and require
\begin{align}
 &a_n+b_n+\kappa_n\log n/n+J_ne_n\to0,\label{eq: main design rates}\\
 &\sqrt n\{\bar c_na_n+b_n\rho_n+H_n\rho_n^2+\bar c_n\rho_n\}\to0,
 \label{eq: main joint rates}\\
 &J_n\{\sqrt{\kappa_n/n}+\rho_n+a_n\}\to0.
 \label{eq: main variance rates}
\end{align}
\end{assumption}

The products in \eqref{eq: main joint rates} bound the combined effects of first-stage estimation, spline approximation, and residualization errors. Equation~\eqref{eq: main variance rates} ensures that the derivatives of the estimated spline function consistently estimate $\nabla\lambda_{k0}$ in the analytic variance formula. Appendix~\ref{sec: penalized first stage} gives examples satisfying the approximation and growth conditions in Assumptions~\ref{ass: penalized sieve} and \ref{ass: joint penalized growth}.

To compute the analytic variance, let $\hat U_{ik}$ and $\hat\varepsilon_{ik}$ be the residualized covariates and outcome residuals from \eqref{eq: second stage regression}, and let $\hat\ell_k=B_n'\hat\delta_k$. Form
\begin{align*}
 \hat h_{ik}&=T_{ik}\hat U_{ik}\nabla\hat\ell_k(\hat g_k(X_i))'
                  \partial_\eta g_k(X_i;\hat\eta(X_i)),\\
 \hat\gamma_k&=\frac1n\sum_i\Psi_n(X_i)\hat h_{ik}'.
\end{align*}
The first-stage score-form estimate is
\begin{equation}\label{eq: estimated projected influence}
 \hat\chi_{ik}=\hat\gamma_k'
 \{\hat{\mathcal I}+S(\hat\lambda)/n\}^{-1}
 \Psi_n(X_i)s(D_i,\hat\eta(X_i)).
\end{equation}
Here $\hat{\mathcal I}=n^{-1}\sum_i\Psi_n(X_i)\hat I_i\Psi_n(X_i)'$ is the coefficient information matrix, with $\hat I_i$ the observed or working index information at the estimated indices. Define $\hat Q_k=n^{-1}\sum_iT_{ik}\hat U_{ik}\hat U_{ik}'$, $\hat\psi_{ik}=T_{ik}\hat U_{ik}\hat\varepsilon_{ik}-\hat\chi_{ik}$, and $\bar\psi_k=n^{-1}\sum_i\hat\psi_{ik}$. Then
\begin{equation}\label{eq: V hat}
 \hat{\mathcal V}_k=\hat Q_k^{-1}
 \left\{\frac1n\sum_i(\hat\psi_{ik}-\bar\psi_k)(\hat\psi_{ik}-\bar\psi_k)'\right\}\hat Q_k^{-1}.
\end{equation}

\begin{theorem}[Two-stage inference with data-selected smoothing]\label{thm: penalized two step}
Suppose Assumptions~\ref{ass: sampling}, \ref{ass: penalized sieve}, and \ref{ass: joint penalized growth} hold. For Laplace REML, additionally assume that, with probability tending to one, the estimated indices at every solution lie in a common compact set. Then:
\begin{enumerate}
\renewcommand{\labelenumi}{(\roman{enumi})}
\item $\max_j\hat\lambda_j=O_p(q_n)$, and the first-stage contribution is
\begin{equation}\label{eq: projected influence}
 \chi_{ik}=v_k^*(X_i)s(D_i,\eta_0(X_i)).
\end{equation}
For separate binary equations, $\chi_{ik}$ sums their influences within observation $i$, retaining cross-equation covariances. The estimator satisfies
\begin{equation}\label{eq: asymptotic normality}
 \sqrt n(\hat\beta_k-\beta_{k0})
 =Q_k^{-1}\frac1{\sqrt n}\sum_{i=1}^n
 \{T_{ik}U_{ik}\varepsilon_{ik}-\chi_{ik}\}+o_p(1)
 \ \xrightarrow{d}\ N(0,\mathcal V_k),
\end{equation}
where
\begin{equation}\label{eq: full sample variance}
 \mathcal V_k=Q_k^{-1}\left\{
 E[T_{ik}U_{ik}U_{ik}'\varepsilon_{ik}^2]+E[\chi_{ik}\chi_{ik}']
 \right\}Q_k^{-1}.
\end{equation}
\item The analytic variance estimator \eqref{eq: V hat} satisfies $\hat{\mathcal V}_k\to_p\mathcal V_k$, so $\hat{\mathcal V}_k/n$ consistently estimates the first-order covariance of $\hat\beta_k$.
\end{enumerate}
\end{theorem}

Estimating the first stage contributes to the asymptotic variance through $\chi_{ik}$, while penalization and smoothing have no additional first-order effect. With a fixed second-stage spline basis, the target is the population projection coefficient (Appendix~\ref{sec: fixed basis details}). Inference on $\beta_{k0}$ requires a zero projection gap with a fixed basis or an $o(n^{-1/2})$ gap with a growing basis. Observed and working information can differ under fixed-dimensional probit misspecification.

\subsection{Using the ordered restriction}\label{sec: structural ordered estimation}

The restriction implied by the ordered selection model in Remark~\ref{rem: ordered structural identification} can be imposed directly through the spline basis. Fix category $k$ and write $H_i=(s_i,t_i)=(h_k(X_i),h_{k+1}(X_i))$, with $t_i=1$ for the top category. Let $\varphi_{n\ell}$ be univariate splines centered to have zero integral over $[0,1]$ and define
\begin{equation}\label{eq: structural spline basis}
 A_{n\ell}(v)=\int_0^v\varphi_{n\ell}(u)\,du, \quad Z_{n\ell}(s,t)=\frac{A_{n\ell}(t)-A_{n\ell}(s)}{t-s}
 =\int_0^1\varphi_{n\ell}((1-v)s+vt)\,dv.
\end{equation}
To approximate the selection correction function, I average each spline basis function over the category's selection interval and use these averages as regressors. Given the estimated thresholds $\hat H_i=(\hat s_i,\hat t_i)$, solve the following minimization problem: 
\begin{equation}\label{eq: structural sieve OLS}
 (\hat\beta_k^{\rm str},\hat\alpha_k,\hat\theta_k)
 =\argmin_{\beta,\alpha,\theta}\sum_{i:D_i=k}
 \{Y_i-X_i'\beta-\alpha-Z_n(\hat H_i)'\theta\}^2.
\end{equation}
Approximating the common univariate function $a_k$ with $J_n$ spline bases uses at most $J_n$ coefficients, while an unrestricted spline approximation in both thresholds uses order $J_n^2$ coefficients with $J_n$ splines per threshold.

For inference, let $\hat R_i$ be the residual from regressing each component of $X_i$ on a constant and $Z_n(\hat H_i)$ among observations with $D_i=k$. Let $R_i$ be the mean-square limit of the corresponding population regression residuals, using the true thresholds $H_i$, as the spline basis grows. Define $Q^{\rm str}=E[T_{ik}R_iR_i']$. Using the selection bias $b_k$ defined in \eqref{eq: nonparametric bias ordered}, let $\varepsilon_i=Y_i-X_i'\beta_{k0}-\alpha_{k0}-b_k(H_i)$ in category $k$. The maintained ordered model gives $E[\varepsilon_i\mid X_i,D_i=k]=0$. For an interior category, threshold estimation affects $b_k(s,t)$ through its derivatives
\begin{equation}\label{eq: structural endpoint gradient}
 d_0(s,t):=\nabla b_k(s,t)
 =\begin{pmatrix}
 \{b_k(s,t)-a_k(s)\}/(t-s)\\
 \{a_k(t)-b_k(s,t)\}/(t-s)
 \end{pmatrix}.
\end{equation}
With the upper threshold fixed at one, the top category uses only the first entry of \eqref{eq: structural endpoint gradient}.

For an interior category, the growing-basis result uses a two-dimensional threshold support.\footnote{Appendix~\ref{sec: structural curve support} treats the fixed-basis projection when both thresholds depend on a common index as in the semiparametric common index ordered model.} Write
\[
 \Gamma_k^{\rm str}(x)=p_k(x)R(x)d_0(H_k(x))'\partial_\eta H_k(x;\eta_0(x)),
 \qquad v_k^{\rm str}=\operatorname{Proj}_{I_0}
       (\Gamma_k^{\rm str}I_0^{-1}\mid\mathcal F),
\]
using only the varying threshold at the top category. For separately estimated cumulative probabilities, form these quantities for each binary first-stage equation using its own information matrix, and sum the resulting influences.

\begin{assumption}\label{ass: structural inference}
Maintain Assumption~\ref{ass: sampling}, with $Q^{\rm str}\succ0$ in place of $Q_k\succ0$, and Assumption~\ref{ass: penalized sieve}, using $v_k^{\rm str}$ and its approximation error in (vi)--(vii). Conditional on $D_i=k$, the varying thresholds have support $I_s\times I_t$, $\sup I_s<\inf I_t$, for an interior category, or $I_s\subset(0,1)$ for the top category, with a density bounded above and away from zero on that support. The intervals are compact and the widths are bounded below. The density and $E[X_i\mid H_i,D_i=k]$ are $C^5$ with bounded derivatives. Assume $a_k\in L^2_0[0,1]$ and that its primitive $A_k$ is $C^5$ on neighborhoods of the support intervals of the varying thresholds. Use centered cubic splines with predetermined quasi-uniform knots of mesh order $J_n^{-1}$, with each support boundary at distance at least $c/J_n$, $c>0$, from its adjacent knots. Set
\[
 \bar c_n^{\rm str}=J_n^{-5}+\sqrt{J_n/n}+(1+J_n)\rho_n,
\]
where $\rho_n,e_n$ are defined in Assumption~\ref{ass: penalized sieve}, and require
\begin{align}
 J_ne_n+J_n^2\log n/n+J_n^{3/2}\bar c_n^{\rm str}&\longrightarrow0,\label{eq: structural design rates}\\
 \sqrt n\{J_n^{-5}+J_n^{-4}\rho_n+\rho_n^2
              +\bar c_n^{\rm str}\rho_n\}&\longrightarrow0.
 \label{eq: structural joint rates}
\end{align}
\end{assumption}

The integrated cubic basis approximates $A_k$ by quartic splines at the thresholds, yielding $O(J_n)$ effective coefficients, $O(J_n^{-5})$ approximation error, and consistent derivative estimates. Appendix~\ref{sec: structural inference details} derives these bounds and the residualization rate $\bar c_n^{\rm str}$. The conditions are sufficient: for example, natural cubic first-stage splines with three continuous arguments and $s\geq5/2$ allow $m_n\asymp n^{1/8}$ and $J_n\asymp n^{0.11}$.

Write $\hat b_k(s,t)=\hat\alpha_k+Z_n(s,t)'\hat\theta_k$ for the estimated function in \eqref{eq: structural sieve OLS}, $\hat d_i$ for its gradient in the varying thresholds, and $\hat\varepsilon_i$ for the outcome residual. Use \eqref{eq: estimated projected influence} and \eqref{eq: V hat} with
\[
 \hat h_i^{\rm str}=T_{ik}\hat R_i\hat d_i'
       \partial_\eta H_k(X_i;\hat\eta(X_i)),\qquad
 \hat\gamma^{\rm str}=n^{-1}\sum_i\Psi_n(X_i)\hat h_i^{{\rm str}\prime},
\]
and with $\hat R_i$ replacing $\hat U_{ik}$. Denote the resulting covariance by $\hat{\mathcal V}_k^{\rm str}$.

\begin{theorem}[Structural ordered two-step inference]\label{thm: structural inference}
Under Assumption~\ref{ass: structural inference}, and the compactness condition in Theorem~\ref{thm: penalized two step} when Laplace REML is used,
\begin{equation}\label{eq: structural asymptotic expansion}
 \sqrt n(\hat\beta_k^{\rm str}-\beta_{k0})
 =(Q^{\rm str})^{-1}\frac1{\sqrt n}\sum_i
 \{T_{ik}R_i\varepsilon_i-\chi_i^{\rm str}\}+o_p(1)
 \ \xrightarrow{d}\ N(0,\mathcal V_k^{\rm str}),
\end{equation}
where $\chi_i^{\rm str}=v_k^{\rm str}(X_i)s(D_i,\eta_0(X_i))$ for a joint likelihood, or the sum of these terms for separately estimated cumulative-probability equations, and
\begin{equation}\label{eq: structural asymptotic variance}
 \mathcal V_k^{\rm str}
 =(Q^{\rm str})^{-1}
 \{E[T_{ik}R_iR_i'\varepsilon_i^2]+E[\chi_i^{\rm str}\chi_i^{{\rm str}\prime}]\}
 (Q^{\rm str})^{-1}.
\end{equation}
Moreover, $\hat{\mathcal V}_k^{\rm str}\to_p\mathcal V_k^{\rm str}$. Thus $\sqrt{\hat{\mathcal V}_{k,jj}^{\rm str}/n}$ consistently estimates the first-order standard error of the $j$th slope. The same conclusion holds when the first-stage loading is computed by differentiating the second-stage least-squares estimator.
\end{theorem}

\section{Simulations}\label{sec: simulations}

The simulations assess finite-sample performance of the proposed estimators under ordered and multinomial selection. Each design uses $n=5{,}000$ observations and 500 replications. 

\subsection{Ordered selection}

The main ordered design compares six estimators: (i)~OLS; (ii)~Linear, using the selection term from a linear ordered probit; (iii)~Semiparametric, using a univariate spline in the estimated common index of an ordered probit; (iv)~Oracle, estimating correctly specified parametric models in both stages; (v)~Sieve, using splines in the estimated cumulative probabilities; and (vi)~Restricted Sieve, using one latent spline per category. Sieve and Restricted Sieve share the same penalized sieve probit first stage that estimates $h_j(x)=P(D_i<j\mid X_i=x)$, $j=1,\ldots,K$, with REML, and monotone rearrangement enforces their ordering. Semiparametric instead imposes $h_j(x)=\Phi(c_j-h(x))$ with constant thresholds $c_j$, estimating the common index by penalized ordered-probit likelihood with REML and the same spline terms. In the second stage, Sieve uses $(\hat h_k,\hat h_{k+1})$ for an interior category and $\hat h_K$ for the top category. I use six cubic B-spline basis functions for each estimated probability in Sieve and for the estimated common index in Semiparametric.\footnote{With two probabilities, I also include the 16 products of four basis functions per probability. The Sieve standard errors account for the estimated interior knots (quartiles for marginals, medians for products). The growing-basis theory uses predetermined knots.} Restricted Sieve approximates the latent mean function \(a_k(u)\) using eight cubic B-spline basis functions with equally spaced knots. Its average over each estimated selection interval enters the outcome regression.

The ordered design (DGP1) has three continuous covariates, $K=3$, and category-specific thresholds. Selection follows \eqref{eq: nonparametric ordered} with $U_i=\Phi(\varepsilon_i)$, $\varepsilon_i\sim N(0,1)$, and
\[
 h_k(X,Z,W)=\Phi\bigl(-1+a_k\Lambda(0.8Z+0.4W^2)-h(X,Z,W)\bigr),\qquad (a_1,a_2,a_3)=(0,3,5),
\]
where $\Lambda$ is the logistic distribution function and
\[
 h(X,Z,W)=0.5X+0.3Z+0.8W-0.5X^2+0.2Z^2-0.4XZ,
\]
with $(X,Z,W)\sim N(\mathbf 0,I_3)$. The first threshold moves only with $h$, and the other two also move with $0.8Z+0.4W^2$. The outcome equations are $Y_{ik}=\alpha_k+(X_i,Z_i,W_i)\beta_k+V_{ik}$, with $(\alpha_k,\beta_k')=(0.5,0.5,0.25,0.3)$, $(0.6,0.7,0.5,0.4)$, and $(0.8,1,0.6,0.5)$ for $k=1,2,3$, respectively.\footnote{The errors $(\varepsilon_i,V_{i1},V_{i2},V_{i3})$ are jointly normal with mean zero and independent of the covariates. Each $V_{ik}$ has variance $1.5$, and $\operatorname{Cov}(V_{ik},V_{i\ell})=0.45$ for $k\ne\ell$. Their correlations with $\varepsilon_i$ are $(0.3,0.5,0.7)$.}

\begin{table}[htbp]
\centering
\caption{Ordered selection DGP1}\label{tab: ordered K3 baseline}

\small
\setlength{\tabcolsep}{2.5pt}
\renewcommand{\arraystretch}{1}
\begin{tabular*}{\textwidth}{@{\extracolsep{\fill}}lccccccccc@{}}
\toprule
& \multicolumn{3}{c}{Category 1} & \multicolumn{3}{c}{Category 2} & \multicolumn{3}{c}{Category 3} \\
\cmidrule(lr){2-4}\cmidrule(lr){5-7}\cmidrule(lr){8-10}
Estimator & RMSE & $|$Bias$|$ & Coverage & RMSE & $|$Bias$|$ & Coverage & RMSE & $|$Bias$|$ & Coverage \\
\midrule
OLS & 0.087 & 0.079 & 0.473 & 0.141 & 0.131 & 0.283 & 0.259 & 0.247 & 0.239 \\
Linear & 6.835 & 4.458 & 0.835 & 48.508 & 40.307 & 0.723 & 1.514 & 1.216 & 0.711 \\
Semiparametric & 0.038 & 0.015 & 0.905 & 0.073 & 0.043 & 0.773 & 0.138 & 0.101 & 0.655 \\
Oracle & 0.034 & 0.003 & 0.942 & 0.049 & 0.002 & 0.951 & 0.074 & 0.004 & 0.947 \\
Sieve & 0.049 & 0.005 & 0.952 & 0.082 & 0.011 & 0.970 & 0.079 & 0.006 & 0.942 \\
Restricted Sieve & 0.035 & 0.003 & 0.943 & 0.051 & 0.002 & 0.951 & 0.079 & 0.006 & 0.947 \\
\bottomrule
\end{tabular*}
\smallskip
\parbox{\textwidth}{\footnotesize\textit{Notes:} The reported numbers are averaged over the three slope coefficients within each category. 95\% confidence intervals are computed using analytic standard errors developed in Sections~\ref{sec: asymptotics} and~\ref{sec: structural ordered estimation} for Linear, Semiparametric, Oracle, Sieve, and Restricted Sieve, and heteroskedasticity-robust standard errors for OLS.}
\end{table}

Table~\ref{tab: ordered K3 baseline} shows large bias and RMSE for OLS and Linear. Semiparametric substantially reduces RMSE relative to Linear, but its misspecified common-index restriction leaves bias and substantial undercoverage, especially in the top category. Oracle performs best. Restricted Sieve brings interior-category RMSE close to Oracle's. Its reductions relative to Sieve are 29--37\% here and 11--12\% in the common-index version of this design (Appendix~\ref{sec: appendix ordered common index}). Oracle and both sieve estimators perform similarly in the top category. Imposing the ordered structure reduces the dimension of the nonparametric approximation and permits identification with fewer continuous covariates (Remark~\ref{rem: ordered structural identification}). Together with the gains in precision and coverage close to the nominal level, these advantages support Restricted Sieve as the preferred estimator for the ordered selection model.\footnote{Appendix~\ref{sec: appendix ordered common index} reports common-index designs, including one with a binary covariate. In both interior categories, weakening the nonlinear terms in the selection index raises RMSE much more for Sieve than for Restricted Sieve (Appendix~\ref{sec: appendix weak identification}).}

\subsection{Multinomial selection}

For multinomial selection, I consider the following five estimators: (i)~OLS; (ii)~MLogit, using $\hat\nu_k(x)=-\log\hat p_k(x)$ from the sieve first stage as a scalar control; (iii)~Oracle, using the true correction function $E[V_{ik}\mid X_i=x,D_i=k]$ as a regressor; (iv)~Sieve, using all predicted probabilities through cubic B-spline marginals and pairwise interactions; and (v)~Exch-$L2$, which exploits exchangeability by using a spline in the sum of the non-chosen categories' estimated probabilities ($e_1$) and the sum of their pairwise products ($e_2$). The first stage is a penalized sieve MNL with cubic regression-spline terms in each continuous covariate and their pairwise tensor products, with smoothing parameters selected by restricted maximum likelihood.

Two DGPs are considered, both with $K=3$ and three continuous covariates $(X_i,Z_i,W_i)\sim N(\mathbf 0,I_3)$. Selection is $D_i=\argmax_{j\in\{0,1,2,3\}}\{f_j(X_i,Z_i,W_i)+\varepsilon_{ij}\}$, where $f_0\equiv0$ and the other utilities contain quadratic terms and pairwise interactions.\footnote{The outcome coefficients are $(\alpha_1,\beta_{11},\beta_{12},\beta_{13})=(0.4,0.5,0.7,0.3)$, $(\alpha_2,\beta_{21},\beta_{22},\beta_{23})=(0.6,0.8,0.5,0.4)$, and $(\alpha_3,\beta_{31},\beta_{32},\beta_{33})=(0.5,0.3,0.9,0.6)$.} In DGP1, the preference shocks are independent standard Gumbel and the outcome errors are $V_{ik}=\varepsilon_{ik}+0.2\eta_{ik}$, where $\eta_{ik}\sim N(0,1)$ is independent of the covariates and preference shocks. Under the own-shock restriction, selection bias depends only on $e_1=1-p_k$. Exch-$L2$ additionally uses $e_2$, while Sieve uses the full probability vector.

DGP2 is designed to showcase the practical value of the exchangeability framework. The preference shocks are equicorrelated normal, $\varepsilon_{ij} = \sqrt{\rho}\, c_i + \sqrt{1 - \rho}\, z_{ij},\ \forall j \in \{0,1,2,3\},$ where $c_i \sim N(0,1)$ is a common factor, $z_{ij} \sim N(0,1)$ are independent, and $\rho = 0.5$. This symmetric structure is exchangeable by construction but violates IIA. The covariates, utility functions, and linear outcome coefficients are identical to DGP1. Let $\bar\varepsilon_{i,-k}=K^{-1}\sum_{j\ne k}\varepsilon_{ij}$ be the mean preference shock of alternatives, including non-participation. The outcome errors are
\[
 Q_{ik}=\sum_{j\ne k}(\varepsilon_{ij}-\bar\varepsilon_{i,-k})^2-1,
 \qquad V_{ik}=c\{\varepsilon_{ik}-4Q_{ik}\}+0.2\eta_{ik},
\]
where $\eta_{ik}\sim N(0,1)$ is independent of the covariates and preference shocks.\footnote{The normalization constant $c=\sqrt{23/204}$ keeps the outcome-error variance fixed at $23/12+0.2^2$.} The term $Q_{ik}$ measures the dispersion of competing preference shocks, centered at its unconditional mean. Outcome errors therefore depend on competing shocks as well as the category's own shock. This violates the own-shock restriction while preserving symmetry across alternatives.

\begin{table}[htbp]
\centering
\caption{Multinomial selection DGP1 and DGP2}\label{tab: multinomial results}

\small
\setlength{\tabcolsep}{3pt}
\renewcommand{\arraystretch}{1}
\begin{tabular*}{\textwidth}{@{\extracolsep{\fill}}lcccccc@{}}
\toprule
& \multicolumn{3}{c}{MNL DGP1} & \multicolumn{3}{c}{MNL DGP2} \\
\cmidrule(lr){2-4}\cmidrule(lr){5-7}
Estimator & RMSE & $|$Bias$|$ & Coverage & RMSE & $|$Bias$|$ & Coverage \\
\midrule
OLS & 0.207 & 0.201 & 0.113 & 0.078 & 0.065 & 0.523 \\
MLogit & 0.053 & 0.005 & 0.946 & 0.050 & 0.020 & 0.921 \\
Oracle & 0.043 & 0.002 & 0.950 & 0.040 & 0.001 & 0.953 \\
Sieve & 0.097 & 0.017 & 0.989 & 0.086 & 0.006 & 0.982 \\
Exch-$L2$ & 0.058 & 0.006 & 0.959 & 0.047 & 0.003 & 0.961 \\
\bottomrule
\end{tabular*}
\smallskip
\parbox{\textwidth}{\footnotesize\textit{Notes:} The reported numbers are averaged over all slope coefficients and categories within each DGP. 95\% confidence intervals are computed using analytic standard errors developed in Section~\ref{sec: asymptotics} for MLogit, Sieve, and Exch-$L2$, and heteroskedasticity-robust standard errors for OLS and Oracle.}
\end{table}

Table~\ref{tab: multinomial results} indicates that parsimonious specifications with limited continuous covariate variation are practically useful even when they are misspecified. Both DGPs have only three continuous covariates with $K=3$ so the sufficient condition for identification in Proposition~\ref{prop: derivative rank} is not satisfied. Nevertheless, different covariate values generate the same choice probabilities, leaving enough variation to identify the coefficients on all covariates, as verified in Appendix~\ref{sec: simulation rank proofs}. OLS suffers from substantial selection bias, while Oracle has negligible bias and the lowest RMSE. MLogit remains close to Oracle in DGP1, but it shows substantial bias in DGP2 due to misspecification. Exch-$L2$ greatly reduces this bias and has much lower RMSE than Sieve. A spline in just $e_1$ and $e_2$ yields low bias and coverage close to nominal in both exchangeable designs, including DGP2, where IIA and the own-shock restriction both fail.

In the nonexchangeable design (DGP3), MLogit has the lowest RMSE, and a variant of Exch-$L2$ based on log odds has the smallest average absolute bias (Appendix~\ref{sec: appendix approximation sensitivity}).

\section{Application: Entry-level gender wage gaps in Korea}\label{sec: application}

I study gender wage gaps among graduates aged 35 or younger in the 2008--2019 cohorts of South Korea's Graduates Occupational Mobility Survey (GOMS). Selection into employment and sorting across jobs can both affect measured wage differences \citep{mulligan2008selection, blau2017gender, kim2026gender}. I compare small and medium-sized enterprises (SMEs) with large firms, non-STEM with STEM occupations, and public with private employment. The selected samples comprise wage workers with observed monthly earnings, excluding part-time workers. The outside category includes other employment states and respondents without observed monthly earnings. I divide monthly earnings by monthly hours (4.345 times total weekly hours, including overtime) to obtain hourly wages.

The predetermined covariates are age, gender, college GPA, parental income, four-year university attendance, major, university founding type, school region, and cohort indicators. Age (in years, measured to the month), GPA (scaled to 0--100), and parental income enter the choice equations through splines and the wage equations linearly. I convert the midpoints of the eight parental income brackets reported at college entry or transfer in each cohort to 2015 prices using the cohort's consumer price index. Table~\ref{tab:summary_stats} reports sample sizes and means.

\begin{table}[htbp]
\centering
\caption{Analytical samples and sample means}\label{tab:summary_stats}
\small
\setlength{\tabcolsep}{3pt}
\begin{tabular*}{\textwidth}{@{\extracolsep{\fill}}lcccccc@{}}
\toprule
 & \multicolumn{2}{c}{Firm size} & \multicolumn{2}{c}{Field} & \multicolumn{2}{c}{Sector} \\
\cmidrule(lr){2-3}\cmidrule(lr){4-5}\cmidrule(lr){6-7}
 & SME & Large firm & Non-STEM & STEM & Public & Private \\
\midrule
Female (\%) & 49.3 & 37.2 & 47.1 & 37.5 & 50.4 & 42.8 \\
Age & 26.0 & 26.6 & 26.2 & 26.4 & 26.1 & 26.3 \\
GPA (0--100) & 81.4 & 82.2 & 81.8 & 81.5 & 82.9 & 81.4 \\
Parents' income (10K KRW) & 425 & 465 & 443 & 438 & 447 & 440 \\
Four-year university (\%) & 70.2 & 80.6 & 74.9 & 73.8 & 88.6 & 70.5 \\
Log hourly wage & -0.050 & 0.141 & 0.011 & 0.082 & -0.003 & 0.042 \\
\midrule
Hourly-wage observations & 67,397 & 48,241 & 86,802 & 34,156 & 25,861 & 94,042 \\
Total sample & \multicolumn{2}{c}{184,088} & \multicolumn{2}{c}{189,416} & \multicolumn{2}{c}{188,356} \\
\bottomrule
\end{tabular*}
\smallskip
\parbox{0.98\textwidth}{\footnotesize\textit{Notes:} Covariate means use selected wage workers; hourly-wage means and counts use those with recorded hours. Total samples include the outside category. Parental income is the CPI-adjusted bracket midpoint.}
\end{table}

For firm-size sorting, I employ the nonparametric ordered model in \eqref{eq: nonparametric ordered} where $D=1$ denotes employment in an SME and $D=2$ employment in a large firm. For field sorting, I use the general specification \eqref{eq: partial linear multinomial general}, with the full choice probability vector and categories denoting non-STEM and STEM jobs.\footnote{STEM jobs comprise research, engineering, health, and medical occupations.} For sector sorting, they denote public- and private-sector employment.\footnote{The public category comprises government agencies, government-invested organizations, and educational institutions; the private category comprises private companies, foreign companies, and corporate bodies.} On each margin, $D=0$ is the outside category.

The selection probabilities are estimated separately for women and men, allowing the effects of the observed characteristics on selection to vary by gender. For firm-size sorting, two penalized sieve probit regressions (with fREML) estimate the wage-employment probability $q(x)=P(D\geq1\mid X=x)$ and the large-firm probability among wage employees $r(x)=P(D=2\mid D\geq1,X=x)$. The probabilities of the outside category, SME employment, and large-firm employment are obtained as
\[
\hat p_0(x)=1-\hat q(x),\qquad
\hat p_1(x)=\hat q(x)[1-\hat r(x)],\qquad
\hat p_2(x)=\hat q(x)\hat r(x).
\]
Both regressions use cubic spline marginals and pairwise interactions. For occupational and sectoral sorting, I use a penalized MNL sieve with the same spline terms. Appendix Table~\ref{tab: overlap} reports the distribution of the estimated choice probabilities by category and sex.

Within each category, women and men share the same selection correction function in the wage equation. Its arguments are the estimated cumulative probabilities for firm size and the two nonbaseline category probabilities for field and sector. Each probability enters through six cubic B-spline basis functions. When two probabilities enter, their interaction is approximated using a smaller basis of four functions per probability, yielding 16 product terms. I also compute Restricted Sieve for firm-size sorting and MLogit for field and sector sorting as implemented in the simulations. I use the analytic two-stage standard errors for inference.\footnote{With predetermined wage covariates, simulations calibrated to each margin and sample size give 93.6--97.0\% coverage for the nominal 95\% Sieve intervals on the pooled female coefficients across both outcomes (Appendix Table~\ref{tab: pooled analytic coverage}).}

\begin{table}[htbp]
\centering
\caption{Female coefficients in log hourly-wage regressions}\label{tab:final_main_hourly}
\small
\setlength{\tabcolsep}{3pt}
\begin{tabular*}{\textwidth}{@{\extracolsep{\fill}}lcclcccc@{}}
\toprule
Estimator & \multicolumn{2}{c}{Firm size} & Estimator & \multicolumn{2}{c}{Field} & \multicolumn{2}{c}{Sector} \\
\cmidrule(lr){2-3}\cmidrule(lr){5-6}\cmidrule(lr){7-8}
 & SME & Large firm &  & Non-STEM & STEM & Public & Private \\
\midrule
OLS & -0.043 & -0.040 & OLS & -0.066 & -0.020 & 0.004 & -0.079 \\
 & (0.003) & (0.005) &  & (0.003) & (0.005) & (0.006) & (0.003) \\
\addlinespace[2pt]
Sieve & -0.036 & 0.008 & Sieve & -0.058 & -0.016 & -0.011 & -0.078 \\
 & (0.005) & (0.007) &  & (0.004) & (0.007) & (0.007) & (0.004) \\
\addlinespace[2pt]
Restricted Sieve & -0.040 & 0.008 & MLogit & -0.064 & -0.021 & -0.002 & -0.080 \\
 & (0.004) & (0.007) &  & (0.004) & (0.006) & (0.006) & (0.003) \\
\bottomrule
\end{tabular*}
\smallskip
\parbox{0.98\textwidth}{\footnotesize\textit{Notes:} Parentheses contain heteroskedasticity-robust standard errors for OLS and analytic two-stage standard errors for the selection-corrected estimators.}
\end{table}

In Table~\ref{tab:final_main_hourly}, selection correction narrows the large-firm hourly wage gap: the female coefficient rises from $-0.040$ under OLS to $0.008$ under both Sieve and Restricted Sieve. The SME coefficient remains near $-0.040$. For field and sector, both Sieve and MLogit leave the OLS gaps nearly unchanged. The hourly Sieve estimates change little as I vary the spline bases in the selection and wage equations. For large-firm hourly wages, the female coefficient's 95\% interval still includes zero (Appendix Table~\ref{tab: firststage grid}).\footnote{Within each category, substantial variation in the female indicator remains after projection on the other wage covariates and the baseline spline basis (Appendix Table~\ref{tab: rank diagnostics}).}

The monthly wage regressions control separately for regular and overtime hours per week. As with hourly wages, gaps are smaller in large firms than in SMEs, in STEM than in non-STEM jobs, and in public than in private employment. Selection correction narrows the large-firm monthly gap, though women still earn less (Appendix Table~\ref{tab:final_main_monthly}).\footnote{Appendix Tables~\ref{tab:final_full_coefficients_hourly} and~\ref{tab:final_full_coefficients_monthly} report the other wage coefficients.}

\subsection{The contribution of sorting to the gender wage gap}

To assess how participation and category sorting contribute to the gender wage gap, I distinguish differences in where men and women work from wage gaps within those categories \citep{brown1980incorporating}. I define the sorting contribution as the sum of the selection-correction and employment-composition components in the following accounting decomposition. Let $\bar w_k^g$ be the mean log wage of gender $g$ in category $k$, $s_k^g$ its category share among selected workers. The raw wage gap $\bar w^M-\bar w^F$ is decomposed into the following four components:
\begin{align*}
\underbrace{\sum_k s_k^M(-\hat\beta_k^{Sieve})}_{\text{unexplained gap}}
+\underbrace{\sum_k s_k^M\bigl[(\bar w_k^M-\bar w_k^F)+\hat\beta_k^{OLS}\bigr]}_{\text{covariate contribution}}
+
\underbrace{\sum_k s_k^M(\hat\beta_k^{Sieve}-\hat\beta_k^{OLS})}_{\text{selection correction}}
+\underbrace{\sum_k\bar w_k^F(s_k^M-s_k^F)}_{\text{employment composition}}.
\end{align*}
The unexplained gap remains after adjusting for observed covariates and selection. The covariate contribution reflects differences in observed characteristics within categories. The sorting contribution combines the selection correction with the contribution of differences in employment shares.

Table~\ref{tab: decomposition} shows that firm-size sorting accounts for 34.1\% of the raw hourly-wage gap: 20.9\% through selection and 13.2\% through employment composition. The composition component is positive because men are more likely to work in large firms, where women's mean hourly wages are higher than in SMEs. Field sorting contributes 10.8\% of the raw gap, split almost equally between selection (5.7\%) and employment composition (5.0\%). Women are more likely than men to work in non-STEM jobs, where their mean wages are lower. This concentration widens the gender gap. Public--private sorting contributes $-3.2\%$. Its composition component is small and negative because women are more concentrated in public employment, where their mean wages are slightly higher. Firm-size sorting also contributes most to the raw monthly earnings gap (24.1\%, Appendix Table~\ref{tab: decomposition monthly}).

\begin{table}[htbp]
\centering
\caption{Accounting decomposition of the gender gap in log hourly wages}\label{tab: decomposition}
\small
\setlength{\tabcolsep}{3pt}
\begin{tabular*}{\textwidth}{@{\extracolsep{\fill}}l@{\hspace{1em}}cccccc@{}}
\toprule
\shortstack{Selection\\structure} & \shortstack{Raw\\gap} & \shortstack{Gap after\\covariate\\adjustment\\(OLS)} & \shortstack{Covariate\\contribution} & \shortstack{Employment\\composition} & \shortstack{Unexplained\\gap\\(Sieve)} & \shortstack{Selection\\correction} \\
\midrule
Firm size & +0.126 & +0.042 & +0.068 & +0.017 & +0.015 & +0.026 \\
 & (0.002) & (0.003) & (0.002) & (0.0006) & (0.004) & (0.003) \\
Field & +0.122 & +0.051 & +0.064 & +0.006 & +0.044 & +0.007 \\
 & (0.002) & (0.003) & (0.002) & (0.0004) & (0.004) & (0.003) \\
Sector & +0.123 & +0.063 & +0.061 & -0.002 & +0.065 & -0.002 \\
 & (0.002) & (0.003) & (0.002) & (0.0002) & (0.003) & (0.002) \\
\bottomrule
\end{tabular*}
\smallskip
\parbox{0.98\textwidth}{\footnotesize\textit{Notes:} Selection correction is the OLS minus Sieve gap after adjustment. The sorting contribution is the sum of selection correction and employment composition. Parentheses contain analytic standard errors accounting jointly for estimated shares, means, and coefficients.}
\end{table}

\subsection{Changes across graduating cohorts}

I re-estimate the selection and wage equations separately for the 2008--2011, 2012--2015, and 2016--2019 cohorts. Table~\ref{tab: dynamics results} reports the hourly OLS and Sieve estimates. The STEM gap narrows for both outcomes and nearly disappears for hourly wages in the latest cohorts (Appendix Table~\ref{tab: dynamics monthly} gives the monthly estimates). The non-STEM gap narrows but remains substantial. Gaps in SMEs and private employment change little, and the large-firm estimates stay close to zero. The public-sector disadvantage is larger in the last period.\footnote{Appendix Figures~\ref{fig:yearly_hourly} and~\ref{fig:yearly_monthly} show the annual OLS and Sieve estimates with pointwise 95\% confidence intervals.}

\begin{table}[htbp]
\centering
\caption{Female log hourly-wage coefficients by subperiod}\label{tab: dynamics results}
\footnotesize
\setlength{\tabcolsep}{3pt}
\begin{tabular*}{\textwidth}{@{\extracolsep{\fill}}lcccccc@{}}
\toprule
 & \multicolumn{2}{c}{Firm size} & \multicolumn{2}{c}{Field} & \multicolumn{2}{c}{Sector} \\
\cmidrule(lr){2-3}\cmidrule(lr){4-5}\cmidrule(lr){6-7}
 & SME & Large firm & Non-STEM & STEM & Public & Private \\
\midrule
\multicolumn{7}{l}{\textit{OLS}} \\
2008--2011 & -0.043 (0.006) & -0.052 (0.009) & -0.068 (0.006) & -0.038 (0.010) & +0.011 (0.011) & -0.077 (0.006) \\
2012--2015 & -0.047 (0.006) & -0.022 (0.008) & -0.061 (0.006) & -0.019 (0.010) & +0.016 (0.010) & -0.080 (0.006) \\
2016--2019 & -0.035 (0.005) & -0.041 (0.008) & -0.064 (0.005) & -0.005 (0.008) & -0.012 (0.008) & -0.075 (0.005) \\
\addlinespace
\multicolumn{7}{l}{\textit{Sieve}} \\
2008--2011 & -0.028 (0.008) & -0.009 (0.011) & -0.057 (0.008) & -0.047 (0.015) & +0.015 (0.013) & -0.070 (0.007) \\
2012--2015 & -0.048 (0.010) & +0.018 (0.012) & -0.069 (0.008) & -0.025 (0.013) & -0.007 (0.015) & -0.091 (0.007) \\
2016--2019 & -0.030 (0.008) & +0.004 (0.013) & -0.045 (0.007) & +0.005 (0.012) & -0.025 (0.010) & -0.072 (0.006) \\
\bottomrule
\end{tabular*}
\smallskip
\parbox{0.98\textwidth}{\footnotesize\textit{Notes:} Each subperiod is estimated separately. Parentheses contain heteroskedasticity-robust standard errors for OLS and analytic two-stage standard errors for Sieve.}
\end{table}

\section{Conclusion}\label{sec: conclusion}

I establish when semiparametric models of participation and sorting identify outcome coefficients. In ordered models, averaging one univariate function over each category's selection intervals can ease identification. Under exchangeability in multinomial models, I approximate selection bias as a function of low-order symmetric polynomials of the choice probabilities. The simulations demonstrate that exploiting this structure can improve precision. The empirical application shows that selection matters differently across job margins. Under the ordered model, the estimated gender gap in hourly wages at large firms largely disappears once selection is accounted for. For field and sector, selection correction changes the estimates little, and most of the observed gap remains within these broad categories. Modeling sorting together with participation is important for learning where selection matters.

\clearpage
\appendix
\setcounter{footnote}{0}
\counterwithin{table}{section}
\counterwithin{figure}{section}
\counterwithin{prop}{section}
\counterwithin{corollary}{section}
\counterwithin{assumption}{section}
\counterwithin{theorem}{section}
\counterwithin{lemma}{section}
\counterwithin{remark}{section}
\spacingset{1}

\null\vskip 2em
\begin{center}
  {\LARGE\bfseries Online Supplementary Appendix for ``Identification and Estimation of Semiparametric Multilayered Sample Selection Models''\par}
  \vskip 1.5em
  {\large\lineskip .5em\begin{tabular}[t]{c}
  Dongwoo Kim\hspace{0.2cm}\\
  Simon Fraser University and Korea University
  \end{tabular}\par}
  \vskip 1em
  {\large\today}
\end{center}
\par\vskip 1.5em
\spacingset{1.8}

\section{Technical proofs}\label{sec: appendix proofs}

\subsection{A derivative sufficient condition for residual rank}\label{sec: derivative rank details}
\begin{proof}[Proof of Proposition~\ref{prop: derivative rank}]
Conditional on $D_i=k$, suppress $k$, and write $\Sigma(T)=E[\operatorname{Var}(X_i\mid T)]$. If $\Sigma(T)$ is singular, there is a nonzero $\delta=(\delta_C',\delta_W')'$ with $\operatorname{Var}(X_i\delta\mid T)=0$ almost surely. Hence, for a measurable $r$,
\begin{equation}\label{eq: derivative rank equality}
C_i\delta_C+W_i\delta_W=r(T(C_i,W_i))
\end{equation}
almost surely. Consider a value $w$ for which this equality and conditions (i)--(ii) hold. Positive conditional density implies that \eqref{eq: derivative rank equality} holds Lebesgue almost everywhere on $\Omega_w$.

At any point of $\Omega_w$, the constant rank theorem \citep{krantz2013implicit} gives local coordinates $(t,z)\in\mathbb R^q\times\mathbb R$ and a continuously differentiable inverse coordinate map $c=\chi_w(t,z)$ such that $T(\chi_w(t,z),w)=t$. In a sufficiently small rectangular coordinate neighborhood, \eqref{eq: derivative rank equality} becomes
\[
\chi_w(t,z)\delta_C+w\delta_W=r(t)
\]
almost everywhere. Changes of coordinates preserve Lebesgue null sets locally. By Fubini's theorem, for almost every $t$ the left side is constant for almost every $z$. Its continuity makes it constant for every $z$ in that interval. Differentiation in $z$, followed by continuity of that derivative in $(t,z)$, gives
$(\partial\chi_w(t,z)/\partial z)\delta_C=0$ everywhere in the neighborhood. These derivatives span the one-dimensional null space of $J(c,w)'$. Hence $N(C_i,W_i)\delta_C=0$ almost surely. Condition (iii) now gives
\[
0=\delta_C'E[N(X_i)\mid D_i=k]\delta_C,
\]
so $\delta_C=0$.

If $W_i$ is absent this already contradicts $\delta\ne0$. Otherwise, \eqref{eq: derivative rank equality} reduces to
$r(T(C_i,w))=w\delta_W$ almost surely conditional on $w$. The same local coordinates, the positivity of the conditional density, and Fubini's theorem imply
$r(t)=w\delta_W$ for Lebesgue almost every $t$ in the open image $T(\Omega_w,w)$. By condition (iv), all these images contain $\mathcal O_*$. Comparing any two values of $w$ in the full-probability set on this common open set gives the same constant, say $c_0$, for $w\delta_W$. Thus $X_i\delta=c_0$ almost surely on the selected sample. This contradicts condition (v) for nonzero $\delta$. Therefore $\Sigma(T)$ is positive definite.
\end{proof}

\subsection{Structural ordered restriction}\label{sec: structural proofs}

Condition on $D=k$ and suppress the category subscript. Let $P_k$ and $E_k$ denote the conditional distribution and expectation, write $X$ as a column vector, and put $H=(s,t)=(h_k(X),h_{k+1}(X))$ and $w=t-s$. For the top category, set $t=1$. Assume $w\geq\underline w>0$ and $E_k[\|X\|^2+Y^2]<\infty$. For $a\in L^2_0[0,1]:=\{f\in L^2[0,1]:\int_0^1 f(u)\,du=0\}$, write the conditional mean in \eqref{eq: nonparametric bias ordered} as
\begin{equation}\label{eq: structural ordered mean}
m(X)=\alpha+X'\beta+(\mathcal K a)(H),
\qquad
(\mathcal K a)(s,t)=\frac{1}{t-s}\int_s^t a(u)\,du.
\end{equation}
Define the linear spaces
\[
 \mathcal S_0=\{c+(\mathcal K a)(H):c\in\mathbb R,\ a\in L^2_0[0,1]\},
 \qquad \mathcal S=\overline{\mathcal S_0}^{\,L^2(P_k)}.
\]
With projection applied coordinatewise \citep[Section~4.1]{chen2014local}, define
\begin{equation}\label{eq: structural residual rank}
R^{\rm str}=X-\operatorname{Proj}_{\mathcal S}X,
\qquad
\Sigma^{\rm str}=E_k[R^{\rm str}R^{{\rm str}\prime}].
\end{equation}
For comparison, define $\Sigma^{\rm plm}=E_k[(X-E_k[X\mid H])(X-E_k[X\mid H])']$.

\begin{prop}\label{prop: structural mean rank}
The identified set of $\beta$ under the conditional-mean restriction \eqref{eq: structural ordered mean} is
\begin{equation}\label{eq: exact structural mean set}
 \mathcal B^{\rm mean}(m)=\{\beta:m-X'\beta\in\mathcal S_0\}.
\end{equation}
If $(\alpha_0,\beta_0,a_0)$ satisfies \eqref{eq: structural ordered mean}, this set is
$\beta_0+\{\delta:X'\delta\in\mathcal S_0\}$. The parameter $\beta$ is point identified if
\begin{equation}\label{eq: structural absorption}
 (t-s)(c+X'\delta)=B(t)-B(s)\qquad P_k\text{-almost surely}
\end{equation}
implies $\delta=0$ for every $c\in\mathbb R$ and every absolutely continuous $B$ satisfying $B(0)=B(1)=0$ and $B'\in L^2[0,1]$. In particular, $\Sigma^{\rm str}\succ0$ is sufficient for identification of $\beta$ under \eqref{eq: structural ordered mean}.

Replacing $\mathcal S_0$ by $\mathcal S$ gives the identified set $\beta_0+\{\delta:\Sigma^{\rm str}\delta=0\}$, which is a singleton if and only if $\Sigma^{\rm str}\succ0$. When this condition holds,
\begin{equation}\label{eq: structural beta formula}
 \beta_0=(\Sigma^{\rm str})^{-1}E_k[R^{\rm str}Y].
\end{equation}
The residual covariance matrices satisfy $\Sigma^{\rm str}\succeq\Sigma^{\rm plm}$. For the top category, $\Sigma^{\rm str}=\Sigma^{\rm plm}$.
\end{prop}

\begin{proof}[Proof of Proposition~\ref{prop: structural mean rank}]
Cauchy--Schwarz gives
$|(\mathcal K a)(s,t)|^2\leq(t-s)^{-1}\int_s^ta(u)^2du\leq\underline w^{-1}\|a\|_2^2$, so $\mathcal S_0$ is a well-defined linear subspace. The exact mean set consists precisely of slopes for which the remaining conditional mean belongs to that space. If $m=X'\beta_0+s_0$ with $s_0\in\mathcal S_0$, then $m-X'(\beta_0+\delta)\in\mathcal S_0$ if and only if $X'\delta\in\mathcal S_0$, because the space is linear. Writing its generating function as an integral proves the characterization by \eqref{eq: structural absorption}, with the sign of the unrestricted constant absorbed into $c$. Every absolutely continuous $B$ with the stated endpoints and derivative is generated by a member of $L^2_0[0,1]$.

The squared distance from $X'\delta$ to $\mathcal S$ is
\[
 \inf_{v\in\mathcal S}E_k[(X'\delta-v)^2]
 =E_k[(R^{{\rm str}\prime}\delta)^2]
 =\delta'\Sigma^{\rm str}\delta.
\]
Thus $X'\delta\in\mathcal S$ is equivalent to $\Sigma^{\rm str}\delta=0$. Positive rank forces exact uniqueness since $\mathcal S_0\subseteq\mathcal S$. Orthogonality gives $E_k[R^{\rm str}]=0$, $E_k[R^{\rm str}\mathcal K a_0]=0$, and $E_k[R^{\rm str}X']=\Sigma^{\rm str}$. Multiplying the conditional mean by $R^{\rm str}$ proves \eqref{eq: structural beta formula}. Finally $\mathcal S\subseteq L^2(H,P_k)$, a closed subspace of $L^2(P_k)$. Minimizing the squared distance to the larger space cannot increase it, which proves $\Sigma^{\rm str}\succeq\Sigma^{\rm plm}$ for every quadratic form.

To verify that $\Sigma^{\rm str}\succ0$ can hold when $\Sigma^{\rm plm}=0$, take $X=(s^2,t^2)'$, with $(s,t)$ uniform on a nondegenerate rectangle $I_s\times I_t\subset(0,1)^2$ satisfying $\sup I_s<\inf I_t$, independently of $U\sim\operatorname{Unif}[0,1]$. Selection into category $k$ occurs when $s\leq U<t$. Then $H$ determines $X$, so $\Sigma^{\rm plm}=0$, and the selected density is proportional to $t-s$.

Multiplication by $t-s$ maps $\mathcal S_0$ into additive functions of $s$ and $t$. These form a closed subspace under product Lebesgue measure: the orthogonal projection is the sum of the two coordinate averages minus the overall average. Since the selected density and the width are bounded above and away from zero, $(t-s)v$ is therefore additive for every $v\in\mathcal S$. Consequently, $X'\delta\in\mathcal S$ requires
\[
 0=\frac{\partial^2}{\partial s\,\partial t}
       \bigl[(t-s)(\delta_1s^2+\delta_2t^2)\bigr]
   =2\delta_1s-2\delta_2t.
\]
Varying $s$ and $t$ over the rectangle forces $\delta=0$, so the rank characterization above gives $\Sigma^{\rm str}\succ0$.
\end{proof}

\subsection{Multinomial Logit selection}\label{sec: gumbel proofs}

\begin{lemma}[Gumbel class]\label{lemma: gumbel class}
Let $G$ be standard Gumbel and let $\mathcal B_G$ denote the set of selection-bias functions $b(v)=E[\mu(G+v)]$, $v\geq0$, generated by measurable $\mu$ satisfying $E|\mu(G)|<\infty$ and $E[\mu(G)]=0$. Every $b\in\mathcal B_G$ is real analytic on $(0,\infty)$ and satisfies $b(0+)=0$. Its values on any nonempty open interval in $(0,\infty)$ uniquely determine $\mu$ almost everywhere under the Gumbel distribution.
\end{lemma}

\begin{proof}[Proof of Lemma~\ref{lemma: gumbel class}]
With $r(z)=\mu(-\log z)$, the change of variables $z=e^{-e}$ in the Gumbel density gives
\begin{equation}\label{eq: gumbel laplace class}
 b(\log q)=q\int_0^\infty r(z)e^{-qz}\,dz,\quad q\geq1,
 \qquad
 \int_0^\infty |r(z)|e^{-z}\,dz<\infty,\quad
 \int_0^\infty r(z)e^{-z}\,dz=0.
\end{equation}
Allow $q$ to be complex. The Laplace integral $L_r(q)=\int r(z)e^{-qz}dz$ is holomorphic for
$\operatorname{Re}q>1$: on a compact subset of that half-plane, each
derivative is dominated by a constant times $|r(z)|e^{-z}$.
Composition with $q=e^v$ proves real analyticity for $v>0$.
For $1\leq q\leq2$, $q|r(z)|e^{-qz}\leq2|r(z)|e^{-z}$.
Dominated convergence and centering give $b(0+)=0$.

If two such functions coincide on an open interval, analyticity gives equality for all $v>0$. Let $r$ now denote the difference of their transformed latent means. The finite signed measure $dM(z)=r(z)e^{-z}dz$ then has zero Laplace transform. Uniqueness of Laplace transforms, applied to its positive and negative parts, gives $M=0$ \citep[Proposition~1.2]{schilling2010bernstein}. Hence the two latent means agree almost everywhere under the Gumbel distribution.
\end{proof}

\begin{proof}[Proof of Proposition~\ref{prop: gumbel point identification}]
The log odds $\log(p_j(x)/p_0(x))=a_j+x'\gamma_j$ identify the utility
coefficients on the open support under the maintained location and scale
normalizations. All choice probabilities are positive at finite $x$, so
the conditional distribution of $X_i$ given $D_i=k$ has the same support as the unconditional distribution of $X_i$.
Write the difference of two candidate conditional means as
\[
 F(x)=c+x'\delta+h(\nu_k(x)),\qquad h\in\mathcal B_G.
\]
It is zero almost surely conditional on $D_i=k$. Continuity makes it zero
on the open set in the support of $X_i$. By Lemma~\ref{lemma: gumbel class}, $h$ is real analytic on $(0,\infty)$. Since $\nu_k(x)>0$ for every $x\in\mathbb R^{d_X}$, $F$ is real analytic on the connected space $\mathbb R^{d_X}$.
The identity theorem yields $F=0$ throughout that space.

Fix any direction $a$ and point $x_0$, and put $x=x_0+ta$.
Along the resulting line,
\[
 \nu_k(x_0+ta)=\log\sum_{j=0}^K\exp(A_j+B_jt),\quad
 A_j=u_j(x_0)-u_k(x_0),\quad B_j=a'(\gamma_j-\gamma_k),
\]
with $A_k=B_k=0$.
If the $B_j$ include both positive and negative values, this index
diverges at both ends and has a finite interior minimum $t_*$. At that
point its derivative is zero. Differentiating $F(x_0+ta)=0$ gives
$a'\delta=0$ at $t_*$.
If the nonzero $B_j$ have the same sign, the index tends at one end to
\[
 L=\log\sum_{j:B_j=0}e^{A_j}\geq0.
\]
The function $h$ has a finite limit at $L$, by continuity if $L>0$ and
by $h(0+)=0$ if $L=0$. The affine term $ta'\delta$ can therefore satisfy
$F=0$ on the line only if $a'\delta=0$.
If all $B_j=0$, the index is constant and the same conclusion follows
directly. Since the direction was arbitrary, $\delta=0$.

If $\nu_k$ is nonconstant, its image contains an open interval. On this
interval $h=-c$, hence the identity theorem makes $h$ constant for all
$v>0$. Its limit at zero gives $c=0$, and uniqueness of the transform in Lemma~\ref{lemma: gumbel class} implies that the two latent mean functions agree almost surely under the Gumbel distribution.
\end{proof}

\begin{samepage}
\begin{corollary}[Continuous and discrete covariates]\label{cor: gumbel mixed identification}
Write $X_i=(C_i',W_i')'$, with $W_i$ discrete, and maintain a common $\mu_k$
across its values. Suppose utilities are affine in $C_i$ within each
$W_i$ stratum and the outcome mean is
$\alpha_k+C_i'\beta_{Ck}+W_i'\beta_{Wk}$ before selection. If some stratum
$w_0$ of positive probability has an open set in the support of
$C_i\mid W_i=w_0$ on which $\nu_k(C_i,w_0)$ is nonconstant, and no
nonzero affine function of $W_i$ vanishes almost surely, then
$\alpha_k,\beta_{Ck},\beta_{Wk}$ and $\mu_k$ are identified.
\end{corollary}

\begin{proof}[Proof of Corollary~\ref{cor: gumbel mixed identification}]
Within $w_0$, the line argument in the preceding proof identifies
$\beta_{Ck}$. Nonconstancy of $\nu_k(C,w_0)$ supplies an open interval
on which the difference of the bias functions is constant. Analyticity
and the zero boundary limit make that difference zero everywhere.
The same bias function applies to every stratum. The remaining equality
is $\Delta\alpha_k+W_i\Delta\beta_{Wk}=0$ almost surely. The affine-rank
condition on $W_i$ identifies both terms. Transform uniqueness identifies
$\mu_k$.
\end{proof}
\end{samepage}

\subsection{Bias from omitted symmetric coordinates}\label{sec: symmetry approximation details}

Write $E_L=E_L(X_i)$, $\tilde\lambda_k=\tilde\lambda_k(p(X_i))$, and let $E_k$ denote expectation conditional on $D_i=k$. Assume $Y_i$, $X_i$, and $\tilde\lambda_k$ have finite second moments. For fixed $L$, let $\beta_{k,L}$ be the coefficient on $X_i$ in the population regression of $Y_i$ on $X_i$ and all square-integrable functions of $E_L$. Define $R_L=X_i-E_k[X_i\mid E_L]$, $\Sigma_L=E_k[R_LR_L']\succ0$, and $\rho_L^2=E_k[(\tilde\lambda_k-E_k[\tilde\lambda_k\mid E_L])^2]$. Partial regression and Cauchy--Schwarz give, for every vector $a$,
\begin{equation}\label{eq: truncation bias bound}
 |a'(\beta_{k,L}-\beta_{k0})|
 =|a'\Sigma_L^{-1}E_k[R_L\{\tilde\lambda_k-E_k[\tilde\lambda_k\mid E_L]\}]|
 \leq(a'\Sigma_L^{-1}a)^{1/2}\rho_L.
\end{equation}
The bound relates coefficient bias to the part of the selection correction function that $E_L$ cannot capture and to the covariate variation remaining after conditioning on $E_L$. Even under exchangeability, inference on $\beta_{k0}$ requires the bias from omitted symmetric coordinates to be negligible relative to the standard error.

\subsection{Two-step sieve estimation and inference}\label{sec: estimation details}

\subsubsection{Basis construction and first-stage estimation}

For a selection correction function with $L$ arguments, let $B^{[\ell]}_j$ be the univariate B-splines of order $r$ ($r=4$ for cubic splines) with $J_n^{[\ell]}$ interior knots in argument $\ell$. The full tensor basis is
\begin{equation}\label{eq: tensor product basis}
 B_n(s)=\left(\prod_{\ell=1}^L B^{[\ell]}_{j_\ell}(s_\ell)\right)_{j_1,\ldots,j_L},
 \qquad \kappa_n=\prod_{\ell=1}^L(J_n^{[\ell]}+r).
\end{equation}
Linearly dependent basis functions are removed before estimation.

Let $\psi_n$ denote the basis for one selection index and $\Psi_n$ the stacked basis across equations. For the common-index ordered model with known $F_\varepsilon$, the category probabilities are $F_\varepsilon(c_{j+1}-\psi_n(x)'\gamma)-F_\varepsilon(c_j-\psi_n(x)'\gamma)$. Penalized maximum likelihood estimates $\gamma$ jointly with the ordered thresholds, normalizing $c_1=0$ and penalizing only $\gamma$.
For the nonparametric threshold model, penalized binary regression with link $F_0$ gives
\begin{equation}\label{eq: sieve logistic}
 \begin{split}
 \hat a_j=\argmax_a\Bigg\{&\sum_{i=1}^n\Big[
 \mathbf1\{D_i\leq j-1\}\log F_0(\psi_n(X_i)'a)\\
 &+\mathbf1\{D_i>j-1\}\log\{1-F_0(\psi_n(X_i)'a)\}\Big]
 -\tfrac12a'S^{(j)}(\hat\lambda^{(j)})a\Bigg\},
 \end{split}
\end{equation}
where $S^{(j)}$ is the roughness penalty matrix for equation $j$ and $\lambda^{(j)}$ collects its smoothing parameters.
For multinomial selection, normalize $a_0=0$ and estimate the stacked coefficients $a=(a_1',\ldots,a_K')'$ by
\begin{equation}\label{eq: sieve MNL}
 \hat a=\argmax_a\left\{
 \sum_{i=1}^n\sum_{j=0}^K\mathbf1\{D_i=j\}
 \log\frac{\exp\{\psi_n(X_i)'a_j\}}{\sum_{\ell=0}^K\exp\{\psi_n(X_i)'a_\ell\}}
 -\tfrac12a'S(\hat\lambda)a\right\}.
\end{equation}

Suppose the range of $p$ lies in a compact set on which all $K+1$ choice probabilities are strictly positive, $\phi_k$ has bounded first and second derivatives on a neighborhood of that range, and $\|\hat p-p\|_\infty=O_p(\bar r_n)=o_p(1)$. Uniform Taylor expansion gives
\begin{equation}\label{eq: rate condition}
 \|\hat g_k-g_{k0}\|_\infty=O_p(\bar r_n),\qquad
 \|\hat g_k-g_{k0}-J_{\phi_k}(p)(\hat p-p)\|_\infty=O_p(\bar r_n^2).
\end{equation}

\subsubsection{Growing-sieve conditions and influence functions}

Maintain Assumption~\ref{ass: sampling} and the notation of Section~\ref{sec: asymptotics}. Write $Q_k=\pi_k\Sigma_k$, where $\Sigma_k=E[U_{ik}U_{ik}'\mid D_i=k]$, and $\|v\|_{n,k}^2=n^{-1}\sum_iT_{ik}\|v_i\|^2$.

\begin{assumption}[Approximation and generated-index remainders]\label{ass: smoothness}\label{ass: sieve}
On a fixed compact set that contains the true index range and, with probability tending to one, the estimated index range, $\lambda_{k0}$ is twice continuously differentiable with bounded derivatives. There exist spline approximants $\ell_{kn}=B_n'\delta_{kn}$ whose approximation bounds $a_n,b_n,H_n$, defined in Section~\ref{sec: asymptotics}, satisfy the conditions below.
The empirical residuals satisfy $c_n:=\|\hat U_k-U_k\|_{n,k}=o_p(1)$ and $\max_iT_{ik}\|\hat U_{ik}\|=O_p(1)$. With $r_n:=\|\hat g_k-g_{k0}\|_\infty=o_p(1)$ and $\varrho_n:=\|\hat g_k-g_{k0}\|_{n,k}$,
\begin{equation}\label{eq: joint remainder rates}
 a_n=o(1),\qquad
 \sqrt n\{c_na_n+b_n\varrho_n+H_n\varrho_n^2+c_n\varrho_n\}=o_p(1).
\end{equation}
In the category-$k$ sample, the second-stage regressors $X_i$ and $B_n(\hat g_k(X_i))$ are linearly independent with probability tending to one.
\end{assumption}

\begin{assumption}[First-stage functional linearization]\label{ass: first stage rate}\label{ass: first stage influence}
There is a mean-zero, square-integrable vector $\chi_{ik}$, measurable with respect to $(D_i,X_i)$, such that
\begin{equation}\label{eq: first stage functional expansion}
 \frac1n\sum_{i=1}^n T_{ik}U_{ik}
 \nabla\lambda_{k0}(g_{k0}(X_i))'
 \{\hat g_k(X_i)-g_{k0}(X_i)\}
 =\frac1n\sum_{i=1}^n\chi_{ik}+o_p(n^{-1/2}).
\end{equation}
\end{assumption}

\begin{theorem}[Two-step asymptotic expansion]\label{thm: asymptotic normality}
Under Assumptions~\ref{ass: sampling}, \ref{ass: sieve}, and \ref{ass: first stage influence}, the expansion \eqref{eq: asymptotic normality} holds with covariance \eqref{eq: full sample variance}. Equivalently, $\sqrt{n_k}(\hat\beta_k-\beta_{k0})$ has asymptotic variance $V_k=\pi_k\mathcal V_k$.
\end{theorem}
\begin{proof}[Proof of Theorem~\ref{thm: asymptotic normality}]\label{sec: proof asymptotic normality}
Fix $k$ and suppress its subscript. Put $T_i=\mathbf1\{D_i=k\}$ and $\hat Q=n^{-1}\sum_iT_i\hat U_i\hat U_i'$.
By $\|\hat U-U\|_{n,k}=o_p(1)$, boundedness of $U_i$, and the law of large numbers,
$\hat Q=Q+o_p(1)$. The Frisch--Waugh--Lovell identity is
\begin{equation}\label{eq: FWL feasible}
 \hat Q(\hat\beta-\beta_0)
 =\frac1n\sum_iT_i\hat U_i
 \{\varepsilon_i+\lambda_0(g_0(X_i))\}.
\end{equation}
The sample residuals $\hat U_i$ are orthogonal to the basis evaluated at the estimated controls, so
$\sum_iT_i\hat U_i\ell_n(\hat g(X_i))=0$ exactly. Put
$\Delta_i=\hat g(X_i)-g_0(X_i)$ and $A_i=\lambda_0(g_0(X_i))-\ell_n(g_0(X_i))$. Taylor expansion of $\ell_n$ gives
\begin{align*}
 \lambda_0(g_0(X_i))-\ell_n(\hat g(X_i))
 &=A_i-\nabla\lambda_0(g_0(X_i))'\Delta_i-R_i,\\
 |A_i|&\leq a_n,\qquad |R_i|\leq b_n\|\Delta_i\|+\tfrac12 H_n\|\Delta_i\|^2.
\end{align*}
Because $A_i$ depends on $X_i$ only through $g_0(X_i)$ and $E[U_i\mid g_0(X_i),D_i=k]=0$, $n^{-1}\sum_iT_iU_iA_i=O_p(a_nn^{-1/2})$, and Cauchy--Schwarz bounds $n^{-1}\sum_iT_i(\hat U_i-U_i)A_i$ by $c_na_n$. The bound on $\max_iT_i\|\hat U_i\|$ gives $n^{-1}\sum_iT_i\|\hat U_i\||R_i|=O_p(b_n\varrho_n+H_n\varrho_n^2)$, and replacing $\hat U_i$ by $U_i$ in the linear drift costs at most a constant times $c_n\varrho_n$. Multiplied by $\sqrt n$, all of these terms are $o_p(1)$.

For the outcome-error term, condition on the full selection data $\{(D_i,X_i)\}_{i=1}^n$. The residuals $\hat U_i$ and $U_i$ are then fixed. Independence across observations and the maintained conditional mean imply that the conditional mean of
$n^{-1/2}\sum_iT_i(\hat U_i-U_i)\varepsilon_i$ is zero and its conditional variance has trace bounded by a constant times $c_n^2=o_p(1)$. Conditional Chebyshev's inequality therefore gives
\begin{equation}\label{eq: outcome residual replacement}
 \frac1{\sqrt n}\sum_iT_i\hat U_i\varepsilon_i
 =\frac1{\sqrt n}\sum_iT_iU_i\varepsilon_i+o_p(1).
\end{equation}
Combining these steps with Assumption~\ref{ass: first stage influence} yields
\[
 \sqrt n\,\hat Q(\hat\beta-\beta_0)
 =\frac1{\sqrt n}\sum_i\{T_iU_i\varepsilon_i-\chi_i\}+o_p(1).
\]
Each summand has mean zero and finite second moments. Moreover,
$E[T_iU_i\varepsilon_i\chi_i']=0$, because $\chi_i$ is measurable with respect to $(D_i,X_i)$ and $E[\varepsilon_i\mid X_i,D_i=k]=0$. The i.i.d. central limit theorem and Slutsky's theorem establish \eqref{eq: asymptotic normality}--\eqref{eq: full sample variance}. Finally $n_k/n\to\pi_k$ gives the selected-sample normalization.
\end{proof}

For a regular finite-dimensional first stage, suppose $g_{k0}(x)=g_k(x;\alpha_0)$, with bounded first and second derivatives in $\alpha$ on the covariate support, and
\begin{equation}\label{eq: alpha influence}
 \sqrt n(\hat\alpha-\alpha_0)=\frac1{\sqrt n}\sum_{i=1}^n h_i+o_p(1),
 \qquad E[h_i]=0,\quad E\|h_i\|^2<\infty,
\end{equation}
where $h_i$ is measurable with respect to $(D_i,X_i)$. Taylor expansion in $\alpha$ and the law of large numbers give Assumption~\ref{ass: first stage influence} \citep[Theorem~6.1]{neweymcfadden1994large}, with
\begin{equation}\label{eq: score influence}
 \chi_{ik}=A_kh_i,\qquad
 A_k=E[T_{ik}U_{ik}\nabla\lambda_{k0}(g_{k0}(X_i))'
 \partial_\alpha g_k(X_i;\alpha_0)].
\end{equation}
For a correctly specified regular maximum-likelihood first stage, $h_i=\mathcal J^{-1}s_i$, where $s_i$ is the score and $\mathcal J=E[s_is_i']$. Under misspecification, use the first-stage estimator's influence function.
For this finite-dimensional case, $\varrho_n\leq r_n=O_p(n^{-1/2})$ reduces the joint remainder requirements to projection consistency, $a_n=o(1)$, $b_n=o(1)$, $H_n=o(\sqrt n)$, and $\sqrt n\,c_na_n=o_p(1)$.

For a nonparametric probability estimator in either ordered or multinomial selection, let $g_{k0}=\phi_k(p)$ and
$\zeta_k(x)=\nabla\lambda_{k0}(g_{k0}(x))'J_{\phi_k}(p(x))$.
Suppose $\sqrt n\|\hat p-p\|_\infty^2=o_p(1)$ and the first-stage probability estimator satisfies the weighted expansion
\begin{equation}\label{eq: probability functional expansion}
 \frac1n\sum_i T_{ik}U_{ik}\zeta_k(X_i)(\hat p(X_i)-p(X_i))
 =\frac1n\sum_i p_k(X_i)U_{ik}\zeta_k(X_i)(T_i-p(X_i))+o_p(n^{-1/2}),
\end{equation}
where $T_i=(\mathbf1\{D_i=1\},\ldots,\mathbf1\{D_i=K\})'$. The Taylor remainder in \eqref{eq: rate condition} is $o_p(n^{-1/2})$, so Assumption~\ref{ass: first stage influence} holds with
\begin{equation}\label{eq: probability influence}
 \chi_{ik}=p_k(X_i)U_{ik}\zeta_k(X_i)\{T_i-p(X_i)\}.
\end{equation}
Writing $\Omega_k=E[U_{ik}U_{ik}'\varepsilon_{ik}^2\mid D_i=k]$ and $V_p(x)=\operatorname{diag}(p(x))-p(x)p(x)'$, the selected-sample variance is
\begin{equation}\label{eq: selected sample variance}
 V_k=\Sigma_k^{-1}(\Omega_k+\Omega_{\chi,k})\Sigma_k^{-1},\qquad
 \Omega_{\chi,k}=E[p_k(X_i)U_{ik}U_{ik}'\zeta_k(X_i)V_p(X_i)\zeta_k(X_i)'\mid D_i=k].
\end{equation}
Conditional on $X_i$,
$E[T_i-p(X_i)\mid X_i]=0$ and
$E[(T_i-p(X_i))(T_i-p(X_i))'\mid X_i]=V_p(X_i)$.
Thus
\[
 E[\chi_{ik}\chi_{ik}']
 =E[p_k(X_i)^2U_{ik}U_{ik}'\zeta_k(X_i)V_p(X_i)\zeta_k(X_i)'].
\]
The identity $E[T_{ik}f(X_i)]=E[p_k(X_i)f(X_i)]$ shows that this matrix is $\pi_k\Omega_{\chi,k}$. Substituting $Q_k=\pi_k\Sigma_k$ into Theorem~\ref{thm: asymptotic normality} proves \eqref{eq: selected sample variance}.

Under Theorem~\ref{thm: asymptotic normality}, suppose the influences used in \eqref{eq: V hat} satisfy
\begin{align*}
 \frac1n\sum_i\|\hat\chi_{ik}-\chi_{ik}\|^2&=o_p(1),\\
 \frac1n\sum_iT_{ik}\|\hat U_{ik}\hat\varepsilon_{ik}-U_{ik}\varepsilon_{ik}\|^2&=o_p(1).
\end{align*}
The mean-square influence consistency argument of \citet[Lemma~8.3]{neweymcfadden1994large} gives $\hat{\mathcal V}_k\xrightarrow{p}\mathcal V_k$, and the estimated covariance of $\hat\beta_k$ is $\hat{\mathcal V}_k/n$. Equivalently, $\hat V_k=(n_k/n)\hat{\mathcal V}_k$ estimates the variance for the $\sqrt{n_k}$ normalization.

\subsubsection{Penalized first stages with data-selected smoothing}\label{sec: penalized first stage}

Maintain Assumptions~\ref{ass: sampling} and \ref{ass: penalized sieve}. Write $\hat\eta_\lambda=\Psi_n'\hat a_\lambda$, $P_nf=n^{-1}\sum_if(D_i,X_i)$, $Pf=E[f(D_i,X_i)]$, and $\|f\|_{L^1}=E|f(X_i)|$. Let $\eta_n^*=\Psi_n'a_n^*$ be the Kullback--Leibler projection of $\eta_0$ on the sieve. As in Section~\ref{sec: asymptotics}, $I_0(x)$ denotes the conditional Fisher information for the selection indices, evaluated at $\eta_0(x)$. For vector-valued functions $u,v$, define $\langle u,v\rangle_{I_0}=E[u(X_i)'I_0(X_i)v(X_i)]$. The projection in \eqref{eq: drift loading} is applied separately to each row using this information-weighted inner product.
Let $\mathcal F_F$ be the full tensor space and $\mathcal F_A$ the space of additive discrete effects, univariate smooths, and pairwise continuous interactions, with each smooth component normalized to have mean zero.

At an fREML fixed point, $\hat a$ maximizes \eqref{eq: penalized sieve likelihood} given $\hat\lambda$, which satisfies the Karush--Kuhn--Tucker conditions for the working-model restricted likelihood constructed at $\hat a$, subject to $\hat\lambda\in[\lambda_{\min},\lambda_{\max}]^M$. The Laplace REML criterion uses observed information. The observed information is the negative log-likelihood Hessian with respect to the selection indices:
\[
 I_{O,i}(\eta)=-\partial_\eta^2\mathcal L(D_i,\eta),\qquad
 \bar I(x;\eta)=E[I_{O,i}(\eta)\mid X_i=x].
\]
Here $\bar I$ averages over the true conditional distribution of $D_i$, so $\bar I(x;\eta_0)=I_0(x)$.

fREML uses the Fisher weights of the current model, denoted by $I_{W,i}$. Binary probit has $I_{W,i}(\eta_i)=w(\eta_i)$, where $w(v)=\phi(v)^2/[\Phi(v)\{1-\Phi(v)\}]$. Participation and conditional probit use
\[
 I_{W,i}(\eta_i)=\operatorname{diag}\{w(\eta_{i1}),R_{2i}w(\eta_{i2})\},\qquad
 R_{2i}=\mathbf1\{D_i\geq1\},
\]
so only observed participants contribute to the second equation. Joint multinomial logit has $I_{W,i}(\eta_i)=\operatorname{diag}(p_i)-p_ip_i'=I_{O,i}(\eta_i)$. The conditional mean of $I_{W,i}$ equals $I_0(x)$ at the true indices. The chain rule for $\eta_i=\Psi_n(X_i)'a$ gives the coefficient information matrices
\[
 \hat{\mathcal I}_O=P_n\Psi_n I_{O,i}(\hat\eta_i)\Psi_n',\qquad
 \hat{\mathcal I}_W=P_n\Psi_n I_{W,i}(\hat\eta_i)\Psi_n'.
\]
Thus $\hat{\mathcal I}_O$ is the negative Hessian of the average unpenalized log likelihood.

For separately estimated binary first-stage equations, $\chi_{ik}$ is the sum of the equation-specific influences, each computed using its own score and information matrix. The coefficient information matrix is block diagonal, but the covariance of $\chi_{ik}$ retains the covariances between equations estimated on the same observations. For cumulative probabilities, Assumption~\ref{ass: penalized sieve}(i) and uniform consistency ensure that the estimated probabilities are correctly ordered with probability tending to one, so rearrangement leaves them unchanged.

The proof of Theorem~\ref{thm: penalized two step} uses the following lemmas and propositions. Lemmas~\ref{lemma: sieve approximation energy}--\ref{lemma: smoothing bound} control spline approximation and the selected smoothing parameters. Proposition~\ref{prop: penalized expansion} gives the first-stage influence expansion, and Proposition~\ref{prop: residualized regressors} controls residualization using estimated selection controls. Lemma~\ref{lemma: penalized variance} establishes consistency of the analytic variance estimator.

\begin{lemma}[Approximation and bounded roughness]\label{lemma: sieve approximation energy}
Under Assumption~\ref{ass: penalized sieve}, the Kullback--Leibler projection satisfies
\[
 \|\eta_n^*-\eta_0\|_{L^2}=O(\omega_{2n}),\quad
 \|\eta_n^*-\eta_0\|_\infty=O(\sqrt{q_n}\omega_{2n}+\omega_{\infty n}),\quad
 \max_jJ_j(a_n^*)=O(1).
\]
Define $\mathcal I_n=E[\Psi_n\bar I(X_i;\eta_n^*(X_i))\Psi_n']$,
$\mathcal I_n^0=E[\Psi_nI_0\Psi_n']$, $\gamma_n=E[\Psi_n\Gamma_k']$, and
$v_n=(\mathcal I_n^{-1}\gamma_n)'\Psi_n$.
Then $\|v_n\|_\infty+\max_jJ_j(v_n)=O(1)$ and
$\|v_k^*-v_n\|_{I_0}=O(\omega_{vn}+\omega_{2n})$.
\end{lemma}

\begin{proof}[Proof of Lemma~\ref{lemma: sieve approximation energy}] For unconstrained splines, let $\mathcal P_n$ be the tensor product of stable univariate quasi-interpolants \citep[Theorems~3.1--3.5 and Section~4]{castillogarau2023}. For natural cubic splines, use the product of unweighted univariate $L^2$ projections: comparison with the natural cubic interpolant, whose second-derivative energy is minimal \citep{deBoor2001}, and an inverse inequality give the same energy bound. Apply these operators componentwise, with population centering where required; the component decomposition is stable under the density bounds \citep{stone1994use}. Stability in $L^2$ along the other axes and commutation with differentiation along a given axis yield, in any fixed component dimension,
\[
 J_j(\mathcal P_nf)\leq C\sum_lJ_l(f),\quad
 \|f-\mathcal P_nf\|_{L^2}\leq Cm_n^{-2}\sum_lJ_l(f),\quad
 \|\mathcal P_nf\|_\infty\leq C\|f\|_\infty.
\]
The univariate approximation bounds telescope across axes and give $\omega_{2n},\omega_{\infty n}$ in (iii). Define $\omega_{vn}$ by $\|v_k^*-\mathcal P_nv_k^*\|_{L^2}=O(\omega_{vn})$; additional smoothness permits the faster bound in Remark~\ref{rem: penalized rates}. Constants may depend on the fixed degrees and dimensions but not on $n$. The second-derivative inverse inequality gives $\|S_j\|=O(m_n^4)$. The penalty normalization and the bounds on $G_n$ also give $\operatorname{tr}(G_nS_j^+)\leq C\operatorname{rank}S_j\leq Cq_n=o(n)$.

The population score at $\mathcal P_n\eta_0$ is $O(\omega_{2n})$, and the population objective is strongly concave on a coefficient ball of radius $\delta/\sqrt{q_n}$ around it. Condition (vii) implies $\sqrt{q_n}\,\omega_{2n}\to0$, since $q_n/\sqrt n\to0$ and $\sqrt n\,\omega_{2n}^2\to0$. The Kullback--Leibler coefficient $a_n^*$ therefore lies in that ball, with
\[
 \|\eta_n^*-\eta_0\|_{L^2}=O(\omega_{2n}),\qquad
 \|\eta_n^*-\eta_0\|_\infty=O(\sqrt{q_n}\,\omega_{2n}+\omega_{\infty n}).
\]
Its energies are bounded by those of $\mathcal P_n\eta_0$ plus $Cm_n^2\omega_{2n}=O(1)$, using (iii) and $s>2$.

Put $v_n^0=(\mathcal I_n^{0\,-1}\gamma_n)'\Psi_n$, the $I_0$-weighted sieve projection of $v_k^*$. Best approximation and equivalence of the weighted and unweighted norms give $\|v_n^0-\mathcal P_nv_k^*\|_{L^2}=O(\omega_{vn})$. Thus inverse inequalities give
\[
 \|v_n^0\|_\infty\leq C+C\sqrt{q_n}\,\omega_{vn}=O(1),\qquad
 J_j(v_n^0)\leq C+Cm_n^2\omega_{vn}=O(1).
\]
For $\beta_n=\mathcal I_n^{-1}\gamma_n$ and $v_n=\beta_n'\Psi_n$, the two coefficient equations imply
\[
 \mathcal I_n(\beta_n-\mathcal I_n^{0\,-1}\gamma_n)
 =E[\Psi_n\{I_0-\bar I(X_i;\eta_n^*(X_i))\}v_n^0(X_i)'].
\]
Boundedness of $v_n^0$, smoothness of the information, and Cauchy--Schwarz in each unit coefficient direction bound the right side by $C\omega_{2n}$. Consequently $\|v_n-v_n^0\|_{L^2}=O(\omega_{2n})$. The preceding inverse inequalities, $\sqrt{q_n}\omega_{2n}\to0$, and $m_n^2\omega_{2n}=O(1)$ give bounded sup-norm and energies for $v_n$, and
\[
 \nu_n=\|v_k^*-v_n\|_{I_0}=O(\omega_{vn}+\omega_{2n}).
\]
\end{proof}

\begin{lemma}[Smoothing-parameter bound]\label{lemma: smoothing bound}
Under Assumption~\ref{ass: penalized sieve}, there is $C<\infty$ such that, with probability tending to one, every fast restricted maximum likelihood fixed point satisfies $\max_j\hat\lambda_j\leq Cq_n$. The same holds for every Laplace Karush--Kuhn--Tucker point if, with probability tending to one, the estimated indices of all such points lie in a common compact set.
\end{lemma}

\begin{proof}[Proof of Lemma~\ref{lemma: smoothing bound}] Write $\hat{\mathcal I}_W$ for the working information at $\hat a$, $N=n\hat{\mathcal I}_W$, $S=S(\hat\lambda)$, $A=N+S$, $B_j=\hat\lambda_jS_j$, $a_j=\hat\lambda_jJ_j(\hat a)^2$, and $b_j=\operatorname{tr}(S^+B_j)-\operatorname{tr}(A^{-1}B_j)$. If $\hat\lambda_j>\lambda_{\min}$, the Karush--Kuhn--Tucker conditions give $a_j\leq b_j$. On the range of $S$, $S^+-A^{-1}=S^+(N-NA^{-1}N)S^+\preceq S^+NS^+$, and commuting penalties give $S^+B_jS^+\preceq B_j^+$; all these bounds use the common normalized coordinates in Assumption~\ref{ass: penalized sieve}(iv). The working weights are bounded, so $\hat\lambda_ja_j\leq n\operatorname{tr}(\hat{\mathcal I}_WS_j^+)\leq Cn\tau_n$ with $\tau_n=\max_j\operatorname{tr}(G_nS_j^+)=o(n)$.

\emph{Excluding large penalties.} Fix $\epsilon>0$ and let $\mathcal A=\{j:\hat\lambda_j>\epsilon n\}$. At $\hat a$, the objective that drops the penalties in $\mathcal A$ has gradient $\sum_{j\in\mathcal A}\hat\lambda_jS_j\hat a$. Concavity therefore bounds the average log-likelihood deficiency of $\hat a$ relative to $a^*_n$ by $C\sqrt{\tau_n/n}+(\epsilon/2)\sum_jJ_j(a^*_n)^2$, where $\sum_jJ_j(a^*_n)^2=O(1)$ by Lemma~\ref{lemma: sieve approximation energy}. 

The score has mean zero at $a^*_n$. Consider the Bregman remainder of $-\mathcal L$ at the bounded index $\eta^*_n$. Extending each observation's remainder linearly outside a fixed interval around $\eta^*_n$ gives a convex Lipschitz minorant in the index displacement $t=\Psi_n(X_i)'(a-a^*_n)$. In the factorized ordered model, each stage retains its observed participation indicator as a weight. By the participation probability bound in (i), the conditional expectation of the minorant given $X_i$ is at least $c\min(\|t\|,\|t\|^2)$. On the set where $E\|t\|=r_0$, Jensen's inequality applied to $\|t\|^2/(1+\|t\|)$ therefore bounds the population deficiency below by $cr_0^2/(1+r_0)$. 

Because $\|u\|\leq C\sqrt{q_n}\,E|\Psi_n(X_i)'u|$ on the sieve, symmetrization and contraction applied to this minorant bound the empirical deviation on that set by $O_p(r_0q_n/\sqrt n)=o_p(1)$. The linear score term is also uniformly $O_p(r_0q_n/\sqrt n)$, because $\|P_n\Psi_ns(D,\eta_n^*)\|=O_p(\sqrt{q_n/n})$ and $\|u\|\leq C\sqrt{q_n}r_0$ on that set. 

Concavity along the segment from $a^*_n$ to $\hat a$ then gives $\|\hat\eta-\eta^*_n\|_{L^1}\leq C\sqrt\epsilon$ with probability tending to one. For $j\in\mathcal A$, $J_j(\hat a)^2\leq C\tau_n/(\epsilon^2n)\to0$. Integrated squared second derivatives are lower semicontinuous under $L^1$ convergence, and the population centering that defines the components of $\mathcal F_A$ is $L^1$ continuous under (i). By (iv) and (v), the set on which a penalized roughness is small is therefore bounded away from $\eta_0$ in $L^1$, so $\mathcal A$ is empty for small $\epsilon$, hence for every $\epsilon$.

\emph{Selected-penalty bound.} Repeating the argument with no dropped penalty gives $\|\hat\eta-\eta^*_n\|_{L^1}=o_p(1)$ and $J_j(\hat a)\geq cJ_j(\eta_0)$. Since $a_j\leq b_j\leq\operatorname{rank}S_j$, $\hat\lambda_j\leq\operatorname{rank}(S_j)/\{c^2J_j(\eta_0)^2\}\leq Cq_n$. Coordinates at the lower bound are controlled by the fixed $\lambda_{\min}$. These bounds hold on events that do not depend on the fixed point.

\emph{Laplace REML.} For Laplace restricted maximum likelihood, write
\[
 L_j=n\operatorname{tr}\{(n\hat{\mathcal I}_O+S)^{-1}
                      \partial_{\log\lambda_j}\hat{\mathcal I}_O\}.
\]
The compactness condition bounds the observed information away from zero and its derivatives. The derivative of the penalized coefficient equation and the trace bound then give
\[
 a_j\leq b_j+|L_j|,\qquad
 b_j\leq Cn\tau_n/\hat\lambda_j,\qquad
 |L_j|\leq Cq_n^{3/2}(n+\hat\lambda_j)^{-1/2}\sqrt{a_j}.
\]
Multiplying by $\hat\lambda_j$ and applying $xy\leq(x^2+y^2)/2$ yields
$\hat\lambda_j a_j\leq C(n\tau_n+q_n^3)$.
Thus, for $\hat\lambda_j>\epsilon n$, the roughness still tends to zero and the likelihood deficiency above has the additional term $O(q_n^{3/2}/n)=o(1)$. The same localization excludes these large coordinates and gives the positive roughness lower bound. Finally $b_j\leq\operatorname{rank}S_j$ and the last display imply
$a_j\leq C\{\operatorname{rank}S_j+q_n^3/n\}=O(q_n)$.
Dividing by the roughness lower bound proves the asserted smoothing bound, including coordinates on either smoothing boundary.
\end{proof}

\begin{prop}[Penalized first-stage expansion]\label{prop: penalized expansion}
Under Assumptions~\ref{ass: sampling} and \ref{ass: penalized sieve}, with the smoothness conditions on $\lambda_{k0}$ and the control maps in Section~\ref{sec: asymptotics}, Assumption~\ref{ass: first stage influence} holds uniformly over $\Lambda_n=\{\lambda:\max_j\lambda_j\leq Cq_n\}$, and therefore at $\hat\lambda$, with influence \eqref{eq: projected influence}.
Moreover $P_n\|\hat\eta-\eta_0\|^2=O_p(\rho_n^2)$ and $\|\hat\eta-\eta_0\|_\infty=O_p(e_n)$, so that $\varrho_n=O_p(\rho_n)$ and $r_n=O_p(e_n)$. If $\mathcal F=\mathcal F_F$, then $v_k^*=\Gamma_kI_0^{-1}$ and \eqref{eq: projected influence} equals \eqref{eq: probability influence}. For a restricted space under the same joint selection likelihood, with $s_i=s(D_i,\eta_0(X_i))$,
\[
 E[\Gamma_k(X_i)I_0(X_i)^{-1}s_is_i'I_0(X_i)^{-1}\Gamma_k(X_i)']-E[\chi_{ik}\chi_{ik}']
 =E[(\Gamma_kI_0^{-1}-v_k^*)I_0(\Gamma_kI_0^{-1}-v_k^*)']\succeq0,
\]
with equality if and only if the rows of $\Gamma_kI_0^{-1}$ lie in $\overline{\mathcal F}$.
\end{prop}

\begin{proof}[Proof of Proposition~\ref{prop: penalized expansion}] Use the population approximations and matrices of Lemma~\ref{lemma: sieve approximation energy}, and put $\beta_n=\mathcal I_n^{-1}\gamma_n$. On the sphere $\|a-a^*_n\|=r\leq\delta/\sqrt{q_n}$, the penalized objective lies below its value at $a^*_n$ once $r\geq C\sqrt{q_n/n}$, because the score at $a^*_n$ has norm $O_p(\sqrt{nq_n})$ and $a_n^{*\prime}S(\lambda)a^*_n\leq Cq_n\sum_jJ_j(a^*_n)^2=O(q_n)$ on $\Lambda_n$. Hence $\|\hat a_\lambda-a^*_n\|=O_p(\sqrt{q_n/n})$ uniformly on $\Lambda_n$, which gives the stated rates.

The weighted drift in Assumption~\ref{ass: first stage influence} is linear in $\hat\eta-\eta_0$ up to a remainder bounded by $CP_n\|\hat\eta-\eta_0\|^2=O_p(\rho_n^2)$, and $E[T_{ik}U_{ik}\nabla\lambda_{k0}(g_{k0}(X_i))'\partial_\eta g_k(X_i;\eta_0(X_i))\mid X_i]=\Gamma_k(X_i)$. The drift therefore equals $\gamma_n'(\hat a-a^*_n)+E[\Gamma_k(X_i)\{\eta^*_n(X_i)-\eta_0(X_i)\}]+O_p(q_n/n+\rho_n^2)+o_p(n^{-1/2})$, where the $O_p$ term is $o_p(n^{-1/2})$ because (vii) gives $q_n/\sqrt n\to0$ and $\sqrt n\,\rho_n^2\to0$. 

The first-order condition of $a^*_n$ and $\|v_n\|_\infty=O(1)$ make the middle term $O(\omega_{2n}^2+\nu_n\omega_{2n})$. Let $\hat{\mathcal A}$ be the observed Hessian averaged along the segment from $a^*_n$ to $\hat a$, $g_n=\sum_i\Psi_n(X_i)s(D_i,\eta^*_n(X_i))$, and $A=n\mathcal I_n+S(\lambda)$. The exact identity $\{\hat{\mathcal A}+S(\lambda)\}(\hat a-a^*_n)=g_n-S(\lambda)a^*_n$ gives, with $S=S(\lambda)$,
\[
 \gamma_n'(\hat a-a^*_n)=\tfrac1n\beta_n'g_n-\tfrac1n\beta_n'SA^{-1}g_n-\tfrac1n\beta_n'(S-SA^{-1}S)a^*_n
 +(\beta_n-A^{-1}S\beta_n)'(\mathcal I_n-\tfrac1n\hat{\mathcal A})(\hat a-a^*_n).
\]
Since $0\preceq S-SA^{-1}S\preceq S$, the penalty term is at most $n^{-1}(\beta_n'S\beta_n)^{1/2}(a_n^{*\prime}Sa^*_n)^{1/2}=O(q_n/n)$. The second term is $O_p(q_n/n)$. Moreover $\|A^{-1}S\beta_n\|\leq(cn)^{-1/2}(\beta_n'S\beta_n)^{1/2}=O(\sqrt{q_n/n})$, so the effective loading $(\beta_n-A^{-1}S\beta_n)'\Psi_n$ differs from $v_n$ in sup-norm by $O(q_n/\sqrt n)=o(1)$. Its boundedness makes the Hessian term $O_p(q_n/n+q_n^2/n^{3/2})$. 

Finally, $n^{-1}\beta_n'g_n=n^{-1}\sum_iv_n(X_i)s(D_i,\eta^*_n(X_i))$ differs from $n^{-1}\sum_i\chi_{ik}$ by $o_p(n^{-1/2})$, because $v_ns(\cdot,\eta^*_n)$ and $v_ns(\cdot,\eta_0)$ both have mean zero, $\|\eta^*_n-\eta_0\|_\infty\to0$, and $\nu_n\to0$. Each bound depends on $\lambda$ only through $S(\lambda)\preceq S(\bar\lambda)$, so the expansion is uniform over $\Lambda_n$, and Lemma~\ref{lemma: smoothing bound} places $\hat\lambda$ in $\Lambda_n$. The variance comparison follows from orthogonality of $\Gamma_kI_0^{-1}-v_k^*$ to $\overline{\mathcal F}$ in $\langle\cdot,\cdot\rangle_{I_0}$.
\end{proof}

\begin{prop}[Residualized regressors]\label{prop: residualized regressors}
Let $B_n$ be a tensor product of B-splines of fixed degree $p_g\geq3$ with $J_n$ quasi-uniform knots per axis and $\kappa_n$ basis functions, scaled by $\kappa_n^{1/2}$. Suppose the distribution of $g_{k0}(X_i)$ conditional on $D_i=k$ has a density bounded above and away from zero on a compact rectangle, $\Pi_k$ is Lipschitz with spline approximation error $O(J_n^{-s_X})$ on a neighborhood of it, $\kappa_n\log n/n\to0$, $\varrho_n=O_p(\rho_n)$, $r_n=O_p(e_n)$, and $J_ne_n\to0$. Then, with $\bar c_n=J_n^{-s_X}+\sqrt{\kappa_n/n}+(1+J_n)\rho_n$,
\[
 c_n=O_p(\bar c_n),\qquad
 \max_iT_{ik}\|\hat U_{ik}\|=O_p(1).
\]
\end{prop}

\begin{proof}[Proof of Proposition~\ref{prop: residualized regressors}] Write $X_i=M_{0i}+U_{ik}$ with $M_{0i}=\Pi_k(g_{k0}(X_i))$, let $\hat Z$ and $Z_0$ have rows $B_n(\hat g_k(X_i))'$ and $B_n(g_{k0}(X_i))'$ in the selected sample, and let $\hat P$ project on the columns of $\hat Z$. Because $J_nr_n\to_p0$, each generated control stays in a fixed neighborhood of the cell that contains its true value, each coordinate of $B_n(\hat g_k(X_i))-B_n(g_{k0}(X_i))$ is at most $CJ_n\kappa_n^{1/2}\|\Delta_i\|$, and cell counts are $O_p(n/\kappa_n)$. Hence $n^{-1/2}\|\hat Z-Z_0\|_{\rm op}=O_p(J_nr_n)$ and $\hat Z'\hat Z/n$ has eigenvalues bounded away from zero. The spline approximant of $\Pi_k$ evaluated at $\hat g_k$ lies in the column space of $\hat Z$, so best approximation and the Lipschitz property give $\|(I-\hat P)M_0\|_{n,k}\leq C\varrho_n+CJ_n^{-s_X}$. 

For the remaining part, $\|\hat PU\|_{n,k}\leq C\|Z_0'U/n\|+C\|(\hat Z-Z_0)'U/n\|$. The first term is $O_p(\sqrt{\kappa_n/n})$ because $Z_0$ depends on $X_i$ only through $g_{k0}(X_i)$ and $E[U_{ik}\mid g_{k0}(X_i),D_i=k]=0$. For the second, Cauchy--Schwarz within each support neighborhood $N_l$, $|N_l|\leq Cn/\kappa_n$, bounded overlap, and bounded $U$ give
\[
 \|(\hat Z-Z_0)'U/n\|^2\leq CJ_n^2\kappa_n\sum_l\frac{|N_l|}{n^2}\sum_{i\in N_l}T_{ik}\|\Delta_i\|^2\leq CJ_n^2\varrho_n^2.
\]
This bound does not require independence between $\hat g_k$ and $U$. The local Gram matrix has a fixed neighbor pattern and a bounded spectrum, so its inverse has uniformly bounded row sums \citep{demko1984decay}. With rows of $\hat Z$ of size $O(\kappa_n^{1/2})$ and column means of $|\hat Z|$ of order $\kappa_n^{-1/2}$, $\hat P$ is bounded in $\ell_\infty$, which gives the bound on $\max_iT_{ik}\|\hat U_{ik}\|$.
\end{proof}

\begin{lemma}[Variance estimation]\label{lemma: penalized variance}
Suppose the conditions of Theorem~\ref{thm: asymptotic normality}, Lemma~\ref{lemma: smoothing bound}, and Propositions~\ref{prop: penalized expansion}--\ref{prop: residualized regressors} hold, $J_n(\sqrt{\kappa_n/n}+\rho_n+a_n)\to0$, $\sqrt n\,\bar c_n\rho_n\to0$, and $b_n\to0$. Use the score-form influence estimate \eqref{eq: estimated projected influence}, with the observed or working coefficient information. Then both components of the estimated influence are mean-square consistent, and $\hat{\mathcal V}_k\to_p\mathcal V_k$ with $\chi_{ik}$ given by \eqref{eq: projected influence}. The same holds with $\hat{\mathcal I}^{-1}$ in place of $\{\hat{\mathcal I}+S(\hat\lambda)/n\}^{-1}$, and for the influence obtained by differentiating the second-stage least-squares coefficients with respect to $\hat a$.
\end{lemma}

\begin{proof}[Proof of Lemma~\ref{lemma: penalized variance}] Let $W=[X,\hat Z]$ in the selected sample, and let $\hat P_W$ project on its columns. Because $\hat Q_k\to_pQ_k\succ0$ and $\hat Z'\hat Z/n$ is well conditioned, the Schur factorization of $W'W/n$ has bounded outer factors, and $\hat Z'M_X\hat Z/n$ has eigenvalues bounded away from zero. 

\emph{Estimating $\nabla\lambda_{k0}$.} Conditional on the selection data, the outcome errors have mean zero and bounded variance, so $\|\hat\delta_k-\delta_{kn}\|=O_p(\sqrt{\kappa_n/n}+\rho_n+a_n)$. The derivative Gram matrix of $B_n$ at the generated controls is $O_p(J_n^2)$ by the same local-support and cell-count bounds, so $\|\nabla\hat\ell_k(\hat g_k)-\nabla\lambda_{k0}(g_{k0})\|_{n,k}\leq CJ_n\|\hat\delta_k-\delta_{kn}\|+b_n+C\varrho_n=o_p(1)$. 

\emph{First-stage loading.} Multiplying by the bounded $\hat U_{ik}$ and the uniformly consistent $\partial_\eta g_k(X_i;\hat\eta(X_i))$ gives $\|\hat h_k-h_k\|_{n,k}=o_p(1)$ for $h_{ik}=T_{ik}U_{ik}\nabla\lambda_{k0}(g_{k0}(X_i))'\partial_\eta g_k(X_i;\eta_0(X_i))$, whose conditional mean given $X_i$ is $\Gamma_k(X_i)$. The empirical Gram bound and $\|(P_n-P)\Psi_nh_k'\|=O_p(\sqrt{q_n/n})$ then give $\|\hat\gamma_k-\gamma_n\|=o_p(1)$. Uniform consistency of the indices bounds the scores and their empirical Gram matrix. In addition, $\|\{n\hat{\mathcal I}+S(\hat\lambda)\}^{-1}S(\hat\lambda)\beta_n\|\leq(cn)^{-1/2}\{\beta_n'S(\hat\lambda)\beta_n\}^{1/2}=O_p(\sqrt{q_n/n})$, and $\|\hat{\mathcal I}-\mathcal I_n\|=o_p(1)$ for both the observed and the working information. In the factorized probit, the working information includes the observed participation indicator. After the indicator is centered by its conditional mean, its contribution to the sample information matrix converges to zero in probability. Together with $\nu_n\to0$ these bounds give $n^{-1}\sum_i\|\hat\chi_{ik}-\chi_{ik}\|^2=o_p(1)$. 

\emph{Outcome influence.} For the outcome part, the average conditional second moment $n^{-1}\sum_iE[T_{ik}\|(\hat U_{ik}-U_{ik})\varepsilon_{ik}\|^2\mid\{(D_j,X_j)\}_{j=1}^n]$ is bounded by $Cc_n^2$, so Markov's inequality applies, and $T_{ik}\hat U_{ik}(\hat\varepsilon_{ik}-\varepsilon_{ik})$ is controlled by $\max_i\|\hat U_{ik}\|$ and $\|\hat\varepsilon_k-\varepsilon_k\|_{n,k}=o_p(1)$. 

\emph{Complete least-squares derivative.} Differentiating the second-stage normal equations adds the loading $\hat\varepsilon_{ik}\nabla\hat m_k(\hat g_k(X_i))'$, where $\hat m_k$ is the sample least-squares projection of $X$ on $B_n(\hat g_k)$. Since $\|\nabla\hat m_k\|_{n,k}=O_p(J_n)$ and $\hat P_W$ is bounded in $\ell_\infty$, its contribution to $\hat\gamma_k$ is $O_p(J_n\sqrt{q_n/n}+J_n(e_n+a_n))=o_p(1)$, because $\bar c_n\geq J_n\rho_n$ and $\rho_n\geq\sqrt{q_n/n}$ give $J_n\sqrt{q_n/n}\leq J_nq_n/\sqrt n\leq\sqrt n\,\bar c_n\rho_n\to0$, and $J_ne_n\to0$ and $J_na_n\to0$ are conditions of Proposition~\ref{prop: residualized regressors} and of the lemma.
\end{proof}

\begin{remark}[Admissible growth rates]\label{rem: penalized rates}
For unconstrained splines, $C^{s_v}$ components of $v_k^*$ give $\omega_{vn}=O(m_n^{-\min(s_v,p_\eta+1)})$; bounded second-derivative energy gives $O(m_n^{-2})$. For natural cubic first-stage splines with $s\geq5/2$ in $\mathcal F_A$, a cubic second-stage basis in $L=2$ controls, and $\lambda_{k0},\Pi_k\in C^4$, the rate conditions in Assumptions~\ref{ass: penalized sieve}--\ref{ass: joint penalized growth} hold at $m_n\asymp n^{0.15}$ and $J_n\asymp n^{0.075}$.

For any fixed component dimension $d=d_{\mathcal F}$ and control dimension $L$, one sufficient construction uses unconstrained spline degrees $p_\eta\geq4d-1$ and $p_g\geq4L-1$, first-stage indices and $v_k^*$ with $C^{4d}$ components, and $\lambda_{k0},\Pi_k\in C^{4L}$. Set
\[
 m_n\asymp n^{1/(8d)},\qquad J_n\asymp n^{1/(8L)},\qquad
 q_n\asymp\kappa_n\asymp n^{1/8}.
\]
These choices satisfy the first-stage, design, and variance rate conditions; the largest terms in \eqref{eq: main joint rates} are $O(n^{-3/8+1/(8L)})=o(1)$. The remaining support, rank, and model assumptions are maintained.
\end{remark}

Combining these results establishes Theorem~\ref{thm: penalized two step}.

\begin{proof}[Proof of Theorem~\ref{thm: penalized two step}]
Lemma~\ref{lemma: smoothing bound} places the selected smoothing parameters in $\Lambda_n$. Proposition~\ref{prop: penalized expansion} gives Assumption~\ref{ass: first stage influence} with \eqref{eq: projected influence}, together with $\varrho_n=O_p(\rho_n)$ and $r_n=O_p(e_n)$. Under Assumption~\ref{ass: joint penalized growth}, Proposition~\ref{prop: residualized regressors} gives $c_n=O_p(\bar c_n)$ and bounded residualized regressors. The design and approximation bounds in \eqref{eq: main design rates}--\eqref{eq: main joint rates} therefore imply every condition of Assumption~\ref{ass: sieve}. The general expansion in Theorem~\ref{thm: asymptotic normality} proves part (i). Finally, \eqref{eq: main variance rates}, $b_n\to0$, and $\sqrt n\,\bar c_n\rho_n\to0$ are the additional conditions of Lemma~\ref{lemma: penalized variance}, which proves part (ii).
\end{proof}

\subsection{Inference using the ordered restriction}\label{sec: structural inference details}

Use $\mathcal K$, $\mathcal S$, and $R^{\rm str}$ from Appendix~\ref{sec: structural proofs}. Write $T_i=\mathbf1\{D_i=k\}$, $R_i=R_i^{\rm str}$, $Q^{\rm str}=E[T_iR_iR_i']$, and $\|v\|_{n,k}^2=P_nT_i\|v_i\|^2$. The upper threshold is fixed at one in the top category. All derivatives below are taken only with respect to varying thresholds.

\subsubsection{The integrated spline space and its derivatives}

The estimator in \eqref{eq: structural sieve OLS} averages centered cubic splines over the selection interval and includes an intercept. Integrating the cubic spline together with its intercept gives a quartic spline $F$, so
\begin{equation}\label{eq: structural primitive span}
 c=F(1)-F(0),\qquad \varphi=F'-c,\qquad
 c+(\mathcal K\varphi)(s,t)=\frac{F(t)-F(s)}{t-s}.
\end{equation}
Conversely, any quartic spline $F$ with simple interior knots gives such a centered cubic spline and intercept. The two representations therefore generate the same regression space. The quartic representation allows the stability, derivative, and approximation bounds below to be derived from univariate spline properties.
For separated support intervals $I_s,I_t$ of the two thresholds, a sufficiently fine mesh makes the supports of the quartic B-splines intersecting $I_s$ disjoint from those intersecting $I_t$. Their restrictions span the full quartic spline space on each interval. After removing directions that vanish on both intervals and the common constant, the same integrated spline space has the representation
\begin{equation}\label{eq: structural local span}
 z_n(s,t)'\delta=\frac{F_t(t)-F_s(s)}{t-s},
 \qquad \int_{I_s}F_s(s)\,ds=0.
\end{equation}
The normalization only fixes the otherwise common additive constant. For the top category the representation is $z_n(s)'\delta=F_s(s)/(1-s)$, with $F_s$ an unrestricted quartic spline. The value of the primitive at one is a free constant, because that point is separated from $I_s$. The mesh condition in Assumption~\ref{ass: structural inference} leaves a knot-free interval of radius $c/J_n$ around each support boundary. A spline that vanishes on the threshold support vanishes on the entire knot cells meeting that support, by polynomial uniqueness. The same representation and its polynomial continuation therefore remain valid on this $O(J_n^{-1})$ enlargement. The following lemma provides the stability and approximation bounds needed for Theorem~\ref{thm: structural inference}.

\begin{lemma}[Regularity of the integrated cubic space]\label{lemma: structural approximation}
Under the support, density, smoothness, and spline conditions of Assumption~\ref{ass: structural inference}, the space $\operatorname{span}\{1,Z_n\}$ has a representation $z_n$ of dimension $d_n\asymp J_n$ satisfying
\begin{equation}\label{eq: structural Gram bounds}
 cI\preceq E[T_i z_n(H_i)z_n(H_i)']\preceq CI,
 \qquad \sup_h\|\partial^jz_n(h)\|\leq C J_n^{j+1/2},\quad j=0,1,2.
\end{equation}
There are $\ell_n,m_n^{\rm str}\in\operatorname{span}\{1,Z_n\}$ approximating $\alpha_{k0}+b_k$ and $m^{\rm str}=\operatorname{Proj}_{\mathcal S}X$, respectively, with
\begin{equation}\label{eq: structural approximation bounds}
 \begin{aligned}
 \|\ell_n-\alpha_{k0}-b_k\|_\infty+
 \|m_n^{\rm str}-m^{\rm str}\|_\infty&=O(J_n^{-5}),\\
 \|\nabla\ell_n-d_0\|_\infty&=O(J_n^{-4}),\qquad
 \|\nabla^2\ell_n\|_\infty=O(1).
 \end{aligned}
\end{equation}
The analogous first- and second-derivative bounds hold for $m_n^{\rm str}$. In particular, $m^{\rm str}$ and $R=X-m^{\rm str}(H)$ are bounded.
\end{lemma}

\begin{proof}
On the separated rectangle, $f_k(s,t)/(t-s)^2$ is bounded above and away from zero. Hence the squared norm of \eqref{eq: structural local span} is equivalent, uniformly in $n$, to
\[
 |I_t|\|F_s\|_{L^2(I_s)}^2+|I_s|\|F_t\|_{L^2(I_t)}^2,
\]
because the cross term vanishes under $\int F_s=0$. Scale the local quartic B-splines by $J_n^{1/2}$ and impose the mean-zero constraint through an orthonormal coefficient map of operator norm one. This map preserves the univariate Gram and derivative bounds, as does multiplication by $(t-s)^{-1}$ because the width is bounded below, giving \eqref{eq: structural Gram bounds}. For the top category the same argument uses the weight $f_k(s)/(1-s)^2$.

To establish smoothness of the covariate projection, put $v(s,t)=(t-s)E[X\mid s,t,D=k]$ and $\omega(s,t)=f_k(s,t)/(t-s)^2$. Any pair of smooth functions on the separated intervals extends to an absolutely continuous primitive on $[0,1]$ with square-integrable derivative. Such pairs are dense in the two component spaces. Thus projection onto $\mathcal S$ is precisely weighted additive projection of $v$ onto $F_t(t)-F_s(s)$. The preceding norm equivalence makes this additive space closed and bounds the $L^2$ norms of its normalized components. Its normal equations include
\[
 F_s(s)=\frac{\int_{I_t}\omega(s,t)\{F_t(t)-v(s,t)\}\,dt}
                    {\int_{I_t}\omega(s,t)\,dt},
\]
and the corresponding equation for $F_t$. The denominators are bounded away from zero, and the assumed smoothness and finite $L^1$ norms of the other components permit five differentiations. Both components are therefore $C^5$ with bounded derivatives.

For the top category, any polynomial $g$ can be written as $g(s)=c+(\mathcal K A')(s,1)$ by taking $c=g(0)$ and $A(v)=-(1-v)\{g(v)-g(0)\}$. Since $A(0)=A(1)=0$ and $A'\in L^2[0,1]$, polynomial density gives $\mathcal S=L^2(s,P_k)$. Hence $m^{\rm str}(s)=E[X\mid s,D=k]$ is $C^5$ by assumption.

For the top category, $A_k(1)=0$ gives
\begin{equation}\label{eq: structural top derivatives}
 b_k(s,1)=-\frac{A_k(s)}{1-s},\qquad
 \partial_s b_k(s,1)=\frac{b_k(s,1)-a_k(s)}{1-s}.
\end{equation}
Because $1$ is separated from $I_s$, the primitive approximation can satisfy $\tilde A_n(1)=0$ without altering its restriction to $I_s$. The approximation errors for $b_k(\cdot,1)$ and its derivative on $I_s$ are bounded by
\[
 C\|\tilde A_n-A_k\|_{\infty,I_s},\qquad
 C\{\|\tilde A_n-A_k\|_{\infty,I_s}
       +\|\tilde A_n'-a_k\|_{\infty,I_s}\},
\]
respectively.

Apply quartic spline approximation to $A_k$ and the projection components on the support intervals, using $(1-s)m^{\rm str}(s)$ for the top-category projection. This gives value error $O(J_n^{-5})$, derivative error $O(J_n^{-4})$, and bounded second derivatives \citep{schumaker2007spline}. The preceding disjoint-support construction places these approximants in the original integrated space; constants and linear functions preserve the intercept and the normalization in \eqref{eq: structural local span}. Dividing by the bounded-away-from-zero width gives \eqref{eq: structural approximation bounds}. Extend the smooth projection components across the fixed support boundaries. On each boundary cell, quartic polynomial continuation and derivative approximation bounds through order four preserve these rates on the $O(J_n^{-1})$ enlargement.
\end{proof}

\subsubsection{Residualization, first-stage influence, and feasible variance}

Let $\hat Z$ and $Z_0$ have rows $z_n(\hat H_i)'$ and $z_n(H_i)'$ in the selected sample, and let $\hat P$ project onto $\hat Z$. Write $\hat m^{\rm str}$ for the sample least-squares projection of $X$ on this space. All coefficient bounds use the stable basis $z_n$.

\begin{prop}[Generated integrated splines and their estimated influence]\label{prop: structural feasible inference}
Under Assumption~\ref{ass: structural inference} and the stated Laplace compactness condition when needed,
\begin{align}
 \|\hat R-R\|_{n,k}&=O_p(\bar c_n^{\rm str}),\qquad
 \max_iT_i\|\hat R_i\|=O_p(1),\label{eq: structural residual rate}\\
 \|\hat d-d_0(H)\|_{n,k}
 &=O_p\!\left[J_n\{\sqrt{J_n/n}+\rho_n+J_n^{-5}\}
                       +J_n^{-4}+\rho_n\right]=o_p(1).
 \label{eq: structural derivative rate}
\end{align}
The first-stage expansion is
\begin{equation}\label{eq: structural first stage expansion}
 P_nT_iR_i d_0(H_i)'(\hat H_i-H_i)
 =P_n\chi_i^{\rm str}+o_p(n^{-1/2}),
 \qquad \chi_i^{\rm str}=v_k^{\rm str}(X_i)s(D_i,\eta_0(X_i)).
\end{equation}
For separately estimated cumulative probabilities, sum the equation-specific expansions, each computed with its own score and information matrix. The influence estimate defined before Theorem~\ref{thm: structural inference} satisfies
\begin{equation}\label{eq: structural influence consistency}
 P_n\|\hat\chi_i^{\rm str}-\chi_i^{\rm str}\|^2=o_p(1),\qquad
 P_nT_i\|\hat R_i\hat\varepsilon_i-R_i\varepsilon_i\|^2=o_p(1).
\end{equation}
These conclusions also hold when the loading is obtained by differentiating the second-stage least-squares estimator.
\end{prop}

\begin{proof}
\emph{Design and residualization.} Lemma~\ref{lemma: smoothing bound} and Proposition~\ref{prop: penalized expansion} give $\|\hat H-H\|_{n,k}=O_p(\rho_n)$ and $\|\hat H-H\|_\infty=O_p(e_n)$. The Gram matrix at the true thresholds has expected squared Frobenius error $O(J_n^2/n)$ by \eqref{eq: structural Gram bounds}. Since $J_ne_n\to0$, estimated thresholds lie in the knot-cell enlargement with probability tending to one and stay within a fixed number of adjacent cells, each containing $O_p(n/J_n)$ observations. The local basis and its smooth width factor therefore give
\[
 n^{-1/2}\|\hat Z-Z_0\|_{\rm op}=O_p(J_ne_n),\quad
 \lambda_{\min}(\hat Z'\hat Z/n)\geq c+o_p(1),\quad
 \left\|P_nT_i\partial z_n(\hat H_i)\partial z_n(\hat H_i)'\right\|=O_p(J_n^2),
\]
where the last bound applies to differentiation with respect to each threshold. The nonlocal constants introduced by centering have bounded coefficient norm and satisfy the same bounds.

Decompose $X=m^{\rm str}(H)+R$. Best approximation using $m_n^{\rm str}(\hat H)$ gives
$\|(I-\hat P)m^{\rm str}(H)\|_{n,k}\leq C(J_n^{-5}+\rho_n)$.
Orthogonality to $\mathcal S$ gives $E[T_i z_n(H_i)R_i']=0$, even though $E[R_i\mid H_i,D_i=k]$ need not vanish. Because $R$ is bounded,
$\|Z_0'R/n\|=O_p(\sqrt{J_n/n})$.
For a local basis support neighborhood $N_l$, Cauchy--Schwarz, $|N_l|=O_p(n/J_n)$, and bounded overlap imply
\[
 \|(\hat Z-Z_0)'R/n\|^2
 \leq C J_n^3\sum_l\frac{|N_l|}{n^2}
             \sum_{i\in N_l}T_i\|\hat H_i-H_i\|^2
 \leq C J_n^2\|\hat H-H\|_{n,k}^2.
\]
The smooth width and centering terms are bounded by $C\|\hat H-H\|_{n,k}$ and can be included in this bound. It follows that $\|\hat P R\|_{n,k}=O_p(\sqrt{J_n/n}+J_n\rho_n)$, proving \eqref{eq: structural residual rate} in $L^2$.

The normal equations for the projection of $X$ give coefficient error $O_p(\bar c_n^{\rm str})$ relative to $m_n^{\rm str}$ at the estimated thresholds. The sup and derivative inverse bounds in \eqref{eq: structural Gram bounds} then give
\[
 \|\hat m^{\rm str}-m_n^{\rm str}\|_\infty
       =O_p(J_n^{1/2}\bar c_n^{\rm str}),\qquad
 \|\nabla\hat m^{\rm str}-\nabla m_n^{\rm str}\|_\infty
       =O_p(J_n^{3/2}\bar c_n^{\rm str})=o_p(1).
\]
The last term in \eqref{eq: structural design rates} therefore ensures that both $\hat R$ and $\nabla\hat m^{\rm str}$ are uniformly bounded in probability.

\emph{Estimating the threshold derivatives of $b_k$.} By the residualization result, $P_nT_i\hat R_i\hat R_i'\to_pQ^{\rm str}\succ0$. Schur factorization therefore makes the augmented design $[X,\hat Z]$ well conditioned. The conditional mean of the outcome is $X'\beta_{k0}+\ell_n(\hat H)$ plus an error of empirical norm $O_p(J_n^{-5}+\rho_n)$. Conditional mean zero and bounded variance of the outcome errors give coefficient error
$O_p(t_n)$, $t_n=\sqrt{J_n/n}+\rho_n+J_n^{-5}$, for this augmented regression. The derivative Gram bound then yields
\[
 \|\nabla\hat b(\hat H)-d_0(H)\|_{n,k}
 \leq C J_n O_p(t_n)+C J_n^{-4}+C\|\hat H-H\|_{n,k}.
\]
This proves \eqref{eq: structural derivative rate}.

\emph{First-stage loading.} Lemma~\ref{lemma: structural approximation} makes $T_iR_i d_0(H_i)'\partial_\eta H_i$ bounded, with conditional mean $\Gamma_k^{\rm str}(X_i)$. Together with the representer conditions in Assumption~\ref{ass: penalized sieve}(vi)--(vii), imposed by Assumption~\ref{ass: structural inference}, these properties give \eqref{eq: structural first stage expansion} by the proof of Proposition~\ref{prop: penalized expansion}.

The residual and derivative bounds also give $\|\hat h^{\rm str}-h^{\rm str}\|_{n,k}=o_p(1)$. The empirical first-stage Gram bound and
$\|(P_n-P)\Psi_nh^{{\rm str}\prime}\|=O_p(\sqrt{q_n/n})$
imply $\|\hat\gamma^{\rm str}-\gamma_n^{\rm str}\|=o_p(1)$.
The information, penalty, representer, and score bounds in the proof of Lemma~\ref{lemma: penalized variance} then give the first assertion in \eqref{eq: structural influence consistency} for this loading.

For the outcome part, conditional on all selection data,
$P_nE[T_i\|(\hat R_i-R_i)\varepsilon_i\|^2\mid\{D_j,X_j\}_j]
\leq C\|\hat R-R\|_{n,k}^2=o_p(1)$.
Markov's inequality, the augmented regression bound $\|\hat\varepsilon-\varepsilon\|_{n,k}=o_p(1)$, and boundedness of $\hat R$ prove the second assertion.

\emph{Differentiating the least-squares estimator.} The additional loading is the sample average of
$\Psi_n T_i\hat\varepsilon_i\nabla\hat m^{\rm str}(\hat H_i)'\partial_\eta H_i$, up to its immaterial sign. The projection derivative and $\partial_\eta H_i$ are bounded by the preceding results. Conditional on selection data, the outcome-error part has squared mean norm $O(q_n/n)$: apply the contraction property of the residual-maker of $[X,\hat Z]$ separately to each coefficient of this loading. Its conditional-mean remainder is $O_p(\rho_n+J_n^{-5})$ by Cauchy--Schwarz and the first-stage Gram bound. The additional loading is therefore $o_p(1)$ and does not change \eqref{eq: structural influence consistency}.
\end{proof}

Proposition~\ref{prop: structural feasible inference} supplies the bounds used in the proof of Theorem~\ref{thm: structural inference}. Its residual bounds control the remainders in the coefficient expansion, its first-stage expansion accounts for threshold estimation, and mean-square consistency of the estimated influence establishes variance consistency.

\begin{proof}[Proof of Theorem~\ref{thm: structural inference}]
Proposition~\ref{prop: structural feasible inference} gives $\hat Q^{\rm str}=P_nT_i\hat R_i\hat R_i'\to_pQ^{\rm str}$. Since $\ell_n(\hat H)$ belongs to the span of the integrated spline terms at the estimated thresholds, the Frisch--Waugh--Lovell identity yields
\[
 \hat Q^{\rm str}(\hat\beta_k^{\rm str}-\beta_{k0})
 =P_nT_i\hat R_i\{\varepsilon_i+\alpha_{k0}+b_k(H_i)-\ell_n(\hat H_i)\}.
\]
Taylor expansion over the threshold domain, using Lemma~\ref{lemma: structural approximation}, gives
\[
 \alpha_{k0}+b_k(H_i)-\ell_n(\hat H_i)
 =-d_0(H_i)'\Delta H_i+u_{in},\qquad
 |u_{in}|\leq C\{J_n^{-5}+J_n^{-4}\|\Delta H_i\|+\|\Delta H_i\|^2\}.
\]
The enlarged rectangle is convex and has widths bounded away from zero for large $n$; the upper threshold is held fixed in the top category. Bounded $\hat R$ makes the remainder
$O_p(J_n^{-5}+J_n^{-4}\rho_n+\rho_n^2)$.
Replacing $\hat R$ by $R$ in the drift costs $O_p(\bar c_n^{\rm str}\rho_n)$.
All these terms are $o_p(n^{-1/2})$ by \eqref{eq: structural joint rates}.
Conditional mean zero of $\varepsilon_i$ and conditional Chebyshev's inequality give
$n^{-1/2}\sum_iT_i(\hat R_i-R_i)\varepsilon_i=o_p(1)$.
Substitute \eqref{eq: structural first stage expansion} and apply the i.i.d. central limit theorem. The cross covariance is zero because $\chi_i^{\rm str}$ is measurable with respect to selection data. This proves \eqref{eq: structural asymptotic expansion}--\eqref{eq: structural asymptotic variance}.

For the feasible covariance, put
\begin{equation}\label{eq: structural feasible variance}
 \hat\psi_i=T_i\hat R_i\hat\varepsilon_i-\hat\chi_i^{\rm str},\qquad
 \hat{\mathcal V}_k^{\rm str}=(\hat Q^{\rm str})^{-1}
       P_n\{(\hat\psi_i-\bar\psi)(\hat\psi_i-\bar\psi)'\}
       (\hat Q^{\rm str})^{-1}.
\end{equation}
Equation~\eqref{eq: structural influence consistency}, Cauchy--Schwarz, and the law of large numbers give convergence of the middle matrix to the covariance of $T_iR_i\varepsilon_i-\chi_i^{\rm str}$. Inverse convergence completes the proof of variance consistency.
\end{proof}

\begin{remark}[Feasible growth rates]\label{rem: structural rates}
For natural cubic first-stage splines with three continuous arguments and $s\geq5/2$, $m_n\asymp n^{1/8}$ and $J_n\asymp n^{11/100}$ satisfy the first- and second-stage rate conditions. The slowest terms, $J_ne_n$ and $\sqrt n\bar c_n^{\rm str}\rho_n$, are both $O(n^{-3/200})$, while $\sqrt n(\omega_{vn}+\omega_{2n})\omega_{2n}=O(n^{-1/16})$. The support-boundary spacing condition holds, for example, with $J_n\in10\mathbb N$, knots $(j+1/2)/J_n$, $I_s=[0.1,0.3]$, and $I_t=[0.6,0.8]$ (omit $I_t$ for the top category).
\end{remark}

\subsection{Fixed-basis projection targets}\label{sec: fixed basis details}

\subsubsection{Projection target and influence function}

For a fixed second-stage spline basis and first-stage probability limit $g_k(x;\alpha_0)$, define
\begin{equation}\label{eq: fixed target}
 \theta_{kJ}=\argmin_\theta E[T_{ik}\{Y_i-W_{iJ}(\alpha_0)'\theta\}^2],
 \quad W_{iJ}(\alpha)=(X_i',B_J(g_k(X_i;\alpha))')'.
\end{equation}
Let $G_{kJ}=E[T_{ik}W_{iJ}(\alpha_0)W_{iJ}(\alpha_0)']$ be nonsingular. If $g_k(\cdot;\alpha_0)=g_{k0}$ and $U_{iJ}$ is the population residual of $X_i$ on $B_J(g_{k0}(X_i))$ conditional on $D_i=k$, the slope component has approximation bias
\begin{equation}\label{eq: fixed sieve bias}
 \beta_{kJ}-\beta_{k0}
 =\{E[T_{ik}U_{iJ}U_{iJ}']\}^{-1}
 E[T_{ik}U_{iJ}\lambda_{k0}(g_{k0}(X_i))].
\end{equation}
For fixed $J$, this gap must be zero for inference centered on $\beta_{k0}$.

Assume fixed dimensions and a first stage satisfying \eqref{eq: alpha influence}. Let $u_{iJ}=Y_i-W_{iJ}(\alpha_0)'\theta_{kJ}$. Require $W_{iJ}(\alpha)$ to be twice continuously differentiable near $\alpha_0$, with integrable envelopes for its derivatives and products with $Y_i$ sufficient for uniform laws of large numbers and Taylor expansion, and $E\|T_{ik}W_{iJ}(\alpha_0)u_{iJ}\|^2<\infty$. Taylor expansion of the least-squares normal equations gives the full influence representation \citep[Theorem~6.1]{neweymcfadden1994large}
\begin{equation}\label{eq: fixed basis expansion}
 \sqrt n(\hat\theta_{kJ}-\theta_{kJ})
 =G_{kJ}^{-1}\frac1{\sqrt n}\sum_i
 \{T_{ik}W_{iJ}(\alpha_0)u_{iJ}+C_{kJ}h_i\}+o_p(1),
\end{equation}
where, with $\dot W_{iJ}=\partial_\alpha W_{iJ}(\alpha_0)$,
\begin{equation}\label{eq: fixed basis derivative}
 C_{kJ}=E[T_{ik}\{\dot W_{iJ}u_{iJ}
       -W_{iJ}(\alpha_0)\theta_{kJ}'\dot W_{iJ}\}].
\end{equation}
\subsubsection{Supplementary result for common-index curves}\label{sec: structural curve support}

In the semiparametric common-index ordered model of Section~\ref{sec: ordered}, put $z=c_k-h(x)$ and $\Delta=c_{k+1}-c_k>0$ for an interior category. The probability thresholds lie on the curve
\[
 H(x)=\{F_\varepsilon(z),F_\varepsilon(z+\Delta)\}.
\]

\begin{lemma}[Fixed spline space on the common-index curve]\label{lemma: structural curve basis}
For an interior category of the semiparametric common-index model, put $T=\mathbf1\{D=k\}$. Suppose the observations are i.i.d., $E[T(\|X\|^2+Y^2)]<\infty$, $h(X)$ has support $\mathbb R$, and $Q_k\succ0$ with $g_{k0}=h$. Fix a finite centered cubic spline basis for $a_k$, using the full mean-zero spline space. Then, with $W_J=(1,X',Z_J(H)')'$, $E[TW_JW_J']$ is positive definite. Each basis function and its first two derivatives with respect to the thresholds extend continuously and boundedly to $0\leq s\leq t\leq1$, including $s=t$.
If the estimated thresholds satisfy $0\leq\hat H_{i1}\leq\hat H_{i2}\leq1$ and $\|\hat H-H\|_{n,T}=o_p(1)$, the least-squares coefficients converge to their fixed-space population projection $\theta_J$. Let $b_J$ denote the spline component of that projection. Then 
\[
 \|\nabla\hat b_J(\hat H)-\nabla b_J(H)\|_{n,T}=o_p(1).
\]
\end{lemma}

\begin{proof}
The category probability $F_\varepsilon(z+\Delta)-F_\varepsilon(z)$ is positive, so $h(X)$ also has support $\mathbb R$ conditional on $D=k$. Consider an intercept $c$ and a centered spline $a$ such that $c+(\mathcal K a)(H)=0$ almost surely in category $k$. Define
\[
 F(u)=cu+\int_0^u a(v)\,dv,\qquad q(z)=F(F_\varepsilon(z)).
\]
The equation gives $q(z+\Delta)=q(z)$, and continuity and full support extend it to every $z\in\mathbb R$. Since $q(z)\to F(1)$ as $z\to\infty$, periodicity makes $q$ constant. Its limit at $-\infty$ is $F(0)=0$. As $F_\varepsilon(\mathbb R)=(0,1)$, $F=0$ on $[0,1]$; centering then gives $c=F(1)=0$ and $a=0$.

Now suppose $c+X'b+Z_J(H)'d=0$ in category $k$. Subtracting its conditional expectation given $h(X)$ gives $b'U_k=0$, hence $b'Q_kb=0$. The rank condition implies $b=0$, and the preceding argument gives $c=d=0$. The moment condition therefore yields a finite, positive definite Gram matrix for each fixed $J$.

For a centered cubic spline $\varphi$, its averaged regressor is
\[
 Z_\varphi(s,t)=\int_0^1\varphi\{(1-v)s+vt\}\,dv.
\]
For $a+b\leq2$, differentiating this integral gives
\[
 \partial_s^a\partial_t^b Z_\varphi(s,t)
 =\int_0^1(1-v)^av^b\varphi^{(a+b)}\{(1-v)s+vt\}\,dv.
\]
The spline and its first two derivatives are bounded on $[0,1]$. The same bound applies even to a collapsing interval, without division by its width. The mean-square convergence of the estimated thresholds thus gives mean-square convergence of the basis and its first derivatives. The stated second moments, Cauchy--Schwarz, the law of large numbers, and the nonsingular fixed Gram matrix give coefficient consistency. Bounded first and second basis derivatives then give the displayed derivative convergence.
\end{proof}

Under the conditions for \eqref{eq: fixed basis expansion}, Lemma~\ref{lemma: structural curve basis} and mean-square consistency of the estimated first-stage influences yield a consistent estimator of the full two-stage covariance.

\clearpage
\section{Control-function linearity and nonlinearity}\label{sec: appendix cf figures}

\begin{figure}[htbp]
\centering
\includegraphics[width=0.85\textwidth]{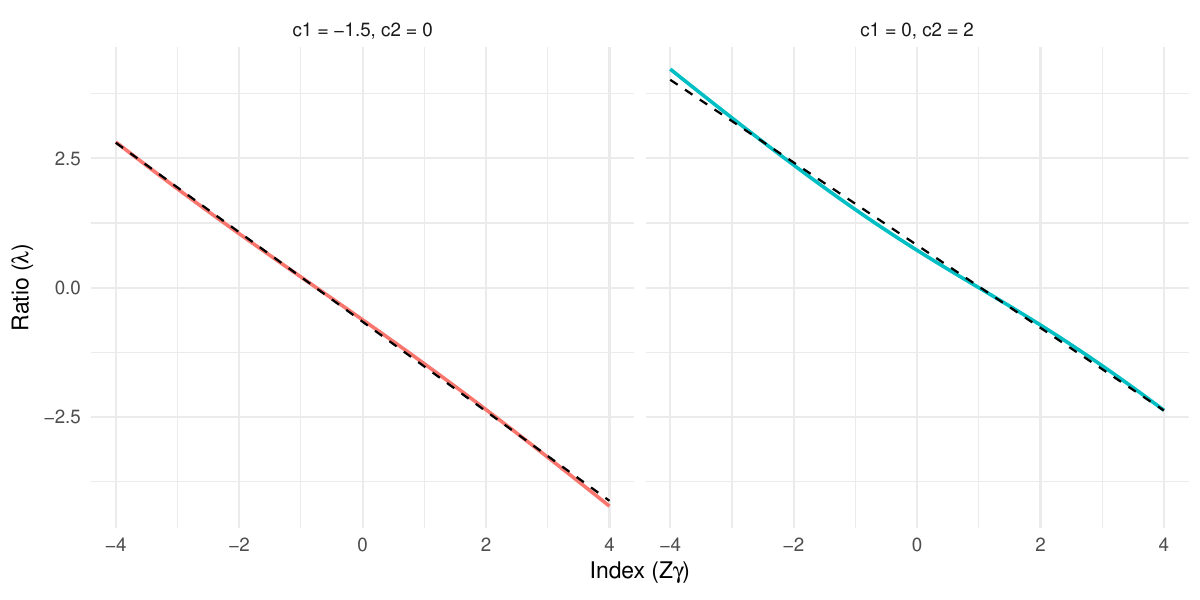}
\caption{Near-linearity of the ordered probit selection term. Solid curves plot the ratio in \eqref{eq: parametric bias} for the two threshold intervals shown; dashed lines are least-squares linear fits over the plotted range.}\label{fig: ordered linearity}
\end{figure}

\begin{figure}[htbp]
\centering
\includegraphics[width=0.85\textwidth]{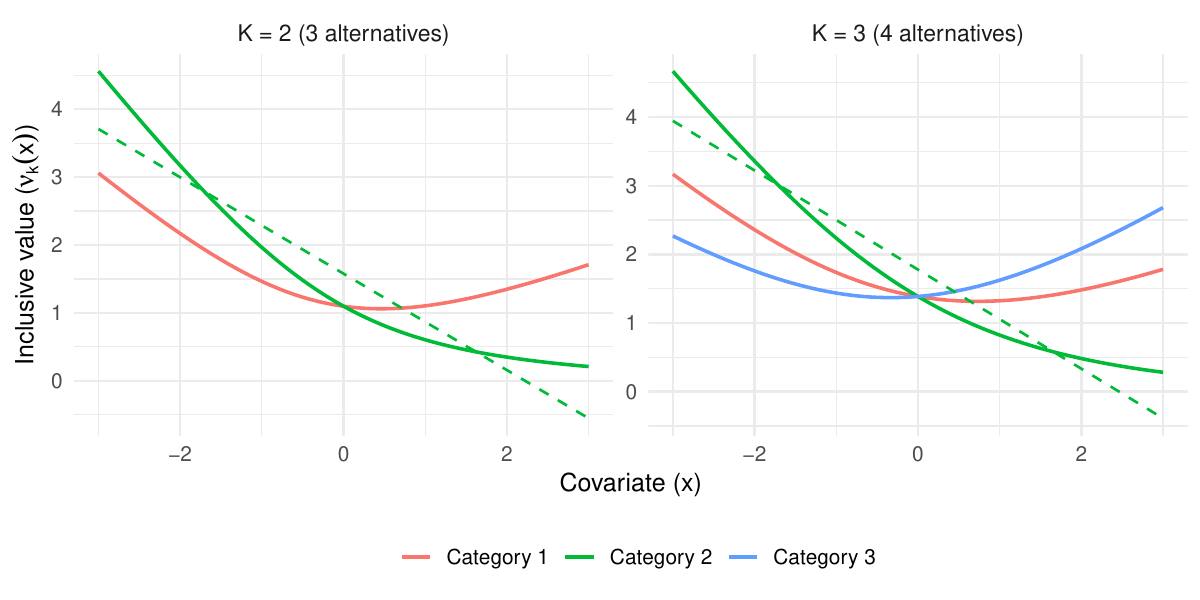}
\caption{Nonlinearity of the multinomial logit control index. Solid curves plot $\nu_k(x)=-\log p_k(x)$ for linear utilities $u_j(x)=x\gamma_j$, with $\gamma=(0,1,1.5)$ for $K=2$ and $\gamma=(0,1,1.5,0.7)$ for $K=3$. Dashed lines are least-squares linear fits for Category~2.}\label{fig: mnl nonlinearity}
\end{figure}

\clearpage
\section{Additional simulations}\label{sec: appendix additional sims}

\subsection{Ordered selection with a common index}\label{sec: appendix ordered common index}

The two supplementary ordered designs have $K=2$ and a common selection index. The sample size and number of replications follow Section~\ref{sec: simulations}. In DGP2, selection follows
$$D_i = k \quad \text{if} \quad c_k \leq 0.5 X_i - 0.5 X_i^2 + 0.2 X_i^3 + 0.5 X_i Z_i + Z_i - 0.5 Z_i^2 + U_i < c_{k+1},$$
with thresholds $c_1 = -1.5$ and $c_2 = 0.5$. The outcome equations are $Y_{i1} = 0.5 + 0.5 X_i + 0.25 Z_i + V_{i1}$ and $Y_{i2} = 0.6 + 0.7 X_i + 0.5 Z_i + V_{i2}$. The covariates $(X_i,Z_i)\sim N(\mathbf0,I_2)$ are independent of $(U_i, V_{1i}, V_{2i})$. The error terms are jointly normally distributed. The selection error has variance 1, the outcome errors have variances 4 and 6, and their covariances with the selection error are both 1.5. The second design (DGP3) modifies DGP2 by replacing the continuous covariate $Z_i$ with a binary indicator $Z_i = \mathbf{1}[Z_i' > 0]$ where $Z_i' \sim N(0,1)$, and enriching the selection index with interaction terms $X_i^2 Z_i$ and $X_i^3 Z_i$:
$$\tilde{h}(X_i, Z_i) = -0.2 X_i - 0.5 X_i^2 + 0.3 X_i^3 + 0.1 X_i Z_i + 0.5 Z_i - 0.3 X_i^2 Z_i + 0.2 X_i^3 Z_i.$$ The variances of outcome errors $(V_{1i}, V_{2i})$ are 1 and 1.5, and their covariances with the selection error are both 0.75.

\begin{table}[htbp]
\centering
\caption{Ordered selection DGP2 and DGP3}\label{tab: ordered results}
\small
\setlength{\tabcolsep}{3pt}
\renewcommand{\arraystretch}{1}
\begin{tabular*}{\textwidth}{@{\extracolsep{\fill}}llcccccccccc@{}}
\toprule
& & \multicolumn{5}{c}{DGP2: continuous $X,Z$} & \multicolumn{5}{c}{DGP3: continuous $X$, binary $Z$} \\
\cmidrule(lr){3-7}\cmidrule(lr){8-12}
& & OLS & Linear & Oracle & Sieve & \shortstack{Restricted\\Sieve} & OLS & Linear & Oracle & Sieve & \shortstack{Restricted\\Sieve} \\
\midrule
Occ. 1 & RMSE & 0.702 & 5.581 & 0.064 & 0.088 & 0.068 & 0.177 & 5.811 & 0.031 & 0.088 & 0.040 \\
 & $|$Bias$|$ & 0.700 & 2.026 & 0.001 & 0.003 & 0.003 & 0.175 & 5.265 & 0.002 & 0.017 & 0.010 \\
 & Coverage & 0.000 & 0.910 & 0.956 & 0.996 & 0.954 & 0.000 & 0.439 & 0.945 & 0.996 & 0.947 \\
\addlinespace
Occ. 2 & RMSE & 0.550 & 0.383 & 0.136 & 0.153 & 0.155 & 0.184 & 0.827 & 0.058 & 0.073 & 0.071 \\
 & $|$Bias$|$ & 0.542 & 0.110 & 0.004 & 0.016 & 0.014 & 0.177 & 0.765 & 0.001 & 0.013 & 0.014 \\
 & Coverage & 0.000 & 0.927 & 0.936 & 0.938 & 0.938 & 0.057 & 0.307 & 0.946 & 0.953 & 0.948 \\
\bottomrule
\end{tabular*}
\smallskip
\parbox{\textwidth}{\footnotesize\textit{Notes:} The reported numbers are averaged over the two slope coefficients within each category. 95\% confidence intervals are computed using analytic standard errors developed in Sections~\ref{sec: asymptotics} and~\ref{sec: structural ordered estimation} for Linear, Oracle, Sieve, and Restricted Sieve, and heteroskedasticity-robust standard errors for OLS.}
\end{table}

Table~\ref{tab: ordered results} shows that Restricted Sieve improves precision relative to Sieve in the interior category, with performance close to Oracle and coverage near the nominal level. The two sieve estimators perform similarly in the top category. DGP3 also illustrates the method with one continuous covariate and one binary covariate. 

Table~\ref{tab: common A scalar} reports the common-index version of DGP1, in which $h_k(X,Z,W)=\Phi(c_k-h(X,Z,W))$ with $(c_1,c_2,c_3)=(-1,0.5,1.5)$, so that all thresholds move with the single index $h$. Relative to Sieve, Restricted Sieve reduces RMSE by 11.0\% and 12.3\% in the two interior categories. The corresponding reductions are 29.2\% and 37.3\% when the two selection thresholds vary with different indices (Table~\ref{tab: ordered K3 baseline}). In the common-index design, the two estimators have the same reported RMSE in the top category. Because the interior threshold pair lies on a curve, the pivoted QR rule drops one or two numerically collinear tensor terms from the Sieve basis in most samples. Sieve coverage is 98.3\% and 98.7\% in the interior categories, against 94.9\% and 95.3\% for Restricted Sieve.

\begin{table}[htbp]
\centering
\caption{Common-index version of Ordered DGP1}\label{tab: common A scalar}

\small
\setlength{\tabcolsep}{2.5pt}
\renewcommand{\arraystretch}{1}
\begin{tabular*}{\textwidth}{@{\extracolsep{\fill}}lccccccccc@{}}
\toprule
& \multicolumn{3}{c}{Category 1} & \multicolumn{3}{c}{Category 2} & \multicolumn{3}{c}{Category 3} \\
\cmidrule(lr){2-4}\cmidrule(lr){5-7}\cmidrule(lr){8-10}
Estimator & RMSE & $|$Bias$|$ & Coverage & RMSE & $|$Bias$|$ & Coverage & RMSE & $|$Bias$|$ & Coverage \\
\midrule
OLS & 0.111 & 0.105 & 0.318 & 0.191 & 0.183 & 0.232 & 0.199 & 0.188 & 0.297 \\
Linear & 8.739 & 5.589 & 0.848 & 21.081 & 1.515 & 0.945 & 0.348 & 0.062 & 0.932 \\
Oracle & 0.038 & 0.002 & 0.946 & 0.050 & 0.001 & 0.955 & 0.054 & 0.002 & 0.968 \\
Sieve & 0.043 & 0.004 & 0.983 & 0.059 & 0.002 & 0.987 & 0.058 & 0.002 & 0.959 \\
Restricted Sieve & 0.038 & 0.003 & 0.949 & 0.052 & 0.002 & 0.953 & 0.058 & 0.002 & 0.958 \\
\bottomrule
\end{tabular*}
\smallskip
\parbox{\textwidth}{\footnotesize\textit{Notes:} $n=5{,}000$, $R=500$. The reported numbers are averaged over the three slope coefficients within each category. Oracle estimates an ordered probit with constant cutpoints and a quadratic index. The 95\% confidence intervals use heteroskedasticity-robust standard errors for OLS and the two-stage standard errors of Sections~\ref{sec: asymptotics} and~\ref{sec: structural ordered estimation} for the selection-corrected estimators.}
\end{table}

\subsection{Weak identification: reducing the DGP's nonlinear variation}\label{sec: appendix weak identification}

To examine how weaker identifying nonlinearity affects estimation, I scale the quadratic and interaction terms in Ordered DGP1 by $\delta$:
\[
  h(X, Z, W) \;=\; 0.5\, X \;+\; 0.3\, Z \;+\; 0.8\, W \;+\; \delta\cdot(-0.5\, X^2 \;+\; 0.2\, Z^2 \;-\; 0.4\, X Z).
\]
All other design features and the estimation bases remain unchanged. Table~\ref{tab: ordered K3} reports results for $\delta\in\{1,0.7,0.3,0.1\}$, with $\delta=1$ reproducing the main-text design. Reducing $\delta$ weakens the nonlinear variation used to distinguish the outcome slopes from selection, while residual rank remains positive for every $\delta>0$.

\begin{table}[htbp]
\centering
\caption{Ordered selection with $K=3$: variation in identifying nonlinearity}\label{tab: ordered K3}
\small
\setlength{\tabcolsep}{2pt}
\renewcommand{\arraystretch}{1}
\begin{tabular*}{\textwidth}{@{\extracolsep{\fill}}lcccccccccccc@{}}
\toprule
& \multicolumn{4}{c}{RMSE} & \multicolumn{4}{c}{$|$Bias$|$} & \multicolumn{4}{c}{Coverage} \\
\cmidrule(lr){2-5}\cmidrule(lr){6-9}\cmidrule(lr){10-13}
& $\delta=1$ & $0.7$ & $0.3$ & $0.1$ & $\delta=1$ & $0.7$ & $0.3$ & $0.1$ & $\delta=1$ & $0.7$ & $0.3$ & $0.1$ \\
\midrule
\multicolumn{13}{l}{\textit{Category 1}} \\
OLS & 0.087 & 0.097 & 0.116 & 0.123 & 0.079 & 0.089 & 0.109 & 0.116 & 0.473 & 0.374 & 0.293 & 0.300 \\
Linear & 6.835 & 4.772 & 2.803 & 2.378 & 4.458 & 3.436 & 1.879 & 1.299 & 0.835 & 0.821 & 0.881 & 0.925 \\
Oracle & 0.034 & 0.036 & 0.046 & 0.064 & 0.003 & 0.003 & 0.003 & 0.002 & 0.942 & 0.940 & 0.954 & 0.951 \\
Sieve & 0.049 & 0.053 & 0.080 & 0.168 & 0.005 & 0.004 & 0.004 & 0.032 & 0.952 & 0.954 & 0.959 & 0.961 \\
Restricted Sieve & 0.035 & 0.037 & 0.050 & 0.073 & 0.003 & 0.003 & 0.004 & 0.010 & 0.943 & 0.941 & 0.951 & 0.965 \\
\addlinespace
\multicolumn{13}{l}{\textit{Category 2}} \\
OLS & 0.141 & 0.162 & 0.196 & 0.215 & 0.131 & 0.155 & 0.192 & 0.211 & 0.283 & 0.127 & 0.014 & 0.005 \\
Linear & 48.508 & 48.847 & 38.816 & 29.716 & 40.307 & 43.877 & 35.664 & 26.815 & 0.723 & 0.501 & 0.395 & 0.540 \\
Oracle & 0.049 & 0.050 & 0.056 & 0.063 & 0.002 & 0.002 & 0.002 & 0.001 & 0.951 & 0.947 & 0.945 & 0.949 \\
Sieve & 0.082 & 0.084 & 0.103 & 0.143 & 0.011 & 0.013 & 0.019 & 0.031 & 0.970 & 0.977 & 0.983 & 0.985 \\
Restricted Sieve & 0.051 & 0.051 & 0.058 & 0.067 & 0.002 & 0.002 & 0.002 & 0.002 & 0.951 & 0.947 & 0.953 & 0.951 \\
\addlinespace
\multicolumn{13}{l}{\textit{Category 3}} \\
OLS & 0.259 & 0.281 & 0.288 & 0.285 & 0.247 & 0.272 & 0.280 & 0.279 & 0.239 & 0.137 & 0.078 & 0.070 \\
Linear & 1.514 & 1.635 & 1.238 & 0.866 & 1.216 & 1.428 & 1.096 & 0.730 & 0.711 & 0.557 & 0.501 & 0.528 \\
Oracle & 0.074 & 0.081 & 0.089 & 0.091 & 0.004 & 0.005 & 0.005 & 0.004 & 0.947 & 0.947 & 0.950 & 0.941 \\
Sieve & 0.079 & 0.086 & 0.093 & 0.099 & 0.006 & 0.009 & 0.014 & 0.016 & 0.942 & 0.951 & 0.949 & 0.946 \\
Restricted Sieve & 0.079 & 0.085 & 0.094 & 0.099 & 0.006 & 0.009 & 0.015 & 0.018 & 0.947 & 0.954 & 0.950 & 0.944 \\
\bottomrule
\end{tabular*}
\smallskip
\parbox{\textwidth}{\footnotesize\textit{Notes:} $n=5{,}000$, $R=500$ per $\delta$, with common random numbers. The reported numbers are averaged over the three slope coefficients within each category. The 95\% confidence intervals use heteroskedasticity-robust standard errors for OLS and the two-stage standard errors of Sections~\ref{sec: asymptotics} and~\ref{sec: structural ordered estimation} for the selection-corrected estimators.}
\end{table}

For $\delta\in\{0.3,0.7,1\}$, average absolute bias remains below 0.02 for Sieve and Restricted Sieve in every category. As the index approaches linearity, less covariate variation remains after conditioning on the selection controls. At $\delta=0.1$, for interior categories, Sieve RMSE and absolute bias rise sharply. Oracle remains nearly unbiased, although its RMSE also increases as nonlinearity weakens. Restricted Sieve keeps absolute bias close to 0 and its RMSE remains close to that of Oracle. Oracle and Restricted Sieve maintain coverage near the nominal level, whereas Sieve is conservative, particularly in the second category. For the top category, Sieve and Restricted Sieve perform similarly.

\subsection{Multinomial selection without exchangeability}\label{sec: appendix approximation sensitivity}

In this design (Multinomial DGP3), I compare the estimators beyond the own-shock and exchangeability restrictions. This DGP violates IIA, exchangeability, and the own-shock restriction all together while keeping $K = 3$. The preference shocks follow a single-factor model: $\varepsilon_{ij} = \lambda_j \, f_i + u_{ij},\ \forall j \in \{0, 1, 2, 3\},$ where $f_i \sim N(0,1)$ is a common factor, $u_{ij} \sim N(0,1)$ are independent. As the loadings $\boldsymbol{\lambda} = (0, 0.3, 0.8, -0.5)$ are \emph{heterogeneous} across alternatives, the factor does not cancel in utility differences: $\text{Cov}(\varepsilon_{ij} - \varepsilon_{ik}, \varepsilon_{il} - \varepsilon_{ik}) = (\lambda_j - \lambda_k)(\lambda_l - \lambda_k)+1+\mathbf{1}[j=l]$ for $j,l\ne k$ depends on the identities of $j$ and $l$, violating exchangeability. The covariates, utility specifications, and linear outcome coefficients are identical to DGP1--2. The outcome error is $V_{ik}=\varepsilon_{ik}+2\varepsilon_{ik}(\varepsilon_{ik}-\bar\varepsilon_{i,-k})+0.2\eta_{ik}$, where $\bar\varepsilon_{i,-k}=K^{-1}\sum_{j\ne k}\varepsilon_{ij}$ and $\eta_{ik}\sim N(0,1)$ is independent of the other variables.

Here the Exch-$L2$ variant uses the first two elementary symmetric polynomials of the estimated log odds $\log(\hat p_j/\hat p_k)$ for $j\ne k$. For this estimator, nearly redundant linear combinations of the second-stage spline terms are removed before estimation. Oracle-P approximates choice probabilities by simulating the known selection model and uses the same second-stage spline basis as Sieve to approximate the selection correction function. Computing the true correction function would additionally require the first and second moments of the correlated preference shocks conditional on category choice. Oracle-P avoids this additional calculation.

\begin{table}[htbp]
\centering
\caption{Multinomial DGP3: approximation and specification sensitivity}\label{tab: multinomial dgp4 baseline}

\small
\setlength{\tabcolsep}{3pt}
\renewcommand{\arraystretch}{1}
\begin{tabular*}{\textwidth}{@{\extracolsep{\fill}}lccccc@{}}
\toprule
Statistic & OLS & MLogit & Oracle-P & Sieve & Exch-$L2$ (log odds) \\
\midrule
RMSE & 0.762 & 0.222 & 0.299 & 0.427 & 0.235 \\
$|$Bias$|$ & 0.733 & 0.045 & 0.056 & 0.073 & 0.035 \\
Coverage & 0.132 & 0.946 & 0.940 & 0.989 & 0.960 \\
\bottomrule
\end{tabular*}
\smallskip
\parbox{\textwidth}{\footnotesize\textit{Notes:} The reported numbers are averaged over all slope coefficients and categories. MLogit, Sieve, and Exch-$L2$ use two-stage standard errors (Section~\ref{sec: asymptotics}). Oracle-P uses simulated choice probabilities and adjusts its standard errors for the estimated spline knots. OLS uses heteroskedasticity-robust standard errors.}
\end{table}

Table~\ref{tab: multinomial dgp4 baseline} shows that all selection correction methods substantially reduce the bias relative to OLS. MLogit achieves the lowest RMSE, while Exch-$L2$ has the smallest absolute bias and a similar RMSE despite the failure of exchangeability. Sieve has substantially higher RMSE and conservative coverage.

\subsection{Identifying variation in the simulation designs}\label{sec: simulation rank proofs}

Positive selection probabilities preserve the covariate support in each design. To verify positive residual rank conditional on the true controls $T(X_i)$, it suffices to show that $b'X_i=g(T(X_i))$ almost surely conditional on $D_i=k$ implies $b=0$. Conditioning on a summary of $T(X_i)$ leaves at least as much covariate variation as conditioning on $T(X_i)$ itself, so positive residual rank is preserved.

\medskip\noindent\textbf{Ordered designs.}
In DGP1 and its variants, all true controls are functions of $(h,S)$, where $S=0.8Z+0.4W^2$ and $h$ is the index in Appendix~\ref{sec: appendix weak identification}. Tangents to their joint level sets at $(0,0,0)$, $(1,0,0)$ and $(0,0,1)$ are proportional to
\[
 (-0.8,0,0.5)',\qquad(-0.8,0,0.5-\delta)',\qquad(-0.5,-0.5,0.5)'.
\]
Their determinant is $0.4\delta$. By the local-coordinate argument in the proof of Proposition~\ref{prop: derivative rank}, a vector $b$ satisfying $b'X_i=g(h,S)$ must be orthogonal to all three, hence zero for every $\delta>0$. This verifies rank for the interior and top categories, including the common-index variant.

In DGP2, the level-set tangents $(h_z,-h_x)$ at $(0,0)$ and $(0,1)$ are $(1,-0.5)$ and $(0,-1)$, so the same argument applies. In DGP3, the index within each binary stratum is cubic, has two distinct stationary points and maps onto $\mathbb R$. Multiple regular preimages within a stratum force $b_X=0$; the overlapping index ranges across strata then force $b_Z=0$. Since each cumulative threshold is a one-to-one function of the common index, these checks cover both categories in DGPs 2--3.

\medskip\noindent\textbf{Multinomial designs.}
DGPs 1--3 share the utility vector $F(x)=a+Ax+Bq(x)$, where $x=(X,Z,W)'$, $q(x)=(X^2,Z^2,W^2,XZ,XW,ZW)'$, $a=(0.5,0.4,0.3)'$, and
\[
 A=\begin{pmatrix}0.4&0.4&0.3\\0.6&0.5&-0.2\\-0.3&-0.4&0.5\end{pmatrix},\qquad
 B=\begin{pmatrix}
 -0.5&-0.3&0.2&0.15&0.15&-0.10\\
 -0.4&-0.25&0.35&0.20&-0.15&0.12\\
 0.35&0.20&-0.3&-0.15&0.20&-0.20
 \end{pmatrix}.
\]
Let $e_j$ be the $j$th coordinate vector and $x_j^\pm=m_j\pm e_j/2$, with
\[
 (m_1,m_2,m_3)=\begin{pmatrix}
 -5/7&-104/59&53/38\\
 -320/49&15/59&335/76\\
 -44/49&-10/59&-13/76
 \end{pmatrix}.
\]
Direct substitution gives $DF(m_j)e_j=0$ and $\det DF(x_j^\pm)\ne0$. Since $F$ is quadratic, $F(x_j^+)=F(x_j^-)$. If $b'x=g(F(x))$ held almost surely, the inverse function theorem and positive selected densities would give $b'\phi_j^+(u)=b'\phi_j^-(u)$ almost everywhere on the common image of the two local inverse maps. Continuity of their difference then implies $b'(x_j^+-x_j^-)=b_j=0$ for every $j$. Thus residual rank is positive conditional on $F(X_i)$, and consequently on the choice probabilities or their scalar and symmetric summaries.

\section{Additional empirical results}\label{sec: appendix robustness}

\subsection{Monthly earnings}\label{sec: monthly earnings}

Tables~\ref{tab:final_main_monthly}--\ref{tab: dynamics monthly} report the monthly earnings results using the same selection models as for hourly wages. The wage equations control separately for weekly regular and overtime hours worked. I exclude observations if weekly regular hours are missing or outside \((0,100]\), or if weekly overtime exceeds 60 hours. I code missing or negative overtime hours as zero. The Sieve female coefficient is more negative than for hourly wages in every category ($-0.023$ versus $0.008$ in large firms), while the patterns are similar to those for hourly wages.

\begin{table}[htbp]
\centering
\caption{Female coefficients in log monthly-wage regressions}\label{tab:final_main_monthly}
\small
\setlength{\tabcolsep}{3pt}
\begin{tabular*}{\textwidth}{@{\extracolsep{\fill}}lcclcccc@{}}
\toprule
Estimator & \multicolumn{2}{c}{Firm size} & Estimator & \multicolumn{2}{c}{Field} & \multicolumn{2}{c}{Sector} \\
\cmidrule(lr){2-3}\cmidrule(lr){5-6}\cmidrule(lr){7-8}
 & SME & Large firm &  & Non-STEM & STEM & Public & Private \\
\midrule
OLS & -0.079 & -0.078 & OLS & -0.108 & -0.044 & -0.039 & -0.118 \\
 & (0.003) & (0.004) &  & (0.003) & (0.005) & (0.005) & (0.003) \\
\addlinespace[2pt]
Sieve & -0.078 & -0.023 & Sieve & -0.098 & -0.042 & -0.050 & -0.119 \\
 & (0.005) & (0.007) &  & (0.004) & (0.006) & (0.007) & (0.003) \\
\addlinespace[2pt]
Restricted Sieve & -0.074 & -0.023 & MLogit & -0.105 & -0.045 & -0.038 & -0.120 \\
 & (0.004) & (0.007) &  & (0.004) & (0.006) & (0.006) & (0.003) \\
\bottomrule
\end{tabular*}
\smallskip
\parbox{0.98\textwidth}{\footnotesize\textit{Notes:} The wage regressions control separately for regular and overtime hours per week. Parentheses contain heteroskedasticity-robust standard errors for OLS and analytic two-stage standard errors for the selection-corrected estimators.}
\end{table}

\begin{table}[htbp]
\centering
\caption{Accounting decomposition of the gender gap in log monthly wages}\label{tab: decomposition monthly}
\small
\setlength{\tabcolsep}{3pt}
\begin{tabular*}{\textwidth}{@{\extracolsep{\fill}}l@{\hspace{1em}}cccccc@{}}
\toprule
\shortstack{Selection\\structure} & \shortstack{Raw\\gap} & \shortstack{Gap after\\covariate\\adjustment\\(OLS)} & \shortstack{Covariate\\contribution} & \shortstack{Employment\\composition} & \shortstack{Unexplained\\gap\\(Sieve)} & \shortstack{Selection\\correction} \\
\midrule
Firm size & +0.192 & +0.079 & +0.093 & +0.020 & +0.052 & +0.026 \\
 & (0.002) & (0.003) & (0.002) & (0.0006) & (0.004) & (0.003) \\
Field & +0.188 & +0.087 & +0.090 & +0.011 & +0.080 & +0.007 \\
 & (0.002) & (0.003) & (0.002) & (0.0004) & (0.004) & (0.003) \\
Sector & +0.189 & +0.103 & +0.084 & +0.002 & +0.106 & -0.003 \\
 & (0.002) & (0.003) & (0.002) & (0.0002) & (0.003) & (0.002) \\
\bottomrule
\end{tabular*}
\smallskip
\parbox{0.98\textwidth}{\footnotesize\textit{Notes:} Selection correction is the OLS minus Sieve gap after adjustment. The sorting contribution is the sum of selection correction and employment composition. Parentheses contain analytic standard errors accounting jointly for estimated shares, means, and coefficients.}
\end{table}

\begin{table}[htbp]
\centering
\caption{Female log monthly-wage coefficients by subperiod}\label{tab: dynamics monthly}
\footnotesize
\setlength{\tabcolsep}{3pt}
\begin{tabular*}{\textwidth}{@{\extracolsep{\fill}}lcccccc@{}}
\toprule
 & \multicolumn{2}{c}{Firm size} & \multicolumn{2}{c}{Field} & \multicolumn{2}{c}{Sector} \\
\cmidrule(lr){2-3}\cmidrule(lr){4-5}\cmidrule(lr){6-7}
 & SME & Large firm & Non-STEM & STEM & Public & Private \\
\midrule
\multicolumn{7}{l}{\textit{OLS}} \\
2008--2011 & -0.096 (0.005) & -0.087 (0.008) & -0.121 (0.005) & -0.059 (0.008) & -0.040 (0.010) & -0.124 (0.005) \\
2012--2015 & -0.079 (0.005) & -0.063 (0.007) & -0.101 (0.005) & -0.039 (0.009) & -0.032 (0.009) & -0.117 (0.005) \\
2016--2019 & -0.056 (0.005) & -0.076 (0.007) & -0.092 (0.005) & -0.034 (0.008) & -0.043 (0.008) & -0.108 (0.005) \\
\addlinespace
\multicolumn{7}{l}{\textit{Sieve}} \\
2008--2011 & -0.083 (0.007) & -0.039 (0.010) & -0.109 (0.008) & -0.066 (0.013) & -0.039 (0.012) & -0.118 (0.006) \\
2012--2015 & -0.089 (0.008) & -0.015 (0.011) & -0.111 (0.007) & -0.046 (0.012) & -0.052 (0.013) & -0.132 (0.006) \\
2016--2019 & -0.052 (0.008) & -0.027 (0.013) & -0.075 (0.007) & -0.025 (0.010) & -0.047 (0.010) & -0.105 (0.006) \\
\bottomrule
\end{tabular*}
\smallskip
\parbox{0.98\textwidth}{\footnotesize\textit{Notes:} Each subperiod is estimated separately. Parentheses contain heteroskedasticity-robust standard errors for OLS and analytic two-stage standard errors for Sieve.}
\end{table}

\subsection{Wage regression coefficients}

Tables~\ref{tab:final_full_coefficients_hourly} and~\ref{tab:final_full_coefficients_monthly} report the Sieve estimates for the female indicator, age, GPA, parental income, and four-year college attendance, with analytic two-stage standard errors. The monthly earnings table also reports the coefficients on regular and overtime hours. 

\begin{table}[htbp]
\centering
\caption{Selection-corrected hourly-wage coefficients}
\label{tab:final_full_coefficients_hourly}
\small
\setlength{\tabcolsep}{5pt}
\begin{tabular*}{\textwidth}{@{\extracolsep{\fill}}lrrrrrr@{}}
\hline
Variable & SME & Large firm & Non-STEM & STEM & Public & Private \\
\hline
Female & -0.036 & 0.008 & -0.058 & -0.016 & -0.011 & -0.078 \\
 & (0.005) & (0.007) & (0.004) & (0.007) & (0.007) & (0.004) \\
Age & 0.023 & 0.042 & 0.024 & 0.026 & 0.028 & 0.021 \\
 & (0.001) & (0.001) & (0.001) & (0.002) & (0.002) & (0.001) \\
GPA & 0.00136 & 0.00105 & 0.00243 & 0.00167 & 0.00104 & 0.00314 \\
 & (0.00021) & (0.00032) & (0.00017) & (0.00025) & (0.00033) & (0.00020) \\
Parental income & 0.010 & 0.007 & 0.014 & 0.008 & 0.007 & 0.013 \\
 & (0.001) & (0.001) & (0.001) & (0.001) & (0.001) & (0.001) \\
Four-year college & 0.105 & -0.010 & 0.112 & 0.127 & 0.019 & 0.133 \\
 & (0.006) & (0.008) & (0.005) & (0.007) & (0.014) & (0.007) \\
Major effects & Yes & Yes & Yes & Yes & Yes & Yes \\
Founding-type effects & Yes & Yes & Yes & Yes & Yes & Yes \\
School region effects & Yes & Yes & Yes & Yes & Yes & Yes \\
Year effects & Yes & Yes & Yes & Yes & Yes & Yes \\
$N$ & 67,397 & 48,241 & 86,802 & 34,156 & 25,861 & 94,042 \\
\hline
\end{tabular*}
\par\smallskip\begin{minipage}{0.98\linewidth}\footnotesize Notes: Analytic two-stage standard errors in parentheses. Parental income is in millions of 2015 KRW per month. The intercept is omitted.\end{minipage}
\end{table}

\begin{table}[htbp]
\centering
\caption{Selection-corrected monthly-wage coefficients}
\label{tab:final_full_coefficients_monthly}
\small
\setlength{\tabcolsep}{5pt}
\begin{tabular*}{\textwidth}{@{\extracolsep{\fill}}lrrrrrr@{}}
\hline
Variable & SME & Large firm & Non-STEM & STEM & Public & Private \\
\hline
Female & -0.078 & -0.023 & -0.098 & -0.042 & -0.050 & -0.119 \\
 & (0.005) & (0.007) & (0.004) & (0.006) & (0.007) & (0.003) \\
Age & 0.020 & 0.041 & 0.021 & 0.028 & 0.026 & 0.019 \\
 & (0.001) & (0.001) & (0.001) & (0.002) & (0.002) & (0.001) \\
GPA & 0.00110 & 0.00040 & 0.00200 & 0.00157 & 0.00064 & 0.00298 \\
 & (0.00019) & (0.00030) & (0.00017) & (0.00022) & (0.00031) & (0.00019) \\
Parental income & 0.011 & 0.007 & 0.014 & 0.011 & 0.006 & 0.014 \\
 & (0.001) & (0.001) & (0.001) & (0.001) & (0.001) & (0.001) \\
Four-year college & 0.088 & -0.030 & 0.090 & 0.125 & -0.005 & 0.124 \\
 & (0.005) & (0.007) & (0.005) & (0.006) & (0.012) & (0.007) \\
Regular weekly hours & 0.0078 & 0.0052 & 0.0085 & 0.0006 & 0.0119 & 0.0048 \\
 & (0.0003) & (0.0004) & (0.0003) & (0.0004) & (0.0006) & (0.0002) \\
Overtime weekly hours & 0.0043 & 0.0065 & 0.0067 & 0.0045 & 0.0056 & 0.0057 \\
 & (0.0002) & (0.0003) & (0.0002) & (0.0003) & (0.0004) & (0.0002) \\
Major effects & Yes & Yes & Yes & Yes & Yes & Yes \\
Founding-type effects & Yes & Yes & Yes & Yes & Yes & Yes \\
School region effects & Yes & Yes & Yes & Yes & Yes & Yes \\
Year effects & Yes & Yes & Yes & Yes & Yes & Yes \\
$N$ & 67,388 & 48,227 & 86,790 & 34,143 & 25,858 & 94,021 \\
\hline
\end{tabular*}
\par\smallskip\begin{minipage}{0.98\linewidth}\footnotesize Notes: Analytic two-stage standard errors in parentheses. Parental income is in millions of 2015 KRW per month. The intercept is omitted.\end{minipage}
\end{table}

\subsection{Selection probabilities}\label{sec: appendix overlap}\label{sec: appendix cf tests}

Table~\ref{tab: overlap} shows how often the estimated choice probabilities are near zero. Probabilities below 0.01 are most common for STEM employment: 0.95\% of men and 0.33\% of women in the field sample have an estimated STEM probability below this threshold.

\begin{table}[htbp]
\centering
\caption{Estimated category probabilities by sex}\label{tab: overlap}
\small
\setlength{\tabcolsep}{5pt}
\begin{tabular*}{\textwidth}{@{\extracolsep{\fill}}llrrrrr@{}}
\toprule
Category & Sex & Minimum & 1st pct. & Median & 99th pct. & $p<0.01$ (\%) \\
\midrule
\multicolumn{7}{l}{\textit{Firm size}} \\
Outside & Men & 0.090 & 0.180 & 0.345 & 0.587 & 0.00 \\
Outside & Women & 0.093 & 0.156 & 0.399 & 0.598 & 0.00 \\
SME & Men & 0.004 & 0.109 & 0.344 & 0.592 & 0.01 \\
SME & Women & 0.057 & 0.200 & 0.380 & 0.668 & 0.00 \\
Large firm & Men & 0.013 & 0.067 & 0.301 & 0.579 & 0.00 \\
Large firm & Women & 0.011 & 0.037 & 0.210 & 0.470 & 0.00 \\
\addlinespace
\multicolumn{7}{l}{\textit{Field}} \\
Outside & Men & 0.089 & 0.178 & 0.334 & 0.570 & 0.00 \\
Outside & Women & 0.078 & 0.152 & 0.388 & 0.577 & 0.00 \\
Non-STEM & Men & 0.042 & 0.124 & 0.428 & 0.743 & 0.00 \\
Non-STEM & Women & 0.049 & 0.076 & 0.516 & 0.724 & 0.00 \\
STEM & Men & $3.1\times10^{-4}$ & 0.010 & 0.193 & 0.663 & 0.95 \\
STEM & Women & 0.001 & 0.012 & 0.035 & 0.760 & 0.33 \\
\addlinespace
\multicolumn{7}{l}{\textit{Sector}} \\
Outside & Men & 0.093 & 0.172 & 0.338 & 0.572 & 0.00 \\
Outside & Women & 0.092 & 0.155 & 0.391 & 0.584 & 0.00 \\
Public & Men & 0.009 & 0.016 & 0.103 & 0.545 & 0.00 \\
Public & Women & 0.012 & 0.030 & 0.110 & 0.631 & 0.00 \\
Private & Men & 0.018 & 0.098 & 0.549 & 0.789 & 0.00 \\
Private & Women & 0.061 & 0.118 & 0.477 & 0.784 & 0.00 \\
\addlinespace
\bottomrule
\end{tabular*}
\smallskip
\parbox{0.98\textwidth}{\footnotesize\textit{Notes:} Estimated choice probabilities are evaluated over all observations of each sex in the corresponding selection sample, including the outside category.}
\end{table}

\FloatBarrier
\subsection{Specification sensitivity and calibrated coverage}\label{sec: appendix calibrated}

\paragraph{Joint first-stage and second-stage flexibility.}
Table~\ref{tab: firststage grid} varies the spline basis sizes in both the selection and wage equations. Across the nine specifications, the Sieve female coefficients for hourly wages vary by at most about 0.003 within any category. All confidence intervals in the nine specifications include zero for large firms and public employment, while all exclude zero in the other four categories. For monthly wage, all confidence intervals exclude zero in every category. 

\begin{table}[htbp]
\centering
\caption{Female coefficients across first-stage and second-stage flexibility}
\label{tab: firststage grid}
\footnotesize
\setlength{\tabcolsep}{3.5pt}
\begin{tabular}{lrrrrrrrrrr}
\toprule
 & \multicolumn{3}{c}{Low first stage} & \multicolumn{3}{c}{Baseline first stage} & \multicolumn{3}{c}{High first stage} & Intervals \\
\cmidrule(lr){2-4}\cmidrule(lr){5-7}\cmidrule(lr){8-10}
Category & Low & Base & High & Low & Base & High & Low & Base & High & excluding zero \\
\midrule
\multicolumn{11}{l}{\textit{Panel A: Log hourly wage}} \\
SME & -0.038 & -0.037 & -0.036 & -0.037 & -0.036 & -0.035 & -0.037 & -0.037 & -0.036 & 9/9 \\
Large firm & 0.007 & 0.008 & 0.008 & 0.008 & 0.008 & 0.008 & 0.007 & 0.007 & 0.007 & 0/9 \\
Non-STEM & -0.057 & -0.058 & -0.058 & -0.057 & -0.058 & -0.058 & -0.057 & -0.058 & -0.058 & 9/9 \\
STEM & -0.015 & -0.016 & -0.016 & -0.016 & -0.016 & -0.016 & -0.016 & -0.016 & -0.016 & 9/9 \\
Public & -0.012 & -0.011 & -0.010 & -0.012 & -0.011 & -0.010 & -0.011 & -0.010 & -0.009 & 0/9 \\
Private & -0.078 & -0.078 & -0.079 & -0.078 & -0.078 & -0.079 & -0.078 & -0.078 & -0.079 & 9/9 \\
\midrule
\multicolumn{11}{l}{\textit{Panel B: Log monthly wage}} \\
SME & -0.080 & -0.079 & -0.079 & -0.079 & -0.078 & -0.078 & -0.079 & -0.078 & -0.078 & 9/9 \\
Large firm & -0.024 & -0.023 & -0.023 & -0.023 & -0.023 & -0.023 & -0.025 & -0.024 & -0.024 & 9/9 \\
Non-STEM & -0.097 & -0.098 & -0.098 & -0.097 & -0.098 & -0.098 & -0.097 & -0.098 & -0.098 & 9/9 \\
STEM & -0.041 & -0.042 & -0.041 & -0.042 & -0.042 & -0.042 & -0.042 & -0.042 & -0.041 & 9/9 \\
Public & -0.051 & -0.050 & -0.050 & -0.051 & -0.050 & -0.050 & -0.050 & -0.050 & -0.049 & 9/9 \\
Private & -0.118 & -0.119 & -0.119 & -0.118 & -0.119 & -0.119 & -0.118 & -0.119 & -0.119 & 9/9 \\
\bottomrule
\end{tabular}
\smallskip
\parbox{\textwidth}{\footnotesize\textit{Notes:} First-stage dimensions (age, GPA, parental income; each tensor margin) are $(6,6,4;3)$, $(10,10,5;5)$, and $(14,14,6;6)$ for Low, Baseline, and High. Second-stage degrees of freedom for the marginal and tensor splines are 4/3, 6/4, and 8/5. The final column counts 95\% confidence intervals excluding zero, using analytic two-stage standard errors for each specification. All nine specifications use the same sample.}
\end{table}

\paragraph{Residual variation of the wage regressors.}
Table~\ref{tab: rank diagnostics} reports how much of the original variation in the female indicator remains after projection on the other wage covariates and the spline basis. Even in STEM, where the residual covariance matrix has the smallest eigenvalue, 44.7\% of that variation remains. 

\begin{table}[htbp]
\centering
\caption{Residual variation after projection on the spline basis}
\label{tab: rank diagnostics}
\footnotesize
\setlength{\tabcolsep}{3.5pt}
\begin{tabular}{lrrrrrrrrr}
\toprule
 &  & \multicolumn{2}{c}{Female residual (\%)} & \multicolumn{4}{c}{After spline terms (\%)} & \multicolumn{2}{c}{Residual covariance} \\
\cmidrule(lr){3-4}\cmidrule(lr){5-8}\cmidrule(lr){9-10}
 &  & Covariates & Plus &  &  & Parental &  & Smallest & Condition \\
Category & $N$ & only & spline terms & Age & GPA & income & Four-year & eigenvalue & number \\
\midrule
SME & 67,397 & 58.9 & 32.7 & 82.2 & 94.4 & 93.9 & 76.9 & 0.0278 & 62.5 \\
Large firm & 48,241 & 61.0 & 38.8 & 81.4 & 97.0 & 93.8 & 87.7 & 0.0518 & 35.7 \\
Non-STEM & 86,802 & 61.0 & 50.5 & 73.1 & 95.8 & 99.0 & 79.6 & 0.0186 & 96.6 \\
STEM & 34,156 & 56.9 & 44.7 & 67.0 & 93.2 & 98.5 & 70.1 & 0.0032 & 494.1 \\
Public & 25,861 & 70.1 & 53.0 & 80.3 & 89.8 & 99.0 & 51.6 & 0.0602 & 25.2 \\
Private & 94,042 & 56.1 & 52.8 & 78.0 & 90.4 & 98.4 & 42.4 & 0.0578 & 30.8 \\
\bottomrule
\end{tabular}
\smallskip
\parbox{\textwidth}{\footnotesize\textit{Notes:} Hourly-wage sample. The first two variance columns partial out the other wage regressors, first alone and then together with the spline terms. The next four columns partial out the spline terms and a constant. Variances are percentages of each regressor's original variance. Eigenvalues and condition numbers use the residual covariance matrix of all 43 standardized, nonconstant wage regressors after projection on the spline basis.}
\end{table}

\paragraph{Calibrated coverage.}
To assess the analytic 95\% intervals for the Sieve female coefficients, I draw covariate profiles from each empirical sample at its observed size and treat the estimated wage coefficients as known targets. Both simulated wage equations use the predetermined covariates. In the firm-size design, a common shock and the estimated gender-specific thresholds determine assignment to the outside category, an SME, or a large firm. In the field and sector designs, Gumbel shocks generate choices with the estimated gender-specific probabilities. In each design, the dependence of wages on the selection shocks is calibrated to the estimated selection correction function in the wage equation. I re-estimate both stages in each sample, including smoothing parameters and second-stage knots.

In simulations calibrated to the full 2008--2019 sample, the hourly wage intervals cover the target female coefficient in 93.6\% of replications for large firms and 94.6--96.6\% for the other five categories as reported in Table~\ref{tab: pooled analytic coverage}.

\begin{table}[htbp]
\centering
\caption{Coverage of 95\% confidence intervals for the Sieve female coefficients}
\label{tab: pooled analytic coverage}
\small
\begin{tabular*}{\textwidth}{@{\extracolsep{\fill}}lcccccc@{}}
\toprule
 & \multicolumn{2}{c}{Firm size} & \multicolumn{2}{c}{Field} & \multicolumn{2}{c}{Sector} \\
\cmidrule(lr){2-3}\cmidrule(lr){4-5}\cmidrule(lr){6-7}
Outcome & SME & Large firm & Non-STEM & STEM & Public & Private \\
\midrule
Log hourly wage & 95.8 & 93.6 & 95.6 & 96.6 & 96.2 & 94.6 \\
Log monthly wage & 97.0 & 93.6 & 94.4 & 95.8 & 96.2 & 94.6 \\
\bottomrule
\end{tabular*}
\par\smallskip
\parbox{\textwidth}{\footnotesize\textit{Notes:} Entries are the percentages of 500 replications in which the analytic interval contains the target coefficient. The monthly wage simulation omits the separate controls for regular and overtime hours used in the empirical regression.}
\end{table}

\subsection{Year-specific estimates}\label{sec:yearly_empirical}

I re-estimate the selection and wage equations for each graduating cohort. Figures~\ref{fig:yearly_hourly} and~\ref{fig:yearly_monthly} report the hourly and monthly wage results, respectively. In the hourly wage regressions, the Sieve estimate is negative in every cohort for SMEs, non-STEM jobs, and private employment; in large firms it falls on either side of zero.

\begin{figure}[p!]
\centering
\includegraphics[width=\textwidth]{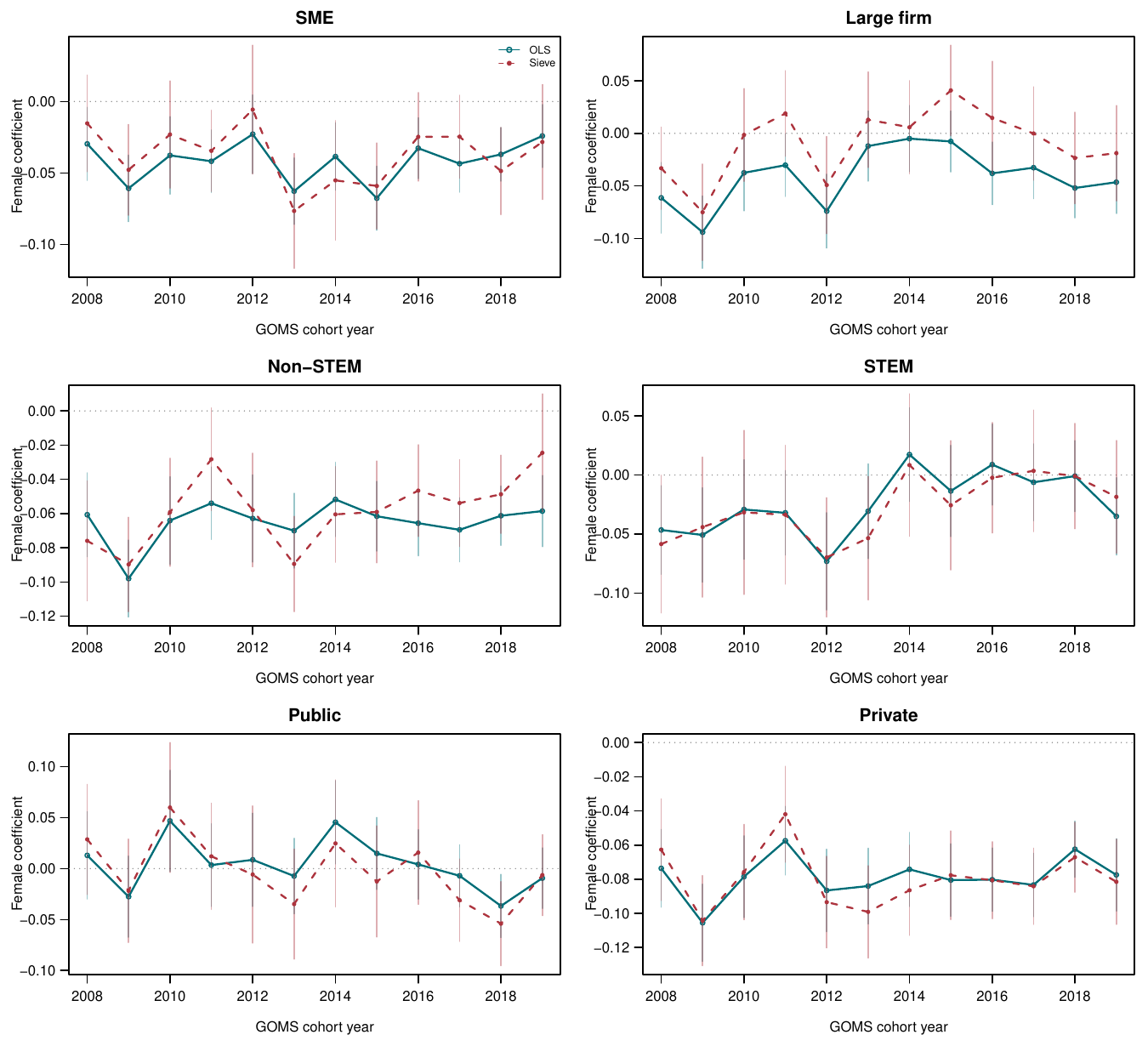}
\caption{OLS and Sieve estimates of the female coefficient in log hourly wages by cohort. Vertical bars show pointwise 95\% confidence intervals based on heteroskedasticity-robust standard errors for OLS and analytic two-stage standard errors for Sieve.}\label{fig:yearly_hourly}
\end{figure}

\begin{figure}[p!]
\centering
\includegraphics[width=\textwidth]{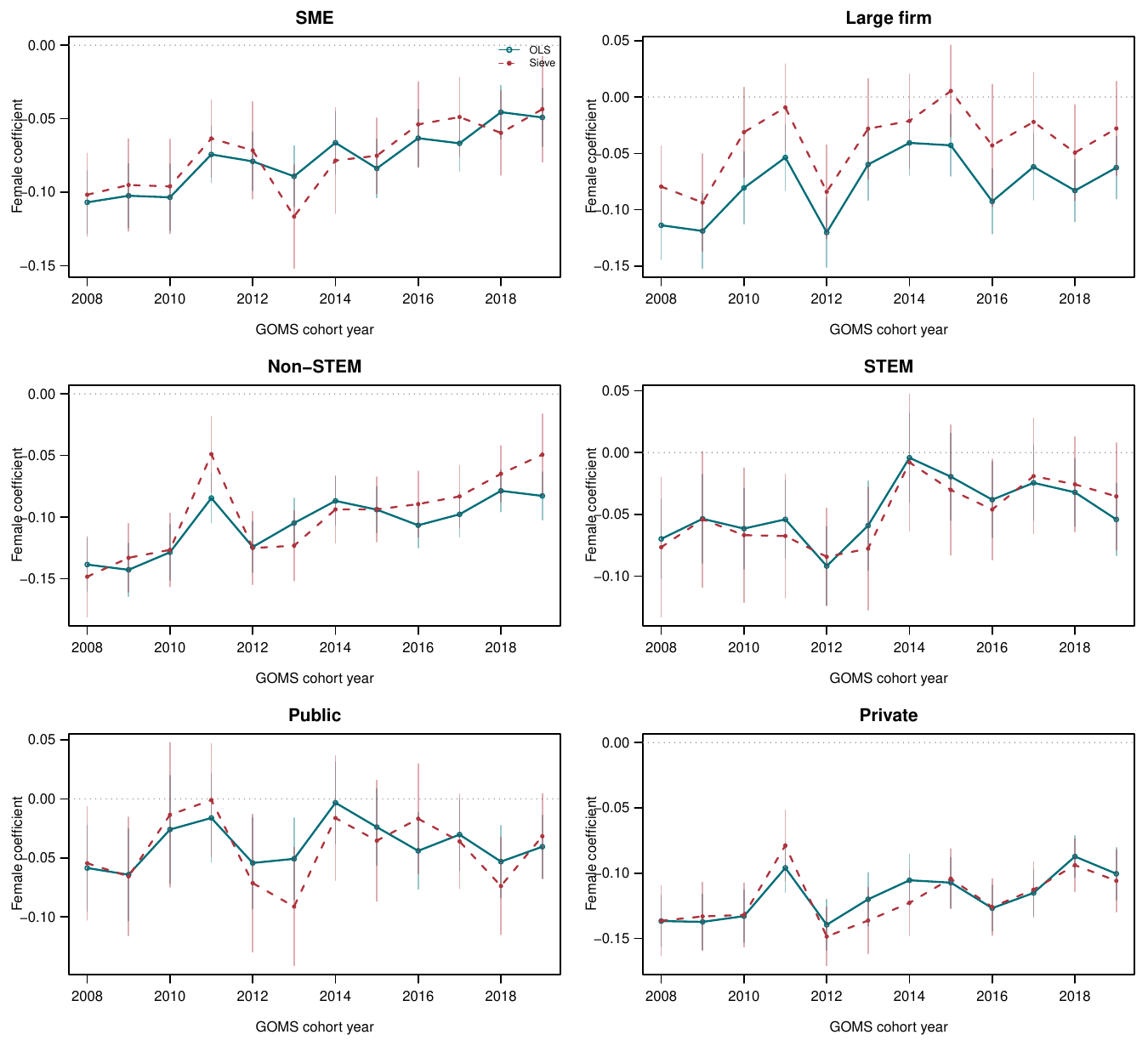}
\caption{OLS and Sieve estimates of the female coefficient in log monthly wages by cohort, controlling separately for regular and overtime hours per week. Vertical bars show pointwise 95\% confidence intervals.}\label{fig:yearly_monthly}
\end{figure}

\end{document}